\documentclass[a4paper,onecolumn,superscriptaddress,11pt,accepted=2020-06-30]{quantumarticle}
\pdfoutput=1
\usepackage[utf8]{inputenc}
\usepackage[english]{babel}
\usepackage[T1]{fontenc}
\usepackage{amsmath}
\usepackage{amssymb}
\usepackage{amsthm}
\usepackage{revsymb}
\usepackage{graphicx,color}
\usepackage{bm}
\usepackage{cite}
\usepackage{enumitem}   
\usepackage{hyperref}
\usepackage{tikz}
\usepackage{lipsum}

\newtheorem{thm}{Theorem}
\newtheorem{prop}{Proposition}
\newtheorem{lemma}{Lemma}
\newtheorem{corol}{Corollary}

\theoremstyle{remark}
\newtheorem{remark}{Remark} 

\theoremstyle{definition}

\numberwithin{equation}{section}

\newcommand{\tr}{\mathop{\mathrm{tr}}\nolimits}
\renewcommand{\Re}{\mathop{\mathrm{Re}}\nolimits}
\renewcommand{\Im}{\mathop{\mathrm{Im}}\nolimits}
\newcommand{\ad}{\mathop{\mathrm{ad}}\nolimits}

\newcommand{\rank}{\mathop{\mathrm{rank}}\nolimits}
\newcommand{\ran}{\mathop{\mathrm{ran}}\nolimits}

\renewcommand{\d}{\mathrm{d}}
\newcommand{\rmd}{\mathrm{d}}
\newcommand{\rmi}{\mathrm{i}}
\newcommand{\rme}{\mathrm{e}}
\newcommand{\ket}[1]{|{#1}\rangle}
\newcommand{\bra}[1]{\langle{#1}|}

\renewcommand{\bullet}{\,\,\begin{picture}(-1,1)(-1,-3)\circle*{2}\end{picture}\ \,\,}

\definecolor{dgreen}{rgb}{0,0.5,0}

\definecolor{delete}{cmyk}{0.5,0,0,0}

\begin{document}
\title[Quantum Zeno Dynamics from General Quantum Operations]{Quantum Zeno Dynamics from General Quantum Operations}
\date{June 30, 2020}
\author{Daniel Burgarth}
\orcid{0000-0003-4063-1264}
\affiliation{Center for Engineered Quantum Systems, Dept.\ of Physics \& Astronomy, Macquarie University, 2109 NSW, Australia}
\author{Paolo Facchi}
\orcid{0000-0001-9152-6515}
\affiliation{Dipartimento di Fisica and MECENAS, Universit\`a di Bari, I-70126 Bari, Italy}
\affiliation{INFN, Sezione di Bari, I-70126 Bari, Italy}
\author{Hiromichi Nakazato}
\orcid{0000-0002-5257-7309}
\affiliation{Department of Physics, Waseda University, Tokyo 169-8555, Japan}
\author{Saverio Pascazio}
\orcid{0000-0002-7214-5685}
\affiliation{Dipartimento di Fisica and MECENAS, Universit\`a di Bari, I-70126 Bari, Italy}
\affiliation{INFN, Sezione di Bari, I-70126 Bari, Italy}
\author{Kazuya Yuasa}
\orcid{0000-0001-5314-2780}
\affiliation{Department of Physics, Waseda University, Tokyo 169-8555, Japan}
\maketitle
\begin{abstract}
We consider the evolution of an arbitrary quantum dynamical semigroup of a finite-dimensional quantum system under frequent kicks, where each kick is a generic quantum operation. We develop a generalization of the Baker-Campbell-Hausdorff formula allowing to reformulate such pulsed dynamics as a continuous one. This reveals an adiabatic evolution. We obtain a general type of quantum Zeno dynamics, which unifies all known manifestations in the literature as well as describing new types.
\end{abstract}


\section{Introduction}
Physics is a science that is often based on approximations. From high-energy physics to the quantum world, from relativity to thermodynamics, approximations not only help us to solve equations of motion, but also to reduce the model complexity and focus on important effects. Among the largest success stories of such approximations are the effective generators of dynamics (Hamiltonians, Lindbladians), which can be derived in quantum mechanics and condensed-matter physics. The key element in the techniques employed for their derivation is the separation of different time scales or energy scales.

Recently, in quantum technology, a more active approach to condensed-matter physics and quantum mechanics has been taken. Generators of dynamics are \emph{reversely engineered} by tuning system parameters and device design. This allows the creation of effective generators useful for many information-theoretic tasks, such as adiabatic quantum computing \cite{ref:AdiabaticQuantComp-RMP}, reservoir engineering \cite{ref:ReservoirEngineeringZoller}, quantum gates \cite{ref:GoogleQuantumSupremacy}, to name a few.

A key player for such approximations has been the adiabatic theorem \cite{Messiah,ref:KatoAdiabatic}. It exploits a clear separation of slow and fast time scales and has fascinated generations of physicists due to its simplicity, its beauty, and its intriguing geometric interpretations. 
In its original formulation, the adiabatic theorem deals with generators of dynamics. On the other hand, in quantum technology, we often deal with \textit{discrete} dynamics such as fixed gates and quantum maps. It is not always straightforward, and sometimes seemingly impossible, to translate between continuous and discrete descriptions. This difficulty can be seen more clearly in the case of non-Markovian quantum channels: these are physical operations [completely positive and trace-preserving (CPTP) maps \cite{ref:NielsenChuang10th}] for which there are no physical (e.g.~Lindbladian) generators [a non-Markovian quantum channel cannot be realized by the time-ordered integral of an infinitesimal CPTP generator of the Gorini-Kossakowski-Lindblad-Sudarshan (GKLS) form \cite{ref:DynamicalMap-Alicki,ref:GKLS-DariuszSaverio}]. Such maps typically arise from unitary dynamics on larger (but finite-dimensional) spaces, such as two-level fluctuators forming the environment of solid-state qubits. This makes them rather common in experiments \cite{ref:GuideSupercondQubits}.

The key question we pose in this article is if discrete dynamics can give rise to a limit evolution? We provide a positive answer to this question by providing a general mapping from pulsed to continuous dynamics. Such connections had been noted before only for a specific unitary case. For the generic situation, we need to develop a more powerful framework. This is because in the nonunitary case one has to take into account both nondiagonalizability and noninvertibility of the maps.

We focus on finite-dimensional systems and provide a route to connecting strong-coupling limits and frequent pulsed dynamics using three key ingredients. The first is a slight generalization of the Baker-Campbell-Hausdorff (BCH) theorem \cite{ref:Hall-BCH} (Lemma~\ref{thm:kicktofield} in Sec.~\ref{sec:KeyLemmas}). We employ operator logarithms to describe discrete maps by (potentially unphysical) generators and unify the product of exponentials as a new exponential up to the first order in a sense explained later. This only works for invertible maps (there is no logarithm of a noninvertible map since the logarithm of the vanishing eigenvalue does not exist), so the second ingredient is a delicate error estimate allowing us to take logarithms of only the invertible part of a generic map: the remaining part decays anyway (Lemma~\ref{prop:Peripheral} in Sec.~\ref{sec:KeyLemmas}). We believe that these key lemmas might find applications in other areas of quantum technology. The third and ultimate ingredient is a strong-coupling theorem developed by the present authors recently in Ref.\ \cite{ref:unity1} (see also Ref.\ \cite{ref:unity1-view}), which can be applied to unphysical generators (Theorem 1 of Ref.\ \cite{ref:unity1}). Even though the logarithm of a physical operation is not of the GKLS form in general, one can deal with it by the adiabatic theorem proved in Ref.\ \cite{ref:unity1}. Through the adiabatic theorem, low-energy components are eliminated and the physicality of the generator is restored. Our approach therefore works for arbitrary quantum maps without unnecessary structural assumptions.

Our generalization goes in two main directions: 1) the unitary dynamics  $\rme^{-\rmi t H}$ is generalized to an arbitrary quantum semigroup $\rme^{t \mathcal{L}}$; 2) the projective measurement $P$ is generalized to an arbitrary quantum operation $\mathcal{E}$. Moreover, we also generalize to 3) kicked dynamics of cycles of quantum operations $\{\mathcal{E}_1,\ldots,\mathcal{E}_m\}$.
This unifies many applications, such as the quantum Zeno effect (QZE) and dynamical decoupling, 
and provides a deeper relationship through adiabaticity. 
The main result of the present work is Theorem~\ref{thm:CPTPBB} in Sec.~\ref{sec:MainTheorem}.
It reveals that in quantum technology one has more freedom than previously thought to achieve effective generators (see the next section for a brief summary of the previous results).
In addition, we derive explicit bounds on matrix functions and the BCH formula for nondiagonalizable matrices (Lemma~\ref{lem:MatrixFuncBound} and Proposition~\ref{prop:BoundBCH} in Appendix~\ref{app:BCH}), which may be of interest independently.

\section{Relation to Previous Work}
The type of dynamics encountered in our main theorem can be considered as a general type of quantum Zeno dynamics (QZD). Different manifestations of QZD are known \cite{ref:ControlDecoZeno,ref:PaoloSaverio-QZEreview-JPA}, via (i) frequent projective measurements \cite{ref:QZEMisraSudarshan,ref:InverseZeno}, via (ii) frequent unitary kicks \cite{ref:BBZeno,ref:Bernad}, via (iii) strong continuous coupling/fast oscillations \cite{ref:SchulmanJumpTimePRA,ref:QZS}, and via (iv) strong damping \cite{ref:NoiseInducedZeno,ref:OreshkovAMS-PRL2010,ref:ZanardiDFS-PRL2014,ref:ZanardiDFS-PRA2015,ref:ZanardiDFS-PRA2017} (see Refs.\ \cite{ref:QZEExp-Ketterle,ref:QZEExp-Schafer:2014aa} for experimental comparisons).
Dynamical decoupling \cite{ref:DynamicalDecoupling-ViolaLloydPRA1998,ref:DynamicalDecoupling-ViolaKnillLloyd-PRL1999,ref:BangBang-VitaliTombesi-PRA1999,ref:DynamicalDecoupling-Zanardi-PLA1999,ref:BangBang-Duan-PLA1999,ref:DynamicalDecoupling-Viola-PRA2002,ref:BangBang-Uchiyama-PRA2002} (see also Ref.\ \cite{ref:MultiProjections-Ticozzi}) is also regarded as a manifestation of the QZD\@.
See Fig.~\ref{fig:QZDs} for a summary of these different manifestations of the QZD\@.
\begin{figure}[h]
\centering
\includegraphics[scale=0.36]{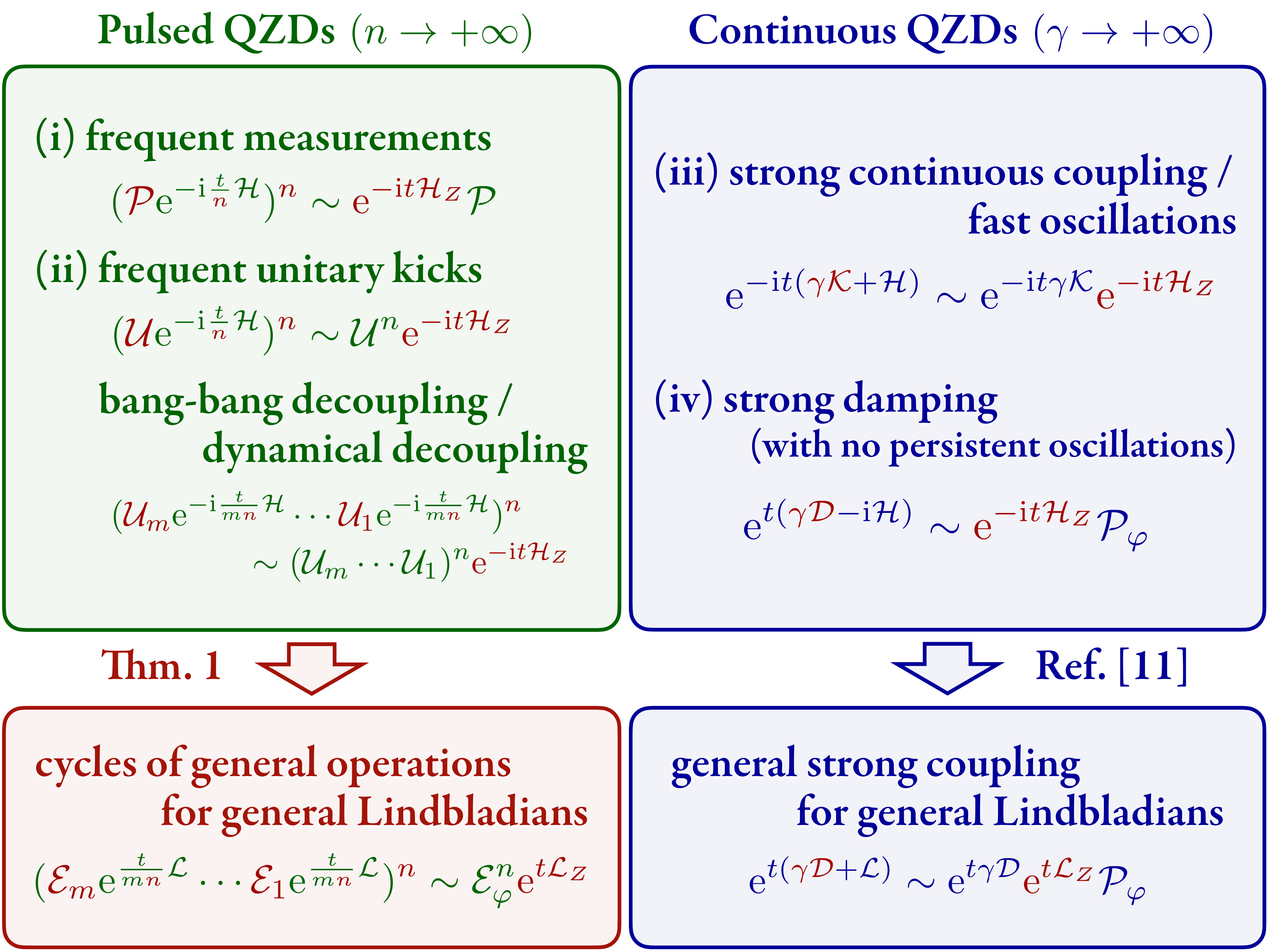}
\caption{
Different manifestations of QZD \cite{ref:ControlDecoZeno,ref:PaoloSaverio-QZEreview-JPA} (experimental comparisons are found in Refs.\ \cite{ref:QZEExp-Ketterle,ref:QZEExp-Schafer:2014aa}). The standard way to induce the QZD is to (i) frequently perform projective measurements, each represented by a projection $\mathcal{P}$ \cite{ref:QZEMisraSudarshan,ref:InverseZeno}. 
It can also be induced via (ii) frequent unitary kicks with an instantaneous unitary $\mathcal{U}$ \cite{ref:BBZeno,ref:Bernad}, or via cycles of a bang-bang sequence of multiple unitaries $\mathcal{U}_j$ ($j=1,\ldots,m$) \cite{ref:DynamicalDecoupling-ViolaLloydPRA1998,ref:DynamicalDecoupling-ViolaKnillLloyd-PRL1999,ref:BangBang-VitaliTombesi-PRA1999,ref:DynamicalDecoupling-Zanardi-PLA1999,ref:BangBang-Duan-PLA1999,ref:DynamicalDecoupling-Viola-PRA2002,ref:BangBang-Uchiyama-PRA2002}. In contrast to these pulsed strategies, the QZD can be induced by (iii) continuously applying a strong external field represented by a unitary generator $\mathcal{K}$ \cite{ref:SchulmanJumpTimePRA,ref:QZS}, or by (iv) putting the system under strong damping $\mathcal{D}$ relaxing the system to a steady subspace as $\rme^{t\mathcal{D}}\to \mathcal{P}_\varphi$
as $t\to+\infty$ \cite{ref:NoiseInducedZeno,ref:OreshkovAMS-PRL2010,ref:ZanardiDFS-PRL2014,ref:ZanardiDFS-PRA2015,Madalin,ref:VictorPRX,ref:ZanardiDFS-PRA2017}. In any case, a unitary generator $\mathcal{H}$ is projected to a \textit{Zeno generator} $\mathcal{H}_Z$. See the text concerning how it is actually projected. In Ref.\ \cite{ref:unity1}, we unified the continuous strategies and generalized them for general Lindbladians $\mathcal{D}$ and $\mathcal{L}$ on the basis of an adiabatic theorem.
 In this paper, we focus on the pulsed strategies, which we unify and generalize for general quantum operations $\mathcal{E}_j$ ($j=1,\ldots,m$) and general Lindbladian $\mathcal{L}$ in Theorem \ref{thm:CPTPBB}.}
\label{fig:QZDs}
\end{figure}

Since pulsed strategies will be the main subject of this article, it is convenient to recapitulate their main features.
\paragraph{(i) Frequent projective measurements:} 
The standard way to induce the QZD is to perform projective measurements frequently \cite{ref:QZS,ref:artzeno}.
Consider a quantum system on a finite-dimensional Hilbert space with a Hamiltonian $H=H^\dagger$.
During the unitary evolution $\rme^{-\rmi t\mathcal{H}}$ for time $t$ with $\mathcal{H}=[H,{}\bullet{}]$, we perform projective measurement $n$ times at regular time intervals.
The measurement is represented by a set of orthogonal projection operators $\{P_k\}$ acting on the Hilbert space, satisfying $P_kP_\ell=P_k\delta_{k\ell}$ and $\sum_kP_k=I$.
We retain no outcome of the measurement, i.e., our measurement is a nonselective one,  described by the projection $\mathcal{P}$ acting on a density operator $\rho$ as
\begin{equation}
\mathcal{P}(\rho)=\sum_kP_k\rho P_k.
\end{equation}
In the limit of infinitely frequent measurements (\textit{Zeno limit}), the evolution of the system is described by \cite{ref:PaoloSaverio-QZEreview-JPA}
\begin{equation}
  \left(\mathcal{P}\rme^{-\rmi \frac{t}{n}\mathcal{H}}\right)^n
  =\rme^{-\rmi t\mathcal{H}_Z}\mathcal{P}
  +\mathcal{O}(1/n)
  \quad \text{as} \quad n\to +\infty,
  \label{eqn:ZenoLimit}
\end{equation}
where
$\mathcal{H}_Z=[H_Z,{}\bullet{}]$ with
\begin{equation}
	H_Z=\sum_kP_kHP_k.
	\label{eqn:ZenoHamiltonian}
\end{equation}
In the Zeno limit, the transitions among the subspaces specified by the projection operators $\{P_k\}$ are suppressed.
This is the QZE\@.
The system meanwhile evolves unitarily within the subspaces (\textit{Zeno subspaces} \cite{ref:QZS}) with the projected Hamiltonian (\textit{Zeno Hamiltonian}) $H_Z$.
This is the QZD\@.

\paragraph{(ii) Frequent unitary kicks:} Instead of measurement, we can apply a series of \textit{unitary kicks} to induce the QZD \cite{ref:BBZeno,ref:Bernad}.
During the unitary evolution $\rme^{-\rmi t \mathcal{H}}$ for time $t$, we apply an instantaneous unitary transformation $\mathcal{U} = U \bullet U^\dagger$, with $U^\dagger = U^{-1}$, repeatedly $n$ times at regular time intervals $t/n$.
In the limit of infinitely frequent unitary kicks, the evolution of the system is described by \cite{ref:BBZeno}
\begin{equation}
  \left(\mathcal{U}\rme^{-\rmi \frac{t}{n} \mathcal{H}}\right)^n
  =\mathcal{U}^n\rme^{-\rmi t \mathcal{H}_Z}
  +\mathcal{O}(1/n)
  \quad \text{as} \quad n\to +\infty,
  \label{eqn:UnitaryKicks}
\end{equation}
where $\mathcal{H}_Z=[H_Z,{}\bullet{}]$ with $H_Z$ defined by Eq.\ (\ref{eqn:ZenoHamiltonian}) with the eigenprojections $\{P_k\}$ of the spectral representation of the unitary $U=\sum_k\rme^{-\rmi\eta_k}P_k$.\footnote{In this unitary case, one can think of this Zeno limit also in the Hilbert space, namely, we can study the limit of $(U\rme^{-\rmi\frac{t}{n}H})^n$ instead of $(\mathcal{U}\rme^{-\rmi\frac{t}{n}\mathcal{H}})^n$. The structure of the Zeno generator $\mathcal{H}_Z=\sum_k\mathcal{P}_k\mathcal{H}\mathcal{P}_k$, where $\{\mathcal{P}_k\}$ are the spectral projections of $\mathcal{U}$, is inherited by the structure of the Zeno Hamiltonian $H_Z=\sum_kP_kHP_k$. A proof of the equivalence of the two formulations can be found in Ref.\ \cite[Eq.\ (3.21)]{ref:unity1}.}
In this way, the frequent unitary kicks project the Hamiltonian $H$ in essentially the same way as the frequent projective measurements do, and the QZD is induced.

Instead of repeating the same unitary $U$, one can think of applying cycles of different unitaries $U_j$ ($j=1,\ldots,m$) as \cite{ref:BBZeno}
\begin{equation}
  \left(\mathcal{U}_m\rme^{-\rmi\frac{t}{mn}\mathcal{H}}\cdots\mathcal{U}_1\rme^{-\rmi\frac{t}{mn}\mathcal{H}}\right)^n
  =(\mathcal{U}_m\cdots\mathcal{U}_1)^n\rme^{-\rmi t \mathcal{H}_Z}
  +\mathcal{O}(1/n)
  \quad \text{as} \quad n\to +\infty.
  \label{eqn:UnitaryKicksBB}
\end{equation}
In this case, the Hamiltonian $H$ is projected to $H_Z$ as
\begin{equation}
H_Z=\sum_kP_k\overline{H}P_k,\qquad
\overline{H}=\frac{1}{m}\,\biggl(
H+\sum_{j=2}^mU_1^\dag\cdots U_{j-1}^\dag HU_{j-1}\cdots U_1
\biggr),
\end{equation}
where $\{P_k\}$ are the spectral projections of the product of the unitaries $U_m\cdots U_1=\sum_k\rme^{-\rmi\eta_k}P_k$.
Such schemes have been eagerly studied as methods of decoupling the dynamics of a system from an environment, and are called bang-bang control or dynamical decoupling \cite{ref:DynamicalDecoupling-ViolaLloydPRA1998,ref:DynamicalDecoupling-ViolaKnillLloyd-PRL1999,ref:BangBang-VitaliTombesi-PRA1999,ref:DynamicalDecoupling-Zanardi-PLA1999,ref:BangBang-Duan-PLA1999,ref:DynamicalDecoupling-Viola-PRA2002,ref:BangBang-Uchiyama-PRA2002}.
Roughly speaking, the idea is to rotate the Hamiltonian describing the interaction with an environment around all possible directions to average it out.
This can be regarded as a manifestation of the QZD\@.

\bigskip\smallskip
In a previous article \cite{ref:unity1}, we have unified and generalized the continuous strategies (iii) and (iv), by proving (Theorem~2 of Ref.\ \cite{ref:unity1}),
for arbitrary Markovian generators $\mathcal{L}$ and $\mathcal{D}$ of
the GKLS form \cite{ref:DynamicalMap-Alicki,ref:GKLS-DariuszSaverio},
\begin{equation}
\rme^{t(\gamma\mathcal{D}+\mathcal{L})}=\rme^{t \gamma \mathcal{D}}\rme^{t\mathcal{L}_Z}\mathcal{P}_\varphi+\mathcal{O}(1/\gamma)\quad\text{as}\quad\gamma\to +\infty
\label{eqn:Unity1}
\end{equation}
with
\begin{equation}
\mathcal{L}_Z
=\sum_{\alpha_k\in \rmi\mathbb{R}}\mathcal{P}_k\mathcal{L}\mathcal{P}_k,
\end{equation}
where $\mathcal{P}_k$ is the spectral projection onto the eigenspace of $\mathcal{D}$ belonging to the eigenvalue $\alpha_k$ 
and $\mathcal{P}_\varphi = \sum_{\alpha_k\in \rmi\mathbb{R}} \mathcal{P}_k$ is the projection onto the peripheral spectrum of $\mathcal{D}$ (i.e.\ its purely imaginary eigenvalues).

\paragraph{(iii) Strong continuous coupling/fast oscillations:}
If $\mathcal{D}$ and $\mathcal{L}$ are both unitary generators $\mathcal{D}=-\rmi\mathcal{K}=-\rmi[K,{}\bullet{}]$ and $\mathcal{L}=-\rmi\mathcal{H}=-\rmi[H,{}\bullet{}]$ with some Hamiltonians $K$ and $H$, the theorem (\ref{eqn:Unity1}) is reduced to 
\begin{equation}
\rme^{-\rmi t(\gamma\mathcal{K}+\mathcal{H})}=\rme^{-\rmi t \gamma \mathcal{K}}\rme^{-\rmi t\mathcal{H}_Z}+\mathcal{O}(1/\gamma)\quad\text{as}\quad\gamma\to +\infty,
\end{equation}
where $\mathcal{H}_Z=[H_Z,{}\bullet{}]$ with $H_Z$ again given by Eq.\ (\ref{eqn:ZenoHamiltonian}) with the spectral projections $\{P_k\}$ of $K$ \cite{ref:SchulmanJumpTimePRA,ref:QZS,alphaTrot}.

\paragraph{(iv) Strong damping (with no persistent oscillations):}
If $\mathcal{D}$ is a generator describing an amplitude-damping to a stationary subspace specified by a projection $\mathcal{P}_\varphi$ with no persistent oscillations there, i.e.\ if $\rme^{t\mathcal{D}}\to\mathcal{P}_\varphi$ as $t\to+\infty$, then the theorem (\ref{eqn:Unity1}) reproduces 
\begin{equation}
\rme^{t(\gamma\mathcal{D}-\rmi\mathcal{H})}
=\rme^{-\rmi t\mathcal{H}_Z}\mathcal{P}_\varphi+\mathcal{O}(1/\gamma)\quad\text{as}\quad\gamma\to +\infty,
\end{equation}
where $\mathcal{H}_Z=\mathcal{P}_\varphi\mathcal{H}\mathcal{P}_\varphi$ \cite{ref:NoiseInducedZeno,ref:OreshkovAMS-PRL2010,ref:ZanardiDFS-PRL2014,ref:ZanardiDFS-PRA2015,Madalin,ref:VictorPRX,ref:ZanardiDFS-PRA2017}.

\bigskip\smallskip
In this way,  theorem (\ref{eqn:Unity1}) includes both previously known QZDs (iii) and (iv), and both mechanisms can be effective simultaneously.
It  also generalizes the QZDs (iii) and (iv) to nonunitary (Markovian) evolutions, projecting generic GKLS generators instead of Hamiltonians.
This theorem has been proved by the generalized adiabatic theorem (Theorem~1 of Ref.\ \cite{ref:unity1}), which is an extension of the adiabatic theorem proved by Kato \cite{ref:KatoAdiabatic} for unitary evolution.

With our main Theorem~\ref{thm:CPTPBB} below, we further unify the pulsed strategies (i) and (ii) with the continuous ones (iii) and (iv).
This also enables us to unify and generalize bang-bang decoupling/dynamical decoupling to those by cycles of multiple quantum operations (Theorem~\ref{thm:CPTPBB}), including non-Markovian (indivisible) ones \cite{wolfnonmarkovian}: structural assumptions for the pulses (kicks) are relaxed in our main Theorem~\ref{thm:CPTPBB}\@.
As an interesting variant of it, we present the QZD via cycles of different selective projective measurements (Corollary~\ref{thm:MeasBB}), generalizing the standard QZD via (i) frequent selective projective measurements.
We shall look at some simple examples and show that realistic unsharp (nonprojective) measurements can be practically more efficient to induce the QZD than the strong (projective) measurements (Sec.~\ref{sec:Examples}). This generalizes Refs.\ \cite{ref:LidarZenoPRL,ref:LidarZenoJPA} for the QZD by a particular type of weak measurement.
The generalization to the QZD via general quantum operations was also explored in Ref.\ \cite{ref:QZEbyQuantOpKwek}, where it is required that the generator be a Hamiltonian, and a single kick is repeated; most importantly, the initial state must be an invariant state of the kick.
Our Theorem~\ref{thm:CPTPBB} does not require such structural assumptions: it deals with cycles of multiple general quantum operations, and works for generic Markovian dynamics with a GKLS generator for arbitrary initial conditions.

\section{Some Preliminaries on Quantum Operations}
Let us recall that every linear operator $A$ on a finite-dimensional space can be expressed (essentially) uniquely in terms of its Jordan normal form (\textit{canonical form} or \textit{spectral representation}) \cite{ref:Katobook}:
\begin{equation}
\label{eq:Jordanform}
A = \sum_k (\lambda_k P_k + N_k),
\end{equation}
where $\{\lambda_k\}$, the \textit{spectrum} of $A$, is the set of distinct eigenvalues of $A$ ($\lambda_k\neq \lambda_\ell$ for $k\neq \ell$), $\{P_k\}$, the \textit{spectral projections} of $A$, are the corresponding eigenprojections, satisfying
\begin{equation}
\label{eq:PVM}
P_k P_\ell = \delta_{k\ell} P_k, \qquad \sum_k P_k = I,
\end{equation}
for all $k$ and $\ell$, and $\{N_k\}$ are the corresponding \textit{nilpotents} of $A$, satisfying for all $k$ and $\ell$
\begin{equation}
\label{eq:nilpot}
P_k N_\ell = N_{\ell} P_k = \delta_{k\ell} N_k, \qquad N_k^{n_k} = 0,
\end{equation}
for some integer $1\leq n_k\leq \rank P_k$.

Notice that the spectral projections, which determine a partition of the space through the \textit{resolution of identity}~\eqref{eq:PVM}, are not Hermitian in general, $P_k \neq P_k^\dagger$. An eigenvalue $\lambda_k$ of $A$ is called \textit{semisimple} or \textit{diagonalizable} if the corresponding nilpotent $N_k$ is zero (equivalently, $n_k=1$). The operator $A$ is diagonalizable if and only if all its eigenvalues are semisimple.

The main actors in our investigation are the \textit{quantum operations} \cite{ref:NielsenChuang10th}, that is, maps $\mathcal{E}$ that are completely positive (CP) and trace-nonincreasing, $\tr [\mathcal{E}{(\rho)}]\leq\tr\rho$. We recall that a map $\mathcal{E}$ on a $d$-dimensional quantum system is a quantum operation  iff it has an operator-sum (Kraus) representation of the form
\begin{equation}
\label{eq:Kraussrep}
\mathcal{E} (X) = \sum_{j=1}^m K_j X K_j^\dagger\quad \text{with}\quad \sum_{j =1}^m K_j^\dagger K_j \leq I,
\end{equation}
where $m\leq d^2$. When $\sum_j K_j^\dagger K_j = I$, the map is trace-preserving (TP) and one gets a completely positive trace-preserving (CPTP) map, also known as a \textit{quantum channel}.

In the following, it will be convenient to endow the space of operators on a $d$-dimensional Hilbert space with the Hilbert-Schmidt inner product $\langle A | B \rangle_{2} = \tr (A^\dagger B)$, which makes the space of operators a $d^2$-dimensional Hilbert space~$\mathcal{T}_2$. 
We get that the adjoint $\mathcal{E}^\dagger$ [with respect to the Hilbert-Schmidt inner product, defined through $\langle A|\mathcal{E}(B)\rangle_{2} =\langle\mathcal{E}^\dag(A)|B\rangle_{2} $ for all $A, B\in \mathcal{T}_2$] of the quantum operation $\mathcal{E}$ in Eq.~\eqref{eq:Kraussrep} has the operator sum 
\begin{equation}
\mathcal{E}^\dagger (X) = \sum_j K_j^\dagger X K_j,
\end{equation}
and thus is subunital $\mathcal{E}^\dagger (I) \leq I$.
It is unital iff  $\mathcal{E}$ is CPTP\@.
Moreover, given a Hermitian operator $H=H^\dagger$, then the corresponding superoperator $\mathcal{H}  
= [H,{}\bullet{}]$ is also Hermitian (with respect to the Hilbert-Schmidt inner product), 
$\mathcal{H}=\mathcal{H}^\dagger$, namely $\langle A | \mathcal{H}(B) \rangle_{2}  = \langle \mathcal{H}(A) | B \rangle_{2} $ for all $A, B\in \mathcal{T}_2$. As a consequence, the unitary group $t\mapsto \rme^{-\rmi t H}$ is lifted to a unitary group $t\mapsto \rme^{-\rmi t \mathcal{H}} = \rme^{-\rmi t H} \bullet \rme^{\rmi t H}$. 
Finally, a (Hermitian) projection $P=P^2\,(=P^\dagger)$ is lifted to a (Hermitian) projection $\mathcal{P} = P\bullet P$, satisfying $\mathcal{P}=\mathcal{P}^2\,(=\mathcal{P}^\dagger)$.

We mainly use the operator norm (2--2 norm), when we need to specify the norm of a map $\mathcal{A}: \mathcal{T}_2 \to \mathcal{T}_2$, 
\begin{equation}
\|\mathcal{A} \| = \sup_{\|X\|_2 = 1} \| \mathcal{A} (X)\|_2,
\label{eqn:OpNorm}
\end{equation}
where $\|X\|_2=(\langle X|X\rangle_{2})^{1/2}$ for $X\in\mathcal{T}_2$.
It coincides with the largest singular value of $\mathcal{A}$, 
	\begin{equation}
		\|\mathcal{A}\|
		= \sup_{\|X\|_2 = 1}(\langle \mathcal{A}(X)| \mathcal{A}(X)\rangle_{2})^{1/2}
		= \sup_{\|X\|_2 = 1}(\langle X|(\mathcal{A}^\dag\mathcal{A})(X)\rangle_{2})^{1/2}
		=r(\mathcal
	{A}^\dag\mathcal{A})^{1/2} ,
\end{equation}
with $r(\mathcal{A}^\dag\mathcal{A})$ being the spectral radius of $\mathcal{A}^\dag\mathcal{A}$.

We now state without proofs some useful spectral properties of the quantum operations. For further details and proofs see e.g.\ Refs.\ \cite{ref:Mixing-Wolf,ref:TextbookWatrous}, and in particular Propositions~6.1--6.3 and Theorem~6.1 of Ref.\ \cite{ref:Mixing-Wolf}.
\begin{prop}[Spectral properties of quantum operations]
\label{prop:CPTP}
Let $\mathcal{E}$ be a quantum operation on a finite-dimensional space. Then, the following properties hold:
\begin{enumerate}[label=(\roman*)]
\item The spectrum $\{\lambda_k\}$ of $\mathcal{E}$ is confined in the closed unit disc $\mathbb{D}=\{\lambda \in \mathbb{C}, |\lambda|\leq 1\}$. Moreover, all the ``peripheral'' eigenvalues, belonging to the boundary of $\mathbb{D}$, i.e.~on the unit circle $\partial\mathbb{D}=\{\lambda \in \mathbb{C}, |\lambda|= 1\}$, are semisimple. If $\mathcal{E}$ is TP, then $\lambda=1$ is an eigenvalue of $\mathcal{E}$.

\item The canonical form of $\mathcal{E}$ reads
\begin{equation}
	\mathcal{E}
	=\mathcal{E}_\varphi +\sum_{|\lambda_k|<1} (\lambda_k \mathcal{P}_k+\mathcal{N}_k),
	\label{eqn:SpectralDecompE}
\end{equation}
where  $\{\mathcal{P}_k\}$ and $\{\mathcal{N}_k\}$ are the spectral projections and the nilpotents of $\mathcal{E}$, respectively, and
\begin{equation}
	\mathcal{E}_\varphi=\sum_{|\lambda_k|=1}\lambda_k\mathcal{P}_k
\end{equation} 
is the ``peripheral'' part of $\mathcal{E}$, i.e.\ its component  belonging to the peripheral spectrum on the unit circle $\partial\mathbb{D}$.

\item
The peripheral part $\mathcal{E}_\varphi$ and the projection onto the peripheral spectrum of $\mathcal{E}$,
\begin{equation}
	\mathcal{P}_\varphi=\sum_{|\lambda_k|=1}\mathcal{P}_k,
\end{equation}
are both quantum operations, and $\mathcal{E}_\varphi=\mathcal{E}\mathcal{P}_\varphi=\mathcal{P}_\varphi \mathcal{E}$. The maps  $\mathcal{E}_\varphi$ and $\mathcal{P}_\varphi$ are TP iff $\mathcal{E}$ is TP\@.

\item
The inverse  of $\mathcal{E}_\varphi$ on the range of $\mathcal{P}_\varphi$,
\begin{equation}
\mathcal{E}_\varphi^{-1} =  \sum_{|\lambda_k|=1}\lambda_k^{-1} \mathcal{P}_k,
\end{equation}
satisfying
$\mathcal{E}_\varphi^{-1}\mathcal{E}_\varphi
	=\mathcal{E}_\varphi\mathcal{E}_\varphi^{-1}
	=\mathcal{P}_\varphi$,
is also a quantum operation, and it is TP if $\mathcal{E}$ is TP\@.
\end{enumerate}
\end{prop}
Similar properties hold for GKLS generators $\mathcal{L}$, whose exponential $\rme^{t\mathcal{L}}$ is CPTP for all $t\geq 0$. See Proposition~2 of Ref.\ \cite{ref:unity1}.
\begin{remark}
Note that if the peripheral spectrum is empty then all peripheral maps are null, $\mathcal{P}_\varphi= \mathcal{E}_\varphi = \mathcal{E}_\varphi^{-1} = 0$. By property (i), this cannot happen if $\mathcal{E}$ is TP\@.
\end{remark}

\section{Main Theorem}
\label{sec:MainTheorem}
The main result of this paper is the unification of the pulsed QZDs, via (i) frequent projective measurements and via (ii) frequent unitary kicks, which at the same time allows us to generalize the pulses to arbitrary quantum operations.
We further generalize the bang-bang decoupling/dynamical decoupling to cycles of generic kicks.
These are all summarized in the following theorem, which will be proved in Sec.~\ref{sec:ProofTheorems}:
\begin{thm}[QZD by cycles of generic kicks]\label{thm:CPTPBB}
Let $\{\mathcal{E}_1,\ldots,\mathcal{E}_m\}$ be a finite set of quantum operations and $\mathcal{L}$ be a GKLS generator acting on a finite-dimensional quantum system. Then, we have 
\begin{equation}
\left(\mathcal{E}_m \rme^{\frac{t}{mn}\mathcal{L}}\cdots
\mathcal{E}_1 \rme^{\frac{t}{mn}\mathcal{L}}\right)^n
=\mathcal{E}_\varphi^{n}\rme^{t \mathcal{L}_Z}+\mathcal{O}(1/n)
\quad\text{as}\quad
n\to+\infty,
\label{eqn:CPTPZenomulti}
\end{equation}
uniformly in $t$ on compact intervals of $[0,+\infty)$,
with
\begin{equation}
\mathcal{L}_Z
=\sum_{|\lambda_k|=1}\mathcal{P}_k\overline{\mathcal{L}}\mathcal{P}_k,\qquad
\overline{\mathcal{L}}
=\frac{1}{m}\,\biggl(\mathcal{L} + \mathcal{E}_\varphi^{-1}\sum_{j=2}^m\mathcal{E}_m\cdots\mathcal{E}_j\mathcal{L}\mathcal{E}_{j-1}\cdots\mathcal{E}_1\biggr),
\label{eq:overLdef}
\end{equation}
where $\mathcal{P}_k$ is the spectral projection 
of $\mathcal{E} =\mathcal{E}_m\cdots\mathcal{E}_1$ belonging to the eigenvalue $\lambda_k$,
and $\mathcal{E}_\varphi$ and $\mathcal{E}_\varphi^{-1}$ are the peripheral part of $\mathcal{E}$ and its peripheral inverse, respectively.
\end{thm}
In particular, for $m=1$, we have the following corollary, which covers both QZDs via (i) projective measurements and via (ii) unitary kicks, and generalizes them to generic  kicks:
\begin{corol}[QZD by generic kicks]\label{thm:kicktozeno}
Let $\mathcal{E}$ be a quantum operation and $\mathcal{L}$ be a GKLS generator of a finite-dimensional quantum system. Then, we have 
\begin{equation}
\left(\mathcal{E}\rme^{\frac{t}{n}\mathcal{L}}\right)^n=\mathcal{E}_\varphi^{n}\rme^{t \mathcal{L}_Z}+  \mathcal{O}(1/n)
\quad\text{as}\quad
n\to+\infty,
\label{eqn:CPTPZeno}
\end{equation}
uniformly in $t$ on compact intervals of $[0,+\infty)$,
with
\begin{equation}
\mathcal{L}_Z
=\sum_{|\lambda_k|=1}\mathcal{P}_k\mathcal{L}\mathcal{P}_k,
\end{equation}
where $\mathcal{P}_k$ is the spectral projection  
of $\mathcal{E}$ belonging to the eigenvalue $\lambda_k$, and $\mathcal{E}_\varphi$ is the peripheral part of $\mathcal{E}$.
\end{corol}

If in the above statements the maps $\mathcal{E}$ and $\mathcal{E}_1,\ldots,\mathcal{E}_m$ are assumed to be CPTP and describe measurement processes, they are nonselective measurements.
An interesting corollary of Theorem~\ref{thm:CPTPBB} is available for selective measurements.
In particular, we provide a corollary for the QZD via cycles of 
multiple selective projective measurements. This is a generalization of the standard QZD via (i) frequent selective projective measurements.
\begin{corol}[QZD by cycles of projective measurements]
\label{thm:MeasBB}
Let $\{\mathcal{P}_1,\ldots, \mathcal{P}_m\}$ be a finite set of CP Hermitian projections on 
the Hilbert-Schmidt space $\mathcal{T}_2$ of operators on a finite-dimensional Hilbert space, and 
$\mathcal{L}$ be a GKLS generator.
The projections are not assumed to be pairwise orthogonal, i.e.~$\mathcal{P}_i \mathcal{P}_j\neq0$ for $i\neq j$, in general.
Then, we have 
\begin{equation}
\left(\mathcal{P}_m \rme^{\frac{t}{mn}\mathcal{L}}\cdots \mathcal{P}_1\rme^{\frac{t}{mn}\mathcal{L}}\right)^n
=\mathcal{P}_\varphi\rme^{t \mathcal{P}_\varphi \mathcal{L} \mathcal{P}_\varphi}+\mathcal{O}(1/n)
\quad\text{as}\quad
n\to+\infty,
\label{eqn:ZenoMultiMeas}
\end{equation}
uniformly in $t$ on compact intervals of $\mathbb{R}$,
where $\mathcal{P}_\varphi= \mathcal{P}_1\wedge \mathcal{P}_2 \wedge\cdots\wedge \mathcal{P}_m $ is the Hermitian projection  onto the intersection of the ranges of the projections $\mathcal{P}_1,\ldots,\mathcal{P}_m$.
If such intersection is trivial, then $\mathcal{P}_\varphi = 0$, and the sequence in Eq.\ (\ref {eqn:ZenoMultiMeas}) just decays to zero. 
\end{corol}
\begin{proof}
The proof makes use of the crucial fact that the peripheral part of the product of the Hermitian projections
$\mathcal{E}= \mathcal{P}_m \cdots \mathcal{P}_1$
reads 
\begin{equation}
\label{eq:peripheralprodP}
\mathcal{E}_\varphi = \mathcal{E}^{-1}_{\varphi}=\mathcal{P}_\varphi
\end{equation}
($\lambda=1$ is the only peripheral eigenvalue of $\mathcal{E}$), and
$\mathcal{P}_\varphi \mathcal{P}_j = \mathcal{P}_j \mathcal{P}_\varphi = \mathcal{P}_\varphi$ for all $j=1,\dots,m$.
This will be proved in  Lemma~\ref{lem:P1P2} in Appendix~\ref{app:BasicLemmas}.
Then, Eq.~\eqref{eq:peripheralprodP} implies that $\mathcal{E}_\varphi^{n}= \mathcal{P}_\varphi^{n}= \mathcal{P}_\varphi$, and $\overline{\mathcal{L}}$ in Eq.\ (\ref{eq:overLdef}) of Theorem~\ref{thm:CPTPBB} is simplified to
\begin{equation}
\overline{\mathcal{L}}
=\frac{1}{m}\,\biggl(\mathcal{L} + \mathcal{P}_\varphi \mathcal{L} \sum_{j=2}^m \mathcal{P}_{j-1}\cdots\mathcal{P}_1\biggr).
\end{equation}
Therefore, $\mathcal{L}_Z$ in Eq.\ (\ref{eq:overLdef}) of Theorem~\ref{thm:CPTPBB} reads
\begin{equation}
\mathcal{L}_Z =
\mathcal{P}_\varphi \overline{\mathcal{L}}\mathcal{P}_\varphi = \mathcal{P}_\varphi \mathcal{L}\mathcal{P}_\varphi ,
\end{equation}
and Eq.~\eqref{eqn:CPTPZenomulti} of Theorem~\ref{thm:CPTPBB} becomes Eq.~\eqref{eqn:ZenoMultiMeas}. 
\end{proof}
\begin{remark}
Let us consider a unitary evolution, $\mathcal{L} = -\rmi [H,{}\bullet{}]$ with $H=H^\dagger$.
For $\mathcal{P}_j = P_{j} \bullet P_j$, with $P_1,\ldots,P_m$ being Hermitian projections, Eq.~\eqref{eqn:ZenoMultiMeas} particularizes to 
\begin{equation}
\left(P_m \rme^{-\rmi \frac{t}{mn}H}\cdots P_1\rme^{-\rmi \frac{t}{mn}H}\right)^n
=P_\varphi\rme^{-\rmi tP_\varphi HP_\varphi}+\mathcal{O}(1/n)
\quad\text{as}\quad
n\to+\infty,
\label{eqn:ZenoMultiMeasP}
\end{equation}
where $P_\varphi= P_1\wedge P_2 \wedge\cdots\wedge P_m$ is the Hermitian projection onto the intersection of the ranges of $P_1,\ldots,P_m$, and one gets a QZD by cycles of (nonorthogonal) selective measurements.
More generally, if 
\begin{equation}
\mathcal{P}_j = \sum_{k=1}^{n_j} P_k^{(j)} \bullet P_k^{(j)},
\label{eqn:SelectiveMeas}
\end{equation}
where $\{P_1^{(j)},\ldots,P_{n_j}^{(j)}\}_{j=1,\ldots,m}$ are sets of Hermitian projections with 
$P^{(j)}_k P^{(j)}_\ell = \delta_{k\ell} P_k^{(j)}$, 
then one gets a QZD by cycles of (nonorthogonal) partially selective measurements. 
A particular case is when $\sum_{k=1}^{n_j} P_k^{(j)}=I$ for all $j=1,\dots,m$, and one has a cycle of nonselective (i.e.~CPTP) measurements.
There exist more general CP Hermitian projections, that cannot be cast in the form (\ref{eqn:SelectiveMeas}).
For instance, $\mathcal{P}=\frac{I}{d}\tr({}\bullet{})$ is a CPTP Hermitian projection for a $d$-dimensional system.
Corollary~\ref{thm:MeasBB} works for general CP Hermitian projections including such a projection.
\end{remark}
The proof of  Theorem~\ref{thm:CPTPBB} consists of several steps as outlined in Fig.~\ref{proofsummary}. Before we prove the theorem, we provide the key lemmas in the next section.
\begin{figure}
\centering
\includegraphics[width=0.75\textwidth]{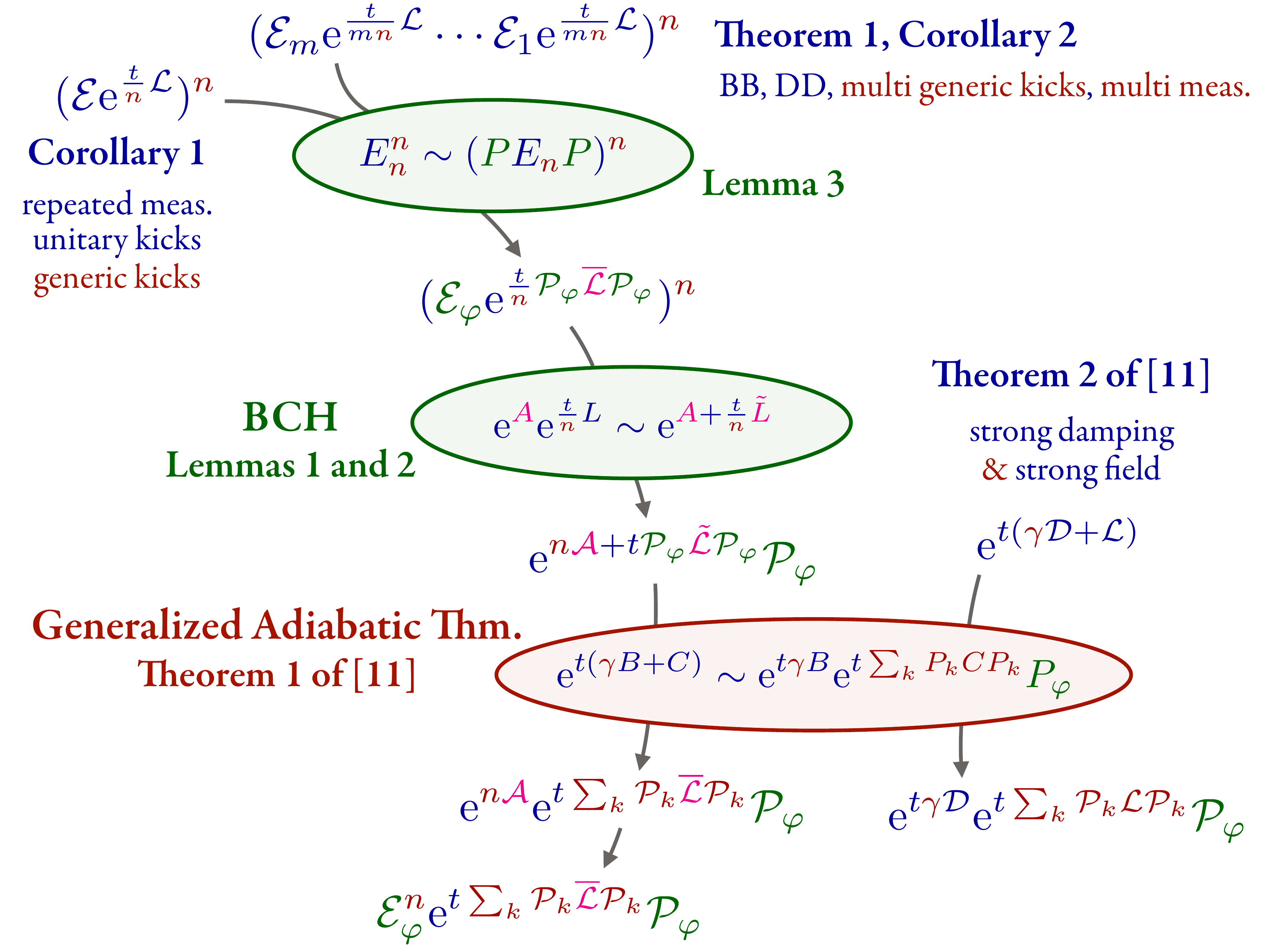}
\caption{\label{proofsummary}Sketch of the proof steps.}
\end{figure}

\section{Key Lemmas}
\label{sec:KeyLemmas}
The key idea is to bridge from the pulsed strategies to the continuous strategies via the BCH formula \cite{ref:Hall-BCH}, and then prove the Zeno limit by the generalized adiabatic theorem, which was proved and used to unify the continuous strategies in Ref.\ \cite{ref:unity1} (Theorems~1 and~2 therein). 
In this way, the pulsed strategies are unified with the continuous strategies.
The key lemma for the bridge is the following generalized BCH formula:
\begin{lemma}[Pulsed vs continuous for invertible kicks]\label{thm:kicktofield}
Let $E$ and $L$ be linear operators on a finite-dimensional Banach space, with $E$ invertible.  
Let $A=\log E$ be a primary logarithm of $E$ so that $\rme^A=E$.
Then, we have
\begin{equation}
\left(E \rme^{\frac{t}{n}L}\right)^n=\rme^{n A +t  \tilde{L}+\mathcal{O}(1/n)}
\quad\text{as}\quad
n\to + \infty,
\label{eq:pulvscontlim}
\end{equation}
uniformly in $t$ on  compact intervals of $\mathbb{R}$,
with
\begin{equation}
\tilde{L}=g(\ad_{A})(L),
\label{eq:tildeLdef}
\end{equation}
where $g$ is the meromorphic function on $\mathbb{C}$ defined by
\begin{equation}
g(z) = 
\begin{cases} 
\medskip
\displaystyle{\frac{z}{1-\rme^{-z}}} & (z\notin 2\pi\rmi \mathbb{Z}),\\
1 & (z=0) ,
\end{cases}
\label{eq:g{z}def}
\end{equation}
and $\ad_A=[A,{}\bullet{}]$.
\end{lemma}
\begin{proof}
	We will prove it in Sec.~\ref{sec:ProofKeyLemma}.
\end{proof}
\begin{remark}
The assumption that $A$ be a \emph{primary} logarithm of $E$ \cite{ref:MatrixFunctions-Higham,ref:MatrixAnalysisTopics-HornJohnson} is necessary to get Eq.~\eqref{eq:pulvscontlim}, as the following example shows.
Take $E=I$ and $L=X$, the identity and the first Pauli matrix on $\mathbb{C}^2$, respectively. Then, we have 
\begin{equation}
\left(E \rme^{\frac{t}{n}L}\right)^n = \left(I \rme^{\frac{t}{n}X}\right)^n = \rme^{t X}.
\end{equation}
On the other hand, consider $A = 2\pi \rmi Z$, where $Z$ is the third Pauli matrix. Then, $\rme^A= \rme^{2\pi \rmi Z}=I =E$, but $A$ is \emph{not} a primary logarithm of $E$. It is apparent that there exists no matrix $\tilde{L}$ such that
\begin{equation}
\rme^{t X} = \rme^{2 n \pi \rmi Z + t\tilde{L}} + \mathcal{O}(1/n),
\end{equation}
since $\rme^{2 n \pi \rmi Z + t\tilde{L}}= \rme^{t\tilde{L}_Z}+\mathcal{O}(1/n)$ as a strong-coupling limit \cite{ref:unity1}, with $\tilde{L}_Z$ a diagonal matrix. Thus, Eq.~\eqref{eq:pulvscontlim} does not hold for $A = 2\pi \rmi Z$. 
\end{remark}

We can apply this lemma to physical situations where $E$ are invertible quantum operations and $L$ are GKLS generators.
This lemma is however useful only for invertible $E$, and cannot accommodate e.g.~the standard QZD via projective measurements. To circumvent this problem, we consider instead the primary logarithm of $E+Q$, with $Q$ a projection onto the kernel of $E$, and by projecting $L$ on a complementary space:
\begin{lemma}[Pulsed vs continuous for noninvertible kicks]\label{thm:kicktofieldP}
Let $E$ and $L$ be linear operators on a finite-dimensional Banach space. Let $Q$ be a projection onto $\ker E$ and set $P=1-Q$. Let $A=\log (E +Q)$ be a primary logarithm of the invertible operator $E+Q$, 
so that $\rme^A=E+Q$.
Then, we have
\begin{equation}
\left(E \rme^{\frac{t}{n} P L P}\right)^n=\rme^{nA+t\tilde{L}+\mathcal{O}(1/n)}P
\quad\text{as}\quad
n\to + \infty,
\label{eq:pulvscontlimP}
\end{equation}
uniformly in $t$ on  compact intervals of $\mathbb{R}$,
with
\begin{equation}
\tilde{L}=Pg(\ad_{A})(L)P,
\label{eq:tildeLdefP}
\end{equation}
where $g$ is the meromorphic function on $\mathbb{C}$ defined in Eq.\ (\ref{eq:g{z}def}).
\end{lemma}
\begin{proof}
Notice first that, even when $E$ is not invertible, $F = E+Q$ is invertible, and we can consider one of its primary logarithms, say $A= \log F$.
Then, we can apply Lemma~\ref{thm:kicktofield} as
\begin{equation}
 	\left(E \rme^{\frac{t}{n}P L P}\right)^n
	=     	 \left( F\rme^{\frac{t}{n} P L P} \right)^n P
	=	\left(\rme^{A}\rme^{\frac{t}{n}P L P}\right)^nP
	=\rme^{n A+t g(\ad_{A})(PLP)+\mathcal{O}(1/n)}P 
\label{eqn:4.8}
\end{equation}
for large $n$.
Since $[P, F]=0$ and $A=\log F$ is primary, it implies $[P,A]=0$, and we have 
\begin{equation}
g(\ad_{A})(PLP)=P g(\ad_{A})(L)P = \tilde{L}.
\end{equation}
The statement of the lemma thus holds.
\end{proof}

\begin{remark}
If $E$ is invertible, i.e.\ $\ker E = \{0\}$, then $Q=0$ and $P=1$,
and Lemma~\ref{thm:kicktofieldP} is reduced to Lemma~\ref{thm:kicktofield}.
\end{remark}
\begin{remark}
If $L$ in the exponent on the left-hand side of Eq.\ (\ref{eq:pulvscontlimP}) is not projected by $P$ as $PLP$, we are not allowed to promote $E$ to $E+Q$ to define $A=\log(E+Q)$, since $\rme^{\frac{t}{n}L}$ is in general not commutative with $P$.
\end{remark}

The second  ingredient to bridge a pulsed dynamics to a continuous one is the following approximation lemma:
\begin{lemma}[Asymptotic projection of a sequence of operators]
\label{prop:Peripheral}
Let $(E_n)$ be a sequence of linear operators on a finite-dimensional Banach space and $P\,(=P^2)$ be a projection.
Assume that the following conditions hold:
\begin{enumerate}
\item
The operators $E_n$  asymptotically commute with $P$ as
\begin{equation}
  PE_n = E_n P + \mathcal{O}(1/n)
\quad\text{as}\quad
n\to+\infty.
  \label{eqn:FamilyAsympMatrix}
\end{equation}

\item There exist $M\ge0$ and $n_0>0$ such that, for all $n\ge n_0$,
\begin{equation}
\|(PE_nP)^k\|\le M,
\quad\forall k\in\mathbb{N}.
\label{eqn:Bound1}
\end{equation}

\item There exist $K\ge0$, $\mu\in[0,1)$, and $n_0>0$ such that, for all $n\ge n_0$,
\begin{equation}
\|(E_n-PE_nP)^k\|\le K \mu^k,
\quad\forall k\in\mathbb{N}.
\label{eqn:Bound2}
\end{equation}
\end{enumerate}
Then, we have
\begin{equation}
E_n^n=(PE_nP)^n+\mathcal{O}(1/n)
\quad\text{as}\quad
n\to+\infty.
\end{equation}
\end{lemma}
\begin{proof}
	The proof is given in Appendix~\ref{sec:ProofPeripheral}\@.
\end{proof}
\begin{remark}
\label{remark:AsympProjCPTP}
For a sequence of quantum operations $(\mathcal{E}_n)$ converging to a quantum operation $\mathcal{E}$ as $\mathcal{E}_n=\mathcal{E}+\mathcal{O}(1/n)$ as $n\to+\infty$, all the conditions~1--3 of Lemma~\ref{prop:Peripheral} are automatically fulfilled with the peripheral projection $\mathcal{P}_\varphi$ of $\mathcal{E}$ taken as $P$.
See Lemmas~\ref{lem:NormCPT} and \ref{lem:BoundK} in Appendix~\ref{app:BasicLemmas}, which guarantee conditions~2 and~3 for the sequence of quantum operations $(\mathcal{E}_n)$.
Then, according to Lemma~\ref{prop:Peripheral}, we have
\begin{equation}
\mathcal{E}_n^n=(\mathcal{P}_\varphi\mathcal{E}_n\mathcal{P}_\varphi)^n+\mathcal{O}(1/n)
\quad\text{as}\quad
n\to+\infty.
\label{eqn:AsympProjCPTP}
\end{equation}
\end{remark}

We will use these lemmas to prove Theorem~\ref{thm:CPTPBB}.

\section{Proof of Lemma~\ref{thm:kicktofield}}
\label{sec:ProofKeyLemma}
Let us prove Lemma~\ref{thm:kicktofield}, which is the key to the proof of Theorem~\ref{thm:CPTPBB}.
\begin{proof}[Proof of Lemma~\ref{thm:kicktofield}]
First, let us recall some properties of functions of operators on a finite-di\-men\-sion\-al Banach space (see e.g.\ Refs.\ \cite{ref:Katobook}, 
\cite[Chap.~1]{ref:MatrixFunctions-Higham}, and \cite[Sec.~6.2]{ref:MatrixAnalysisTopics-HornJohnson}). 
Given a function $h(z)$ on the complex plane, we wish to define a function $h(X)$ of operators $X$. Notice that, in general, an operator function $h(X)$ does not have a series expansion, unless the spectrum of $X$ lies within the convergence radius of a power series of the function $h(z)$. Neverthless, by making use of the resolvent $(zI - X)^{-1}$ of the operator $X$, we can define functions $h(X)$ of $X$ for a large class of functions $h$.

Suppose that $h(z)$ is holomorphic in a domain $\Delta$ of the complex plane containing the spectrum of $X$, and let $\Gamma\subset \Delta$ be a  smooth curve with positive direction enclosing all the eigenvalues in its interior. Then, a \textit{primary function} $h(X)$ is defined by the Dunford-Taylor integral
\begin{equation}\label{integral}
h(X)=\frac{1}{2\pi \rmi} \oint_\Gamma \d z \, h(z) \frac{1}{zI-X}.
\end{equation}
More generally, $\Gamma$ may consist of several simple closed curves, such that the union of their interiors contains all the eigenvalues of $X$. 
Note that Eq.~\eqref{integral} does not depend on $\Gamma$ as long as the latter satisfies these conditions, that is $h(z)$ is holomorphic in $\Delta$ and $\Gamma\subset \Delta$.

Now, let us start the proof of Lemma~\ref{thm:kicktofield}.

\paragraph{Step 1.}
We wish to define a logarithm of $E$.
Since $E$ is invertible, its spectrum $\sigma (E)$ does not contain $0$. Choose a half-line $\mathrm{c}=\{r \rme^{\rmi\varphi}\in\mathbb{C}\,|\,r\ge0\}$ such that $\mathrm{c}\cap \sigma(E)=\emptyset$, and let $h(z)=\log z$ denote a branch of the logarithm function.
Take a contour $\Gamma$ enclosing all the eigenvalues of $E$ and contained in $\Delta=\mathbb{C}\setminus \mathrm{c}$.
See Fig.~\ref{fig:ContourGamma}.
Since  $h(z)$ is analytic in $\Delta$, it can be taken as a stem function to define a primary logarithm function $A = \log E$ by Eq.\ (\ref{integral}), and $\rme^A= E$ is inherited from the functional properties of the  function $h(z)$. Note that there is a neighborhood of $E$ on which this logarithm function is well-defined \cite{ref:Katobook,LAX}. 
\begin{figure}
\centering
\includegraphics[scale=0.35]{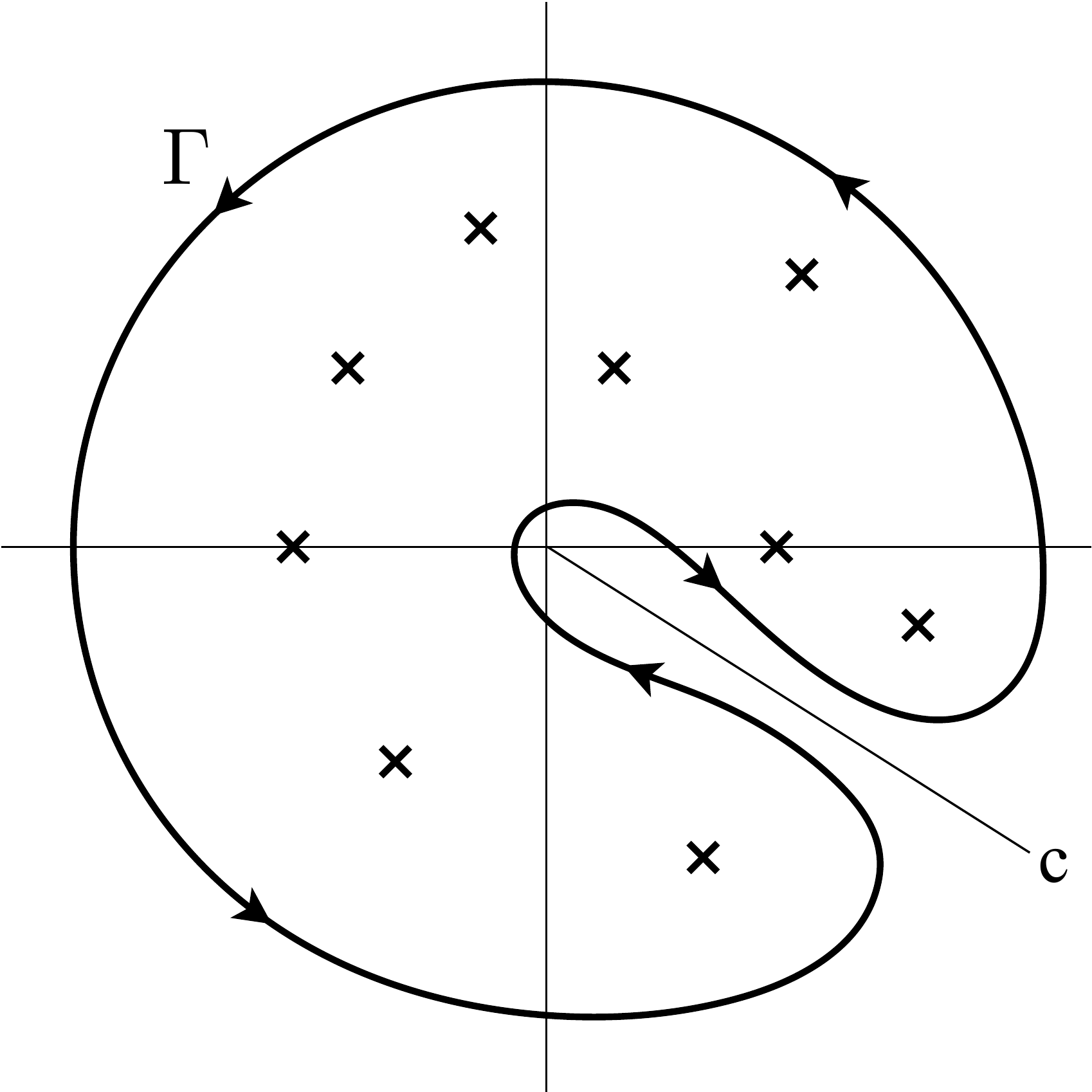}
\caption{An example of contour $\Gamma$ to define a primary logarithm function of an operator $E$. The crosses represent the eigenvalues of $E$, the cut $\mathrm{c}$ does not intersect any eigenvalue, 
and the contour $\Gamma$ runs in the domain of analiticity $\Delta = \mathbb{C}\setminus \mathrm{c}$ of a branch $h(z)=\log z$ of the logarithm function, and encloses all the eigenvalues.}
\label{fig:ContourGamma}
\end{figure}

\paragraph{Step 2.}
Given the operator $A=\log E$, we follow the proof of the BCH formula given in Sec.~5.5 of Ref.\ \cite{ref:Hall-BCH}. 
For $n$ large enough, $\rme^{A}\rme^{\frac{s}{n}L}=E\rme^{\frac{s}{n}L}$
is invertible and lies in the neighborhood of $E$ for all $0\le s\le t$, whence its logarithm is defined by the integral~\eqref{integral} along the same contour $\Gamma$. Let
\begin{equation}
Z(s)=\log(\rme^{A}\rme^{\frac{s}{n}L})
\end{equation}
for $0\le s\le t$.
$Z(s)$ is an analytic operator-valued function
and
\begin{equation}\label{onehand}
	\rme^{-Z(s)}\frac{\d}{\d  s}\rme^{Z(s)}=(\rme^{A}\rme^{\frac{s}{n}L})^{-1}\rme^{A}\rme^{\frac{s}{n}L}\frac{1}{n} L=\frac{1}{n} L.
\end{equation}
On the other hand, by Theorem~5.4 of Ref.\ \cite{ref:Hall-BCH}, we have
\begin{equation}\label{otherhand}
\rme^{-Z(s)}\frac{\d}{\d  s}\rme^{Z(s)}=f(\ad_{Z(s)})\!\left(\frac{\d Z(s)}{\d s}\right),
\end{equation}
where 
\begin{equation}
f(z)=\sum_{n=0}^\infty\frac{(-1)^n}{(n+1)!}z^n = \begin{cases} 
\medskip
\displaystyle
\frac{1-\rme^{-z}}{z} & (z\neq 0),\\
1 & (z=0) 
\end{cases}
\end{equation}
is an entire analytic function, and the superoperator $f(\ad_{Z(s)})$ is defined by Eq.~\eqref{integral} for a given curve $\Gamma$ enclosing the spectrum of $\ad_{Z(s)}$ for all $0\le s\le t$. Notice that $f(\ad_{Z(0)}) = f(\ad_A)$, and that $f(z)=0$ only at the imaginary points $z_k= 2\pi \rmi k$ with $k\in\mathbb{Z}\setminus\{0\}$.

\paragraph{Step 3.}
Now, we wish to invert Eq.~\eqref{otherhand} in order to obtain an explicit expression for the derivative~$\d Z(s) / \d s$. We claim that $f(\ad_{A})$ is invertible. That is to say that all eigenvalues of $\ad_A$ are not zeros of $f(z)$, namely that $\ker(\ad_A - z_k)=\{0\}$ for all $k\in\mathbb{Z}\setminus\{0\}$. Indeed, let the operator $X_*$ belong to $\ker(\ad_A - z_k)$ for some $k\in\mathbb{Z}\setminus\{0\}$.
Then,
\begin{equation}
\label{eq:adAeigen}
\ad_A(X_*) = [A, X_*] = z_k X_*.
\end{equation} 
By exponentiating it, we get
$\rme^{t \ad_A} (X_*) = \rme^{t A} X_*\rme^{-t A} = \rme^{t z_k} X_*$ for all $t\in\mathbb{R}$. In particular, at $t=1$,
\begin{equation}
\rme^{A} X_*\rme^{-A} = \rme^{z_k} X_*, 
\end{equation}
that is $E X_* E^{-1} = X_*$, whence
$[E,X_*]=0$. But this implies that also $A=\log E$, as a function of $E$, commutes with $X_*$, namely 
\begin{equation}
\ad_A(X_*) = [A, X_*] = 0,
\end{equation}
which, together with Eq.~\eqref{eq:adAeigen}, implies that $X_*=0$, since $z_k\neq 0$.

Hence the superoperator $f(\ad_{A})$ is invertible. Furthermore, this implies that the inverse of $f(\ad_{Z(s)})$ does
exist for all $0\le s\le t$,  if $n$ is large enough.
It is given by $g(\ad_{Z(s)})$, defined by the stem function $g(z)=1/f(z)$, that is the meromorphic function given in Eq.~\eqref{eq:g{z}def}.

We can then combine Eqs.\ (\ref{onehand}) and (\ref{otherhand}) to obtain
\begin{equation}
\frac{\d Z(s)}{\d s}
=\frac{1}{n}
g(\ad_{Z(s)})(L).
\label{eq:inverted}
\end{equation}
Noting that $Z(0)=A$, we integrate Eq.\ (\ref{eq:inverted}) to get
\begin{equation}
Z(t)
=\log(\rme^{A}\rme^{\frac{t}{n}L})=A+\frac{1}{n} 
\int_0^t \d s\,
g(\ad_{Z(s)})(L).
\label{eq:inverted-1-1}
\end{equation}

\paragraph{Step 4.}
We are only interested in terms up to $\mathcal{O}(1/n)$ as $n\to+\infty$. 
In general, by the integral representation (\ref{integral}) for $h(z)$ analytic on a domain $\Delta$ and for the spectrum of $X+\frac{1}{n}Y$ enclosed by $\Gamma\subset \Delta$, where $X$ and $Y$ are operators of order $\mathcal{O}(1)$ on a finite-dimensional Banach space, we have
\begin{align}
h\!\left(X+\frac{1}{n}Y\right) 
&=\frac{1}{2\pi \rmi} \oint_\Gamma \d z\, h(z) \frac{1}{zI-X-\frac{1}{n}Y}
\nonumber\\
&=\frac{1}{2\pi \rmi} \oint_\Gamma \d z\, h(z)
\left(
\frac{1}{zI-X}+\frac{1}{n}\frac{1}{zI-X}Y\frac{1}{zI-X-\frac{1}{n}Y}
\right)
\nonumber\\
&=h(X)+\mathcal{O}(1/n),
\vphantom{\oint}
\end{align}
for all $n>n_0$ for some $n_0>0$.
Therefore, for $h(z)=\log z$ we get $
Z(s)=\log(\rme^{A}\rme^{\frac{s}{n}L})=A+\mathcal{O}(1/n)$ by choosing $X=\rme^A$ and $Y=n\rme^{A}(\rme^{\frac{s}{n}L}-I)=\mathcal{O}(1)$, which implies $\ad_{Z(s)}=\ad_A+\mathcal{O}(1/n)$. Then, expanding 
$g(\ad_{Z(s)})$ in Eq.\ (\ref{eq:inverted-1-1}), 
we get
\begin{align}
\log(\rme^{A}\rme^{\frac{t}{n}L}) 
&=A+\frac{1}{n} \int_{0}^t \d s\,  
g(\ad_{A})(L) 
+\mathcal{O}(1/n^2)
\nonumber\displaybreak[0]\\
&= A+\frac{t}{n}
g(\ad_{A})(L) 
+\mathcal{O}(1/n^2)
\nonumber\displaybreak[0]\\
&= A+\frac{t}{n}
\tilde{L} 
+\mathcal{O}(1/n^2),
\vphantom{\int_0^t}
\label{eq:inverted-1-1-1}
\end{align}
with $\tilde{L}$ as in Eq.~\eqref{eq:tildeLdef}.
Exponentiating it, we obtain
\begin{equation}
	\rme^{A}\rme^{\frac{t}{n}L}
	=\rme^{
A+\frac{t}{n} \tilde{L}+\mathcal{O}(1/n^2)
},
\end{equation}
whence 
\begin{equation}
(\rme^{A}\rme^{\frac{t}{n}L})^{n}
=\rme^{nA+t\tilde{L}+\mathcal{O}(1/n)},
\label{eqn:6.14}
\end{equation}
which gives the result~\eqref{eq:pulvscontlim} for $E$ invertible. 
Uniformity in $t$ on compact intervals of $\mathbb{R}$ is straightforward.
We provide, in Proposition~\ref{prop:BoundBCH} in Appendix~\ref{app:BCH}, a concise and explicit expression for the bound on the correction $\mathcal{O}(1/n^2)$ in Eq.\ (\ref{eq:inverted-1-1-1}), which ensures that the correction in Eq.\ (\ref{eqn:6.14}) is $\mathcal{O}(1/n)$ for any finite $t$.
\end{proof}
\begin{remark}
We remark that, compared to the BCH formula, this lemma does not have to impose the condition on $E-I$ or $L$ being small. The difference is that we have the freedom to choose $n$ large, and that we have the freedom to choose an appropriate logarithm $A$. 
\end{remark}
\begin{remark}
Furthermore, we note that implementing $g(\ad_A)$ for the stem function $g(z)$ given in Eq.\ (\ref{eq:g{z}def}) numerically (for a matrix $A$) is a difficult business because $1-\rme^{-{\ad_A}}$ of its denominator is not invertible itself. We can define such matrix functions in general through Jordan form and the derivatives of the stem function $g(z)$.  However, this brings in the usual stability issues of the Jordan form. Instead, we can implement it via
\begin{equation}
f(\ad_A)(X)=\left(
\frac{1-\rme^{-{\ad_A}}}{\ad_A}
\right)\!(X)
=\left.\rme^{-A}\frac{\d}{\d  t}\rme^{A+tX}\right|_{t=0},
\end{equation}
which is obtained by applying the formula (\ref{otherhand}) to $Z(s)=A+sX$.
Notice that the derivative $\left.\frac{\d}{\d  t}\rme^{A+tX}\right|_{t=0}$ on the right-hand side is easy to implement. For instance \cite[Theorem~3.6]{ref:MatrixFunctions-Higham}, it is given by the top-right block of 
\begin{equation}
\exp\!\begin{pmatrix}{A} & X \\ 0 & {A} \end{pmatrix}.
\end{equation}
We vectorize the matrix $X$ and eventually get the matrix elements of the supermap $f(\ad_A)$, which we then invert.
\end{remark}

\section{Proof of Theorem~\ref{thm:CPTPBB}}
\label{sec:ProofTheorems}
Now, we prove Theorem~\ref{thm:CPTPBB}. 
\begin{proof}[Proof of Theorem~\ref{thm:CPTPBB}]
The proof consists of three steps.
We first use Lemma~\ref{prop:Peripheral} to cut the nonperipheral part of the  kick.
Then, we apply Lemma~\ref{thm:kicktofieldP} to bridge from the pulsed strategy to the continuous strategy, which opens a way to carry out the Zeno limit by the generalized adiabatic theorem (Theorem~1 of Ref.\ \cite{ref:unity1}).

\paragraph{Step 1.}
First, we claim that as $n\to+\infty$ one gets
\begin{equation}
\left(\mathcal{E}_m\rme^{\frac{t}{mn}\mathcal{L}}\cdots\mathcal{E}_1\rme^{\frac{t}{mn}\mathcal{L}}\right)^n
=\left(\mathcal{E}_\varphi \rme^{\frac{t}{n}\mathcal{P}_\varphi\overline{\mathcal{L}}\mathcal{P}_\varphi+\mathcal{O}(1/n^2)}\right)^n
+\mathcal{O}(1/n)
\quad\text{as}\quad
n\to+\infty,
\label{eqn:ThmPeripheralMulti}
\end{equation}
where $\overline{\mathcal{L}}$ is given in Eq.~\eqref{eq:overLdef},
$\mathcal{P}_\varphi$ is the projection onto the peripheral spectrum of $\mathcal{E}=\mathcal{E}_m\cdots\mathcal{E}_1$, and $\mathcal{E}_\varphi$ and $\mathcal{E}_\varphi^{-1}$ are the peripheral part of $\mathcal{E}$ and its peripheral inverse, respectively. 
Indeed, $\tilde{\mathcal{E}}_n=\mathcal{E}_m\rme^{\frac{t}{mn}\mathcal{L}}\cdots\mathcal{E}_1\rme^{\frac{t}{mn}\mathcal{L}}$ is a quantum operation, and it approaches $\mathcal{E}$ as
\begin{align}
\tilde{\mathcal{E}}_n
&=\mathcal{E}_m\rme^{\frac{t}{mn}\mathcal{L}}\cdots\mathcal{E}_1\rme^{\frac{t}{mn}\mathcal{L}}
\nonumber\displaybreak[0]\\
&=\mathcal{E}
+\frac{t}{mn}\, \biggl(\mathcal{E} \mathcal{L} +\sum_{j=2}^m\mathcal{E}_m\cdots\mathcal{E}_j\mathcal{L}\mathcal{E}_{j-1}\cdots\mathcal{E}_1 \biggr)
+\mathcal{O}(1/n^2)
\label{eqn:BBseq}
\end{align}
as $n$ increases.
An explicit bound on $\|\tilde{\mathcal{E}}_n-\mathcal{E}\|$, which is $\mathcal{O}(1/n)$, is given in Lemma~\ref{lem:BBperturb} in Appendix~\ref{app:BasicLemmas}.
Thus, the conditions for Lemma~\ref{prop:Peripheral} are all satisfied (see Remark~\ref{remark:AsympProjCPTP}), and by Lemma~\ref{prop:Peripheral} we have
\begin{equation}
\tilde{\mathcal{E}}_n^n
=(\mathcal{P}_\varphi\tilde{\mathcal{E}}_n\mathcal{P}_\varphi)^n
+\mathcal{O}(1/n).
\label{eqn:PropPeripheral0Multi}
\end{equation}
An explicit bound on $\|\tilde{\mathcal{E}}_n^n
-(\mathcal{P}_\varphi\tilde{\mathcal{E}}_n\mathcal{P}_\varphi)^n\|$, which is  $\mathcal{O}(1/n)$, is given in Eq.\ (\ref{eqn:TotalCorrection}), in the proof of Lemma~\ref{prop:Peripheral} in Appendix~\ref{sec:ProofPeripheral}\@.
Since
\begin{align}
	\mathcal{P}_\varphi\tilde{\mathcal{E}}_n\mathcal{P}_\varphi
	&=\mathcal{P}_\varphi\mathcal{E}_m\rme^{\frac{t}{mn}\mathcal{L}}\cdots\mathcal{E}_1\rme^{\frac{t}{mn}\mathcal{L}}\mathcal{P}_\varphi
\nonumber\displaybreak[0]\\
	&=\mathcal{P}_\varphi\,\biggl[
	\mathcal{E} + \frac{t}{mn}\,\biggl(
	\mathcal{E} \mathcal{L} + \sum_{j=2}^m\mathcal{E}_m\cdots\mathcal{E}_j\mathcal{L}\mathcal{E}_{j-1}\cdots\mathcal{E}_1 
	\biggr)+\mathcal{O}(1/n^2)
	\biggr]\,\mathcal{P}_\varphi
\nonumber\displaybreak[0]\\
	&=\mathcal{E}_\varphi\,\biggl[1+  \frac{t}{mn} \mathcal{P}_\varphi\, \biggl(  \mathcal{L}  
	+\mathcal{E}_\varphi^{-1}\sum_{j=2}^m\mathcal{E}_m\cdots\mathcal{E}_j\mathcal{L}\mathcal{E}_{j-1}\cdots\mathcal{E}_1\biggr) \,\mathcal{P}_\varphi +\mathcal{O}(1/n^2)\biggr]
\nonumber\displaybreak[0]\\
	&
	=\mathcal{E}_\varphi\left(
1+\frac{t}{n}\mathcal{P}_\varphi\overline{\mathcal{L}}\mathcal{P}_\varphi+\mathcal{O}(1/n^2)
\right)
\nonumber\displaybreak[0]\\
	&=\mathcal{E}_\varphi \rme^{\frac{t}{n}\mathcal{P}_\varphi\overline{\mathcal{L}}\mathcal{P}_\varphi+\mathcal{O}(1/n^2)},
\vphantom{\sum_{j=2}^m}
\label{eqn:ProofTh1Step1Exp}
\end{align}
Eq.\ (\ref{eqn:PropPeripheral0Multi}) implies Eq.\ (\ref{eqn:ThmPeripheralMulti}). 
The correction $\mathcal{O}(1/n^2)$ in the exponent of the last expression in Eq.\ (\ref{eqn:ProofTh1Step1Exp}) can be bounded by using Lemma~\ref{lem:MatrixFuncBound} in Appendix~\ref{app:BCH}\@.

\paragraph{Step 2.}
By applying Lemma~\ref{thm:kicktofieldP} 
with $E=\mathcal{E}_\varphi$, $L=\overline{\mathcal{L}}$, and $P=\mathcal{P}_\varphi$
to the right-hand side of Eq.\ (\ref{eqn:ThmPeripheralMulti}), one gets for large $n$
\begin{equation}
\left(\mathcal{E}_m\rme^{\frac{t}{mn}\mathcal{L}}\cdots\mathcal{E}_1\rme^{\frac{t}{mn}\mathcal{L}}\right)^n
=\rme^{n\mathcal{A}+t \tilde{\mathcal{L}}+\mathcal{O}(1/n)}\mathcal{P}_\varphi
+\mathcal{O}(1/n),
\label{eqn:4.10}
\end{equation}
where $\mathcal{A}$ is a primary logarithm of $\mathcal{E}_\varphi+(1-\mathcal{P}_\varphi)$, and 
\begin{equation}
\tilde{\mathcal{L}}=\mathcal{P}_\varphi g(\ad_{\mathcal{A}})(\overline{\mathcal{L}}) \mathcal{P}_\varphi.
\end{equation}

\paragraph{Step 3.}
Now, we set $n=\gamma t$ and consider
$
\rme^{\gamma t  \mathcal{A}+ t \tilde{\mathcal{L}}+\mathcal{O}(1/n)}\mathcal{P}_\varphi
$
for arbitrary $\gamma$, noting that for $\gamma t$ noninteger this in general is not a physical (quantum operation) map. 
Since $\mathcal{E}_\varphi$ and $1-\mathcal{P}_\varphi$ are both diagonalizable, $\mathcal{A}$ is also diagonalizable,
\begin{equation}
\mathcal{A}=\sum_{|\lambda_k|=1,\lambda_k\neq1}a_k\mathcal{P}_k+a_0[\mathcal{P}_0+(1-\mathcal{P}_\varphi)],
\end{equation}
with purely imaginary spectrum $a_k$, such that $\rme^{a_k} = \lambda_k$,
where $\mathcal{P}_k$ is the spectral projection of $\mathcal{E}$ belonging to the eigenvalue $\lambda_k$ with $\mathcal{P}_0$ belonging to the unit eigenvalue $\lambda_0=1$ [note that $\mathcal{P}_0+(1-\mathcal{P}_\varphi)$ is the spectral projection of $\mathcal{E}_\varphi+(1-\mathcal{P}_\varphi)$ belonging to the  eigenvalue $\lambda_0=1$ and hence the spectral projection of the primary logarithm $\mathcal{A}$ of $\mathcal{E}_\varphi+(1-\mathcal{P}_\varphi)$ belonging to the eigenvalue $a_0$, which is an integer multiple of $2\pi\rmi$], so that Theorem~1 of Ref.\ \cite{ref:unity1} (generalized adiabatic theorem) can be applied.
In the adiabatic limit, Eq.\ (\ref{eqn:4.10}) becomes
\begin{equation}
\left(\mathcal{E}_m\rme^{\frac{t}{mn}\mathcal{L}}\cdots\mathcal{E}_1\rme^{\frac{t}{mn}\mathcal{L}}\right)^n
=\rme^{n \mathcal{A}}
\rme^{t\hat{\mathcal{P}}(\tilde{\mathcal{L}})+\mathcal{O}(1/n)}\mathcal{P}_\varphi
+\mathcal{O}(1/n),
\label{eqn:IntermediateResult}
\end{equation} 
where the evolution is projected by
\begin{equation}
\hat{\mathcal{P}}({}\bullet{})=\sum_{|\lambda_k|=1,\lambda_k\neq1}\mathcal{P}_k{}\bullet{}\mathcal{P}_k+[\mathcal{P}_0+(1-\mathcal{P}_\varphi)]{}\bullet{}[\mathcal{P}_0+(1-\mathcal{P}_\varphi)].
\end{equation}
We notice that
\begin{equation}	
\hat{\mathcal{P}}(\mathcal{P}_\varphi{}\bullet{}\mathcal{P}_\varphi)=\mathcal{P}_\varphi\hat{\mathcal{P}}({}\bullet{})\mathcal{P}_\varphi=\sum_{|\lambda_k|=1}\mathcal{P}_k{}\bullet{}\mathcal{P}_k.
\end{equation}
In addition, we have $\hat{\mathcal{P}}\circ{\ad_\mathcal{A}}=0$ and hence $\hat{\mathcal{P}}\circ g(\ad_{\mathcal{A}})=\hat{\mathcal{P}}$.
Therefore, 
\begin{equation}	
\hat{\mathcal{P}}(\tilde{\mathcal{L}})=
\hat{\mathcal{P}}(\mathcal{P}_\varphi g(\ad_{\mathcal{A}})(\overline{\mathcal{L}}) \mathcal{P}_\varphi)
=
\mathcal{P}_\varphi\hat{\mathcal{P}}(g(\ad_{\mathcal{A}})(\overline{\mathcal{L}}))\mathcal{P}_\varphi
=
\mathcal{P}_\varphi\hat{\mathcal{P}}(\overline{\mathcal{L}})\mathcal{P}_\varphi
=\sum_{|\lambda_k|=1}\mathcal{P}_k\overline{\mathcal{L}}\mathcal{P}_k
=\mathcal{L}_Z.
\end{equation}
Equation (\ref{eqn:IntermediateResult}) is thus nothing but Eq.\ (\ref{eqn:CPTPZenomulti}) of the theorem.
\end{proof}

\begin{remark}
Let us see how the correction to the QZD in Eq.\ (\ref{eqn:CPTPZenomulti}) of Theorem~\ref{thm:CPTPBB} depends on $t$ and $n$.
As proved in Lemma~\ref{lem:BBperturb} in Appendix~\ref{app:BasicLemmas}, the correction $\mathcal{O}(1/n)$ in Eq.~(\ref{eqn:BBseq}) is bounded by a function of $\frac{t}{n}\|\mathcal{L}\|$.
This is inherited by the correction $\mathcal{O}(1/n)$ in Eq.\ (\ref{eqn:PropPeripheral0Multi}), through $C_n$ in the explicit bound (\ref{eqn:TotalCorrection}) in the proof of Lemma~\ref{prop:Peripheral} in Appendix~\ref{sec:ProofPeripheral}: the bound is a function of $\frac{t}{n}\|\mathcal{L}\|$ and $\frac{t^2}{n}\|\mathcal{L}\|^2$, say, $G(\frac{t}{n}\|\mathcal{L}\|,\frac{t^2}{n}\|\mathcal{L}\|^2)$ (apart from a correction which is exponentially small in $n$).
The correction $\mathcal{O}(1/n^2)$ in the exponent of the last expression in Eq.\ (\ref{eqn:ProofTh1Step1Exp}) can be bounded by using Lemma~\ref{lem:MatrixFuncBound} in Appendix~\ref{app:BCH}, which gives a bound as a function of $\frac{t}{n}\|\mathcal{L}\|$ again.
Collecting all these elements, the corrections $\mathcal{O}(1/n^2)$ in the exponent and $\mathcal{O}(1/n)$ in the last term in Eq.\ (\ref{eqn:ThmPeripheralMulti}) are bounded by a function of $\frac{t}{n}\|\mathcal{L}\|$ and $G(\frac{t}{n}\|\mathcal{L}\|,\frac{t^2}{n}\|\mathcal{L}\|^2)$, respectively.

To get Eq.\ (\ref{eqn:4.10}) from Eq.\ (\ref{eqn:ThmPeripheralMulti}) at Step 2, we apply the generalized BCH formula in Lemma~\ref{thm:kicktofieldP}.
It yields the correction $\mathcal{O}(1/n)$ in the exponent of Eq.\ (\ref{eqn:4.10}), which can be bounded by using Proposition~\ref{prop:BoundBCH} in Appendix~\ref{app:BCH}, and the bound is a function of $\frac{t}{n}\|\mathcal{L}\|$ multiplied by $n$.
Finally, at Step 3, Theorem~1 of Ref.\ \cite{ref:unity1} (generalized adiabatic theorem) is applied and it induces an additional correction, which can be bounded by $(M_1\frac{t}{n\Delta}\|\mathcal{L}\|+M_2\frac{t^2}{n}\|\mathcal{L}\|^2)F_1(\frac{t}{n}\|\mathcal{L}\|)\rme^{M_3t\|\mathcal{L}\|F_2(\frac{t}{n}\|\mathcal{L}\|)}$, with the spectral gap $\Delta=\min_{k\neq\ell}|a_k-a_\ell|$, some positive constants $M_i$ ($i=1,2,3$), and monotonically increasing functions $F_i(x)$ ($i=1,2$) which shrink to $1$ as $x\to0$.
This plus a bound $G(\frac{t}{n}\|\mathcal{L}\|,\frac{t^2}{n}\|\mathcal{L}\|^2)$ from the last contribution in Eq.\ (\ref{eqn:ThmPeripheralMulti}), which is $\mathcal{O}(1/n)$ for large $n$, gives the bound on the overall correction to the QZD in Eq.\ (\ref{eqn:CPTPZenomulti}) of Theorem~\ref{thm:CPTPBB}.
Certainly, this is  not a sharp bound, but it suffices to establish that the error is $\mathcal{O}(1/n)$ as $n\to+\infty$, for any finite $t$.
\end{remark}

\section{Examples}
\label{sec:Examples}
\subsection{QZD by Pulsed Weak Measurements}
\label{sec:ExampleWeakMeas}
Let us look at an example of the Zeno limit presented in Corollary~\ref{thm:kicktozeno} for the QZD via frequent applications of quantum operations, which is a particular case of Theorem~\ref{thm:CPTPBB}.
This corollary unifies and generalizes the QZDs via (i) frequent projective measurements and via (ii) frequent unitary kicks. 
Here, we provide a simple but analytically tractable model for the QZD via frequent weak measurements.

We consider a two-level system with a Hamiltonian
\begin{equation}
H=\frac{1}{2}\Omega Z,
\label{eqn:ExampleH}
\end{equation}
where 
$
Z=\ket{0}\bra{0}-\ket{1}\bra{1}
$.
We will also use
$
X=\ket{0}\bra{1}+\ket{1}\bra{0}
$ and $
Y=-\rmi(\ket{0}\bra{1}-\ket{1}\bra{0})
$ in the following.
During the unitary evolution $\rme^{-\rmi t\mathcal{H}}$ with $\mathcal{H}=[H,{}\bullet{}]$, we repeatedly perform weak nonselective measurement on $X$, whose action on the system state is described by the CPTP map
\begin{equation}
	\mathcal{E}=(1-p)1+p\mathcal{P},
	\label{eqn:WeakMeasEx}
\end{equation}
where
\begin{equation}
\mathcal{P}=P{}\bullet{}P+Q{}\bullet{}Q=\frac{1}{2}(1+X{}\bullet{}X)
\label{eqn:ExampleP}
\end{equation}
is a projection ($\mathcal{P}^2=\mathcal{P}$) with projection operators
\begin{equation}
P=\frac{I+X}{2},\qquad
Q=\frac{I-X}{2}.
\end{equation}
The parameter $p$ ranges from $0$ to $1$, and controls the strength of the measurement: for $p=1$ the map $\mathcal{E}$ describes the perfect projective measurement $\mathcal{P}$, while for $p=0$ it does nothing, gaining no information on $X$.
We focus on the case $p>0$ in the following.
The evolution of the system under the repeated measurements is described by $(\mathcal{E}\rme^{-\rmi\frac{t}{n}\mathcal{H}})^n$, and we are interested in how the dynamics is projected in the Zeno limit $n\to+\infty$ (QZD by this type of weak measurement is studied in Refs.\ \cite{ref:LidarZenoPRL,ref:LidarZenoJPA}).

It is possible to analyze this dynamics analytically.
Indeed, we can explicitly write down the spectral representation of the map,
\begin{equation}	
\Bigl(
\mathcal{E}
\rme^{-\rmi\frac{t}{n}\mathcal{H}}
\Bigr)^n
=\sum_{s,s'=\pm}
\lambda_{ss'}^n\mathcal{P}_{ss'},
\label{eqn:MapEx1}
\end{equation}
where
\begin{equation}
\lambda_{++}=
	1
,\qquad
\lambda_{+-}=
	1-p
,\qquad
\lambda_{-\pm}=
\left(1-\frac{p}{2}\right)\cos(\Omega t/n)
\pm\frac{p}{2}\eta
\end{equation}
are the eigenvalues and
\begin{equation}
\begin{cases}
\medskip
\displaystyle	
	\mathcal{P}_{+\pm}
	=\frac{1}{4}\,\Bigl(
	1+Z{}\bullet{}Z\pm[X(t/n){}\bullet{}X(t/n)+Y(t/n){}\bullet{}Y(t/n)]
	\Bigr),
	\\
\displaystyle	
	\mathcal{P}_{-\pm}
	=\frac{1}{4}
	\,\biggl\{
	1-Z{}\bullet{}Z
	\pm\frac{1}{\eta}
	\,\biggl[
	X(t/n){}\bullet{}X(t/n)-Y(t/n){}\bullet{}Y(t/n)
\\
\displaystyle
\qquad\qquad\qquad\qquad\qquad\qquad\qquad\qquad\qquad\ \ \,
{}-\rmi\left(
\frac{2}{p}-1
\right)
\sin(\Omega t/n)
[Z,{}\bullet{}]
	\biggr]
	\biggr\}
\end{cases}
\end{equation}
are the spectral projections, with
\begin{equation}
\eta
=
\sqrt{
1
-\left(
\frac{2}{p}-1
\right)^2
\sin^2(\Omega t/n)
}	
\end{equation}
and
\begin{equation}
\begin{cases}
\medskip
\displaystyle	
	X(t)=X\cos\frac{\Omega t}{2}
-Y\sin\frac{\Omega t}{2},\\
\displaystyle	
	Y(t)=Y\cos\frac{\Omega t}{2}
+X\sin\frac{\Omega t}{2}.
\end{cases}
\end{equation}
In the Zeno limit $n\to+\infty$, two of the eigenvalues of the map (\ref{eqn:MapEx1}) survive as $\lambda_{++}^n=1$ and $\lambda_{-+}^n=1+\mathcal{O}(1/n)$, while the others decay under the condition $p>0$.
Since $\mathcal{P}_{++}+\mathcal{P}_{-+}\to\frac{1}{2}(1+X{}\bullet{}X)=\mathcal{P}$, we get
\begin{equation}
\Bigl(
\mathcal{E}
\rme^{-\rmi\frac{t}{n}\mathcal{H}}
\Bigr)^n
\to\mathcal{P}\quad\text{as}\quad n\to+\infty.
\end{equation}
It is in accordance with Corollary~\ref{thm:kicktozeno}.
Indeed, the spectrum of $\mathcal{E}$ in Eq.\ (\ref{eqn:WeakMeasEx}) is given by $\{1,1-p\}$, with its peripheral projection being $\mathcal{P}$ in Eq.\ (\ref{eqn:ExampleP}).
The peripheral part of $\mathcal{E}$ is the projection $\mathcal{E}_\varphi=\mathcal{P}$, and the unitary generator $-\rmi\mathcal{H}$ with the Hamiltonian (\ref{eqn:ExampleH}) is projected to $\mathcal{L}_Z=-\rmi\mathcal{P} \mathcal{H}\mathcal{P}=0$.

\begin{figure}
\centering
\includegraphics[scale=0.75]{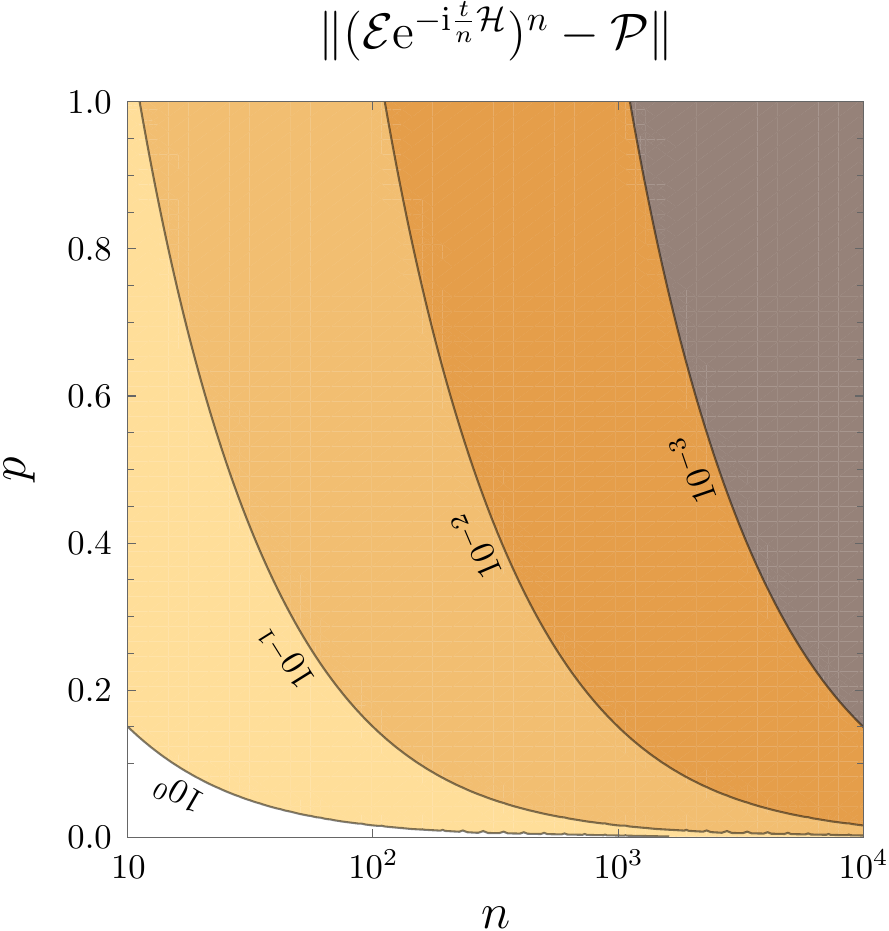}
\caption{Contour plot of $\|(\mathcal{E}\rme^{-\rmi\frac{t}{n}\mathcal{H}})^n-\mathcal{P}\|$ versus the number of measurements $n$ and the strength of the measurement $p$, for the model analyzed in Sec.~\ref{sec:ExampleWeakMeas} [the first term in Eq.\ (\ref{eqn:ExampleWeakMeasConv}) is actually plotted]. The parameter is set at $\Omega t=1$. The convergence to the QZD is faster with a stronger (larger $p$) measurement.}
\label{fig:ExampleWeakMeas}
\end{figure}
In more detail, it is possible to estimate how it converges to the limit:
\begin{equation}
\| (
\mathcal{E}
\rme^{-\rmi\frac{t}{n}\mathcal{H}}
)^n-\mathcal{P}\|
=
\frac{1}{2n}\Omega t
\left(\frac{2}{p}-1\right)
\left(
\sqrt{
\frac{1}{4}(\Omega t)^2+1
}
+
\sqrt{
\frac{1}{4}(\Omega t)^2+\left(\frac{2}{p}-1\right)^{-2}
}
\right)
+\mathcal{O}(1/n^2),
\label{eqn:ExampleWeakMeasConv}
\end{equation}
where we have chosen the operator norm defined in Eq.\ (\ref{eqn:OpNorm}) to estimate the distance.
For a large but finite $n$, the correction is $\mathcal{O}(1/n)$, as stated in Corollary~\ref{thm:kicktozeno}, and the correction depends on the chosen evolution time $t$ and the strength of the measurement $p$.

In Fig.~\ref{fig:ExampleWeakMeas}, the distance $\|(
\mathcal{E}
\rme^{-\rmi\frac{t}{n}\mathcal{H}}
)^n-\mathcal{P}\|$ in Eq.\ (\ref{eqn:ExampleWeakMeasConv}) is plotted as a function of the number of measurements $n$ and the strength of the measurement $p$.
The convergence to the QZD is faster with a stronger (larger $p$) measurement.

\subsection{QZD by CPTP Kicks with Persistent Oscillations}
\label{sec:ExampleCPTPkicks}
Let us look at another example of the Zeno limit presented in Corollary~\ref{thm:kicktozeno} for the QZD via frequent CPTP kicks.
Here, we provide a model in which two mechanisms work to induce the QZD: relaxation and persistent oscillations.
This situation was intractable by previously developed theories.

We consider a three-level system evolving with the GKLS generator
\begin{equation}
\mathcal{L}
=-\rmi[K,{}\bullet{}]-\frac{1}{2}(L^\dag L{}\bullet{}+{}\bullet{}L^\dag L-2L{}\bullet{}L^\dag)
\label{eqn:Ex2L}
\end{equation}
with
\begin{gather}
K
=
\Omega_0\ket{0}\bra{0}
+\Omega_1\ket{1}\bra{1}
+\Omega_2\ket{2}\bra{2}
=
\begin{pmatrix}
	\Omega_0&&\\
	&\Omega_1&\\
	&&\Omega_2
\end{pmatrix},\\
L
=
\sqrt{\Gamma}
\,\Bigl(
\ket{1}\bra{1}
+\ket{2}\bra{2}
\Bigr)
=
\sqrt{\Gamma}
\begin{pmatrix}
	0&&\\
	&1&\\
	&&1
\end{pmatrix}.
\label{eqn:Ex2G}
\end{gather}
During the evolution, we repeatedly kick the system by the CPTP map
\begin{equation}
\mathcal{E}
=
K_0{}\bullet{}K_0^\dag
+
K_1{}\bullet{}K_1^\dag
\end{equation}
with
\begin{gather}
K_0
=
\ket{0}\bra{1}
+
\ket{1}\bra{0}
+
\sqrt{q}\,\ket{2}\bra{2}
=\begin{pmatrix}
0&1&\\
1&0&\\
&&\sqrt{q}
\end{pmatrix},
\\
K_1
=\sqrt{1-q}\,\ket{0}\bra{2}
=
\begin{pmatrix}
0&0&\sqrt{1-q}\\
0&0&0\\
0&0&0
\end{pmatrix}.
\end{gather}
Namely, we look at the evolution $(\mathcal{E}\rme^{\frac{t}{n}\mathcal{L}})^n$ with large $n$.
The generator $\mathcal{L}$ describes pure dephasing between $\ket{0}$ and the rest.
On the other hand, the CPTP map induces transition from $\ket{2}$ to $\ket{0}$ with a rate $1-q$ ($0\le q<1$) and at the same time flips the system between $\ket{0}\leftrightarrow\ket{1}$.
We here restrict ourselves to the case $q<1$.
By repeatedly applying $\mathcal{E}$, the system relaxes to the subspace spanned by $\{\ket{0},\ket{1}\}$, where the system keeps on oscillating between $\ket{0}$ and $\ket{1}$.
In the Zeno limit $n\to+\infty$, the generator $\mathcal{L}$ is projected by the two mechanisms: by the relaxation from $\ket{2}$ to $\ket{0}$ and by the persistent oscillations between $\ket{0}$ and $\ket{1}$.
Notice that the two Kraus operators $K_0$ and $K_1$ do not commute and the two mechanisms act nontrivially.

\begin{figure}
\centering
\begin{tabular}{r@{\ }r@{\ }r@{\ }r}
\includegraphics[scale=0.405]{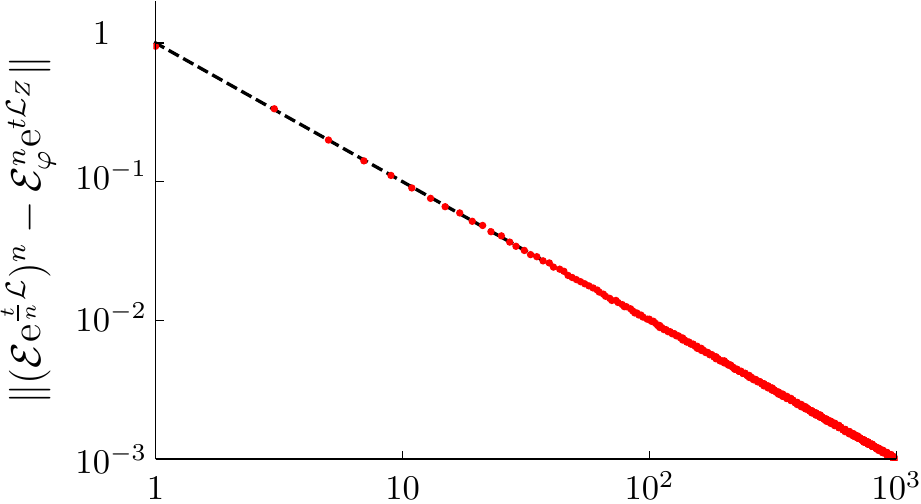}
&
\includegraphics[scale=0.405]{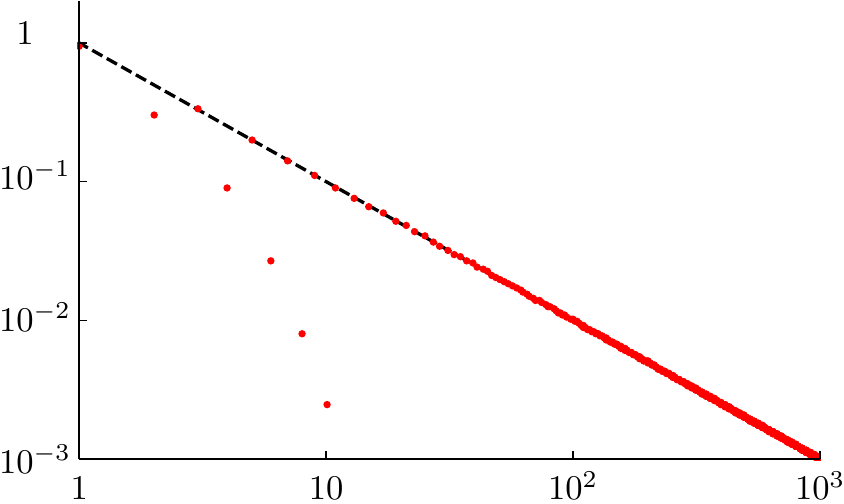}
&
\includegraphics[scale=0.405]{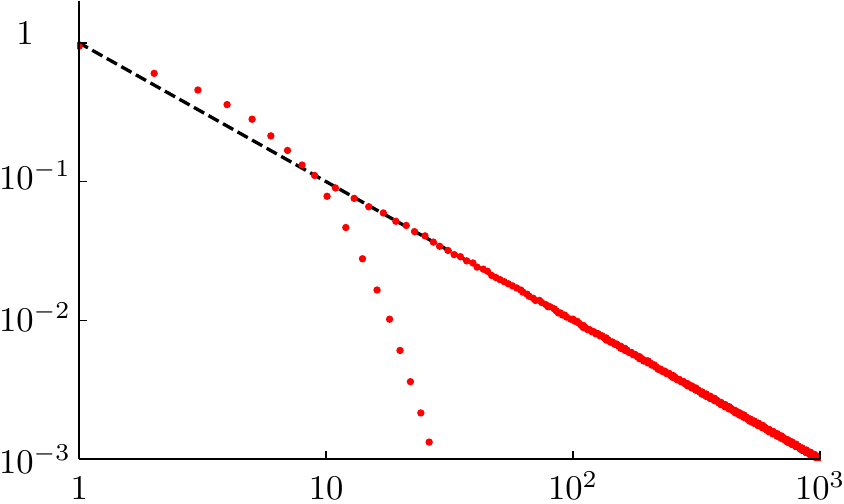}
&
\includegraphics[scale=0.405]{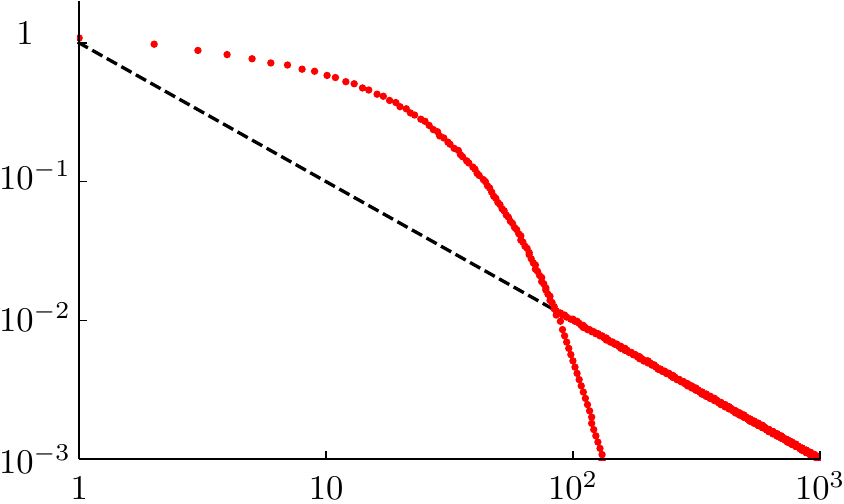}
\\[-24truemm]
\tiny$q=0.0$, $\Gamma t=0.0$
&
\tiny$q=0.3$, $\Gamma t=0.0$
&
\tiny$q=0.6$, $\Gamma t=0.0$
&
\tiny$q=0.9$, $\Gamma t=0.0$
\\[-2.9truemm]
&
\makebox(98.5,59){}
&
&
\\[2.2truemm]
\includegraphics[scale=0.405]{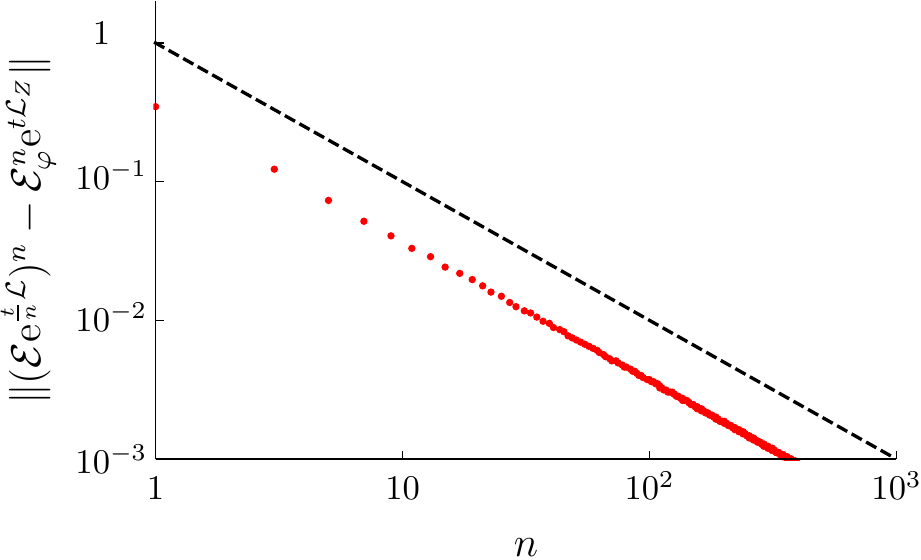}
&
\includegraphics[scale=0.405]{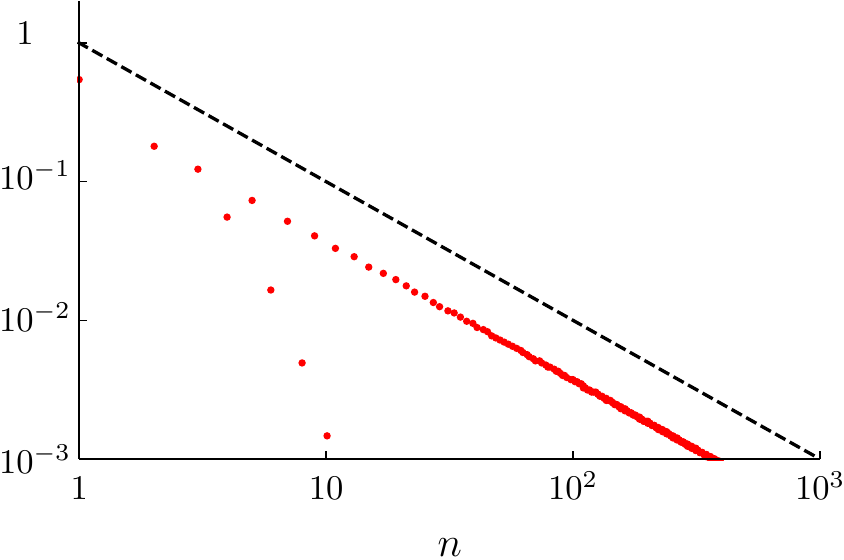}
&
\includegraphics[scale=0.405]{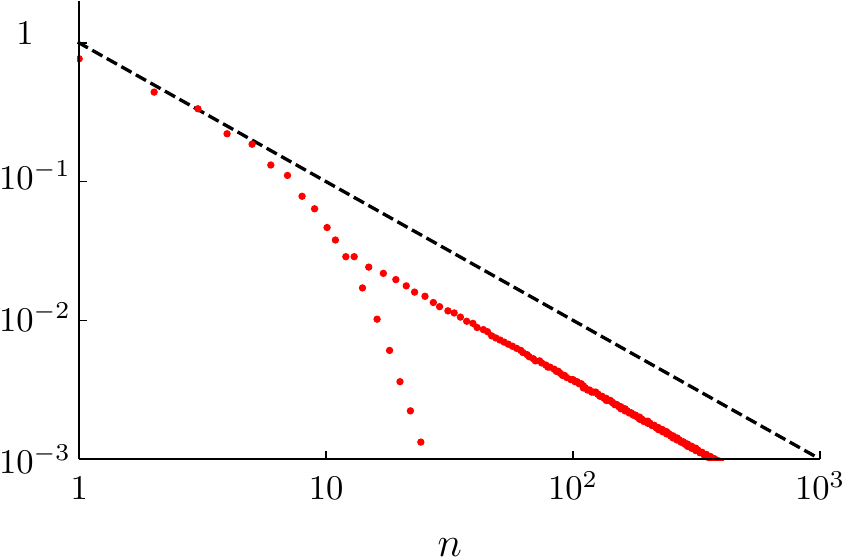}
&
\includegraphics[scale=0.405]{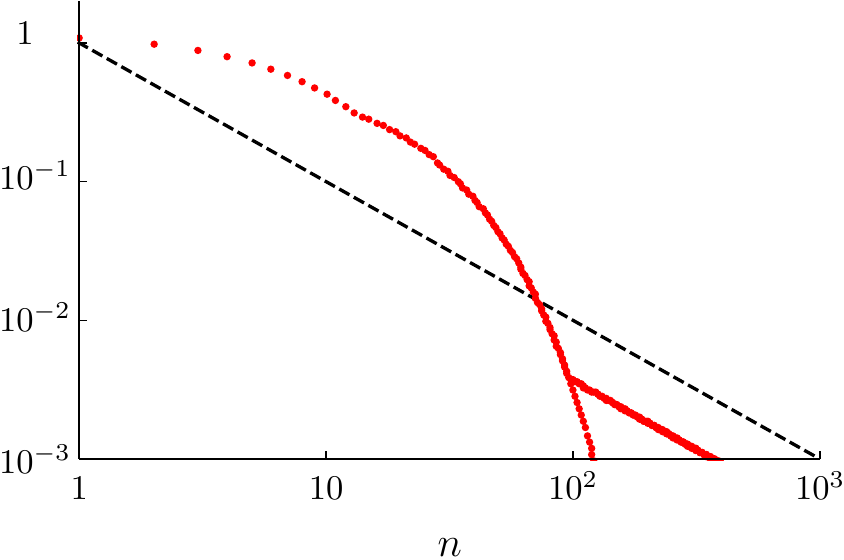}
\\[-26.2truemm]
\tiny$q=0.0$, $\Gamma t=2.0$
&
\tiny$q=0.3$, $\Gamma t=2.0$
&
\tiny$q=0.6$, $\Gamma t=2.0$
&
\tiny$q=0.9$, $\Gamma t=2.0$
\\[-2.9truemm]
&
\makebox(98.5,65.5){}
&
&
\end{tabular}
\caption{The convergence to the QZD via the repeated CPTP kicks in the model analyzed in Sec.~\ref{sec:ExampleCPTPkicks}. The parameters other than $q$ and $\Gamma t$ are set at $\Omega_0t=0.0$, $\Omega_1t=1.0$, and $\Omega_2t=2.0$. The dashed lines indicate $1/n$. We have chosen the operator norm defined in Eq.\ (\ref{eqn:OpNorm}) to estimate the distance. There appear to be two sequences in each panel: one is for odd $n$ and the other for even $n$. The former asymptotically decays as $\mathcal{O}(1/n)$, while the latter decays faster.}
\label{fig:ZenoLimitCPTP}
\end{figure}
Let us first see how the system evolves by the repeated applications of $\mathcal{E}$.
Applying $\mathcal{E}$ repeatedly $n$ times results in 
\begin{align}
\mathcal{E}^n
=\begin{cases}
\medskip
\displaystyle	
K_0^n{}\bullet{}K_0^n
+
\frac{1-q^n}{1+q}
\,\Bigl(
q\ket{0}\bra{0}
+\ket{1}\bra{1}
\Bigr)\,\bra{2}{}\bullet{}\ket{2}
&(n\ \text{even}),
\\
\displaystyle	
K_0^n{}\bullet{}K_0^n
+
\frac{1}{1+q}\,\Bigl(
(1-q^{n+1})\ket{0}\bra{0}
+q(1-q^{n-1})\ket{1}\bra{1}
\Bigr)\,\bra{2}{}\bullet{}\ket{2}
&(n\ \text{odd}).
\end{cases}
\end{align}
As $n$ increases, it asymptotically behaves as
\begin{equation}
\mathcal{E}^n
\sim
\mathcal{U}_\infty^n\mathcal{P}_\varphi,
\end{equation}
with the asymptotic unitary $\mathcal{U}_\infty=(X+\ket{2}\bra{2}){}\bullet{}(X+\ket{2}\bra{2})$ and the projection 
\begin{equation}
\mathcal{P}_\varphi
=
P{}\bullet{}P
+\frac{1}{2}
\left(
P-\frac{1-q}{1+q}Z
\right)\bra{2}{}\bullet{}\ket{2}
\end{equation}
onto the peripheral spectrum of $\mathcal{E}$, where $
P=\ket{0}\bra{0}+\ket{1}\bra{1}
$, $
X=\ket{0}\bra{1}+\ket{1}\bra{0}
$, $
Y=-\rmi(\ket{0}\bra{1}-\ket{1}\bra{0})
$, and $
Z=\ket{0}\bra{0}-\ket{1}\bra{1}
$.
The peripheral spectrum of $\mathcal{E}$ consists of two peripheral eigenvalues, $\lambda_0=1$ and $\lambda_1=-1$, with the corresponding spectral projections given by
\begin{equation}
\begin{cases}
\medskip
\displaystyle
\mathcal{P}_0
=
\frac{1}{2}
\,\Bigl(
P\tr({}\bullet{})
+
X
\tr(X{}\bullet{})
\Bigr),\\
\displaystyle
\mathcal{P}_1
=
\frac{1}{2}
\left(
Y
\tr(
Y
{}\bullet{}
)
+
Z
\tr
(
Z
{}\bullet{}
)
-
\frac{1-q}{1+q}
Z
\bra{2}{}\bullet{}\ket{2}
\right).
\end{cases}
\end{equation}
Then, according to Corollary~\ref{thm:kicktozeno}, the generator $\mathcal{L}$ is projected to
\begin{equation}
	\mathcal{L}_Z
=\sum_{k=0,1}\mathcal{P}_k\mathcal{L}\mathcal{P}_k
=
-\frac{1}{2}\Gamma
\,\Bigl(
\ket{0}\bra{0}
{}\bullet{}
\ket{1}\bra{1}
+
\ket{1}\bra{1}
{}\bullet{}
\ket{0}\bra{0}
\Bigr)
\end{equation}
in the Zeno limit $n\to+\infty$, with
\begin{equation}
\mathcal{E}_\varphi
=\mathcal{U}_\infty
\mathcal{P}_\varphi
=
X{}\bullet{}X
+\frac{1}{2}
\left(
P+\frac{1-q}{1+q}Z
\right)\bra{2}{}\bullet{}\ket{2}.
\end{equation}
See Fig.~\ref{fig:ZenoLimitCPTP}, where the convergence to the QZD is numerically demonstrated for several sets of parameters.

\subsection{QZD by Cycles of Multiple CPTP Kicks}
\label{sec:ExampleBBkicks}
In the previous subsections, we have provided two examples to illustrate Corollary~\ref{thm:kicktozeno}, i.e.\ the QZD by repeating the same quantum operation.
Here, we provide an example that allows us to display Theorem~\ref{thm:CPTPBB}, i.e.\ the QZD by cycles of two different kicks $\mathcal{E}_1$ and $\mathcal{E}_2$.

We consider a three-level system evolving with a GKLS generator $\mathcal{L}$, and being kicked  alternately by $\mathcal{E}_1$ and $\mathcal{E}_2$ defined by
\begin{equation}
\mathcal{E}_j
=
K_0^{(j)}{}\bullet{}K_0^{(j)\dag}
+
K_1^{(j)}{}\bullet{}K_1^{(j)\dag}
\quad(j=1,2),
\end{equation}
with
\begin{gather}
K_0^{(1)}
=
Z
+
\sqrt{q}\,\ket{2}\bra{2}
=\begin{pmatrix}
1&0&\\
0&-1&\\
&&\sqrt{q}
\end{pmatrix},
\\
K_0^{(2)}
=
Y
+
\sqrt{q}\,\ket{2}\bra{2}
=\begin{pmatrix}
0&-\rmi&\\
\rmi&0&\\
&&\sqrt{q}
\end{pmatrix},
\\
K_1^{(1)}
=
K_1^{(2)}
=\sqrt{1-q}\,\ket{0}\bra{2}
=
\begin{pmatrix}
0&0&\sqrt{1-q}\\
0&0&0\\
0&0&0
\end{pmatrix}.
\end{gather}
We again use $
P=\ket{0}\bra{0}+\ket{1}\bra{1}
$, $
X=\ket{0}\bra{1}+\ket{1}\bra{0}
$, $
Y=-\rmi(\ket{0}\bra{1}-\ket{1}\bra{0})
$, and $
Z=\ket{0}\bra{0}-\ket{1}\bra{1}
$.
The kicks $\mathcal{E}_1$ and $\mathcal{E}_2$ rotate the system within the subspace spanned by $\{\ket{0},\ket{1}\}$ around different axes, and at the same time induce decay from $\ket{2}$ to $\ket{0}$ with a rate $1-q$ ($0\le q<1$).
We here restrict ourselves to the case $q<1$.
The kicked evolution is described by $(\mathcal{E}_2\rme^{\frac{t}{2n}\mathcal{L}}\mathcal{E}_1\rme^{\frac{t}{2n}\mathcal{L}})^n$, and we are interested in its Zeno limit $n\to+\infty$.

For the Zeno limit stated in Theorem~\ref{thm:CPTPBB}, the peripheral spectrum of $\mathcal{E}=\mathcal{E}_2\mathcal{E}_1$ matters.
The $n$th power of $\mathcal{E}$ reads
\begin{align}
\mathcal{E}^n
=\begin{cases}
\medskip
\displaystyle	
K_0^n{}\bullet{}K_0^n
+
\frac{1-q^{2n}}{1+q^2}\,\Bigl(
(1-q+q^2)\ket{0}\bra{0}
+q\ket{1}\bra{1}
\Bigr)\,\bra{2}{}\bullet{}\ket{2}
&(n\ \text{even}),
\\
\displaystyle	
K_0^n{}\bullet{}K_0^n
+
\frac{1}{1+q^2}\,\Bigl(
[q-(1-q+q^2)q^{2n}]\ket{0}\bra{0}
\\
\qquad\qquad\qquad\qquad\qquad
{}+(1-q+q^2-q^{2n+1})\ket{1}\bra{1}
\Bigr)\,\bra{2}{}\bullet{}\ket{2}
&(n\ \text{odd}),
\end{cases}
\end{align}
where
\begin{equation}
K_0=K_0^{(2)}K_0^{(1)}=\rmi X+q\ket{2}\bra{2}.
\end{equation}
As $n$ increases, the system relaxes to the subspace $\{\ket{0},\ket{1}\}$, where the system keeps on oscillating between $\ket{0}$ and $\ket{1}$, and it asymptotically behaves as
\begin{equation}
\mathcal{E}^n
\sim
\mathcal{U}_\infty^n\mathcal{P}_\varphi,
\end{equation}
with the asymptotic unitary $\mathcal{U}_\infty=(X+\ket{2}\bra{2}){}\bullet{}(X+\ket{2}\bra{2})$ and the projection 
\begin{equation}
\mathcal{P}_\varphi
=
P{}\bullet{}P
+\frac{1}{2}
\left(
P+\frac{(1-q)^2}{1+q^2}Z
\right)\bra{2}{}\bullet{}\ket{2}
\end{equation}
onto the peripheral spectrum of $\mathcal{E}$.
The peripheral spectrum of $\mathcal{E}$ consists of two peripheral eigenvalues, $\lambda_0=1$ and $\lambda_1=-1$, with the corresponding spectral projections given by
\begin{equation}
\begin{cases}
\medskip
\displaystyle
\mathcal{P}_0
=
\frac{1}{2}
\,\Bigl(
P\tr({}\bullet{})
+
X
\tr(X{}\bullet{})
\Bigr),\\
\displaystyle
\mathcal{P}_1
=
\frac{1}{2}
\left(
Y
\tr(
Y
{}\bullet{}
)
+
Z
\tr
(
Z
{}\bullet{}
)
+
\frac{(1-q)^2}{1+q^2}
Z
\bra{2}{}\bullet{}\ket{2}
\right).
\end{cases}
\end{equation}
In this way, in this example, we have two different kicks and nontrivial peripheral eigenvalues.
Then, according to Theorem~\ref{thm:CPTPBB}, the generator $\mathcal{L}$ is projected to $\mathcal{L}_Z$ as Eq.\ (\ref{eq:overLdef}) in the Zeno limit $n\to+\infty$.
For the GKLS generator $\mathcal{L}$ in Eqs.\ (\ref{eqn:Ex2L})--(\ref{eqn:Ex2G}) considered in the previous subsection, it reads
\begin{equation}
	\mathcal{L}_Z
=\frac{1}{2}\sum_{k=0,1}\mathcal{P}_k(\mathcal{L}+\mathcal{E}_\varphi^{-1}\mathcal{E}_2\mathcal{L}\mathcal{E}_1)\mathcal{P}_k
=-\frac{1}{2}\Gamma
\,\Bigl(
\ket{0}\bra{0}
{}\bullet{}
\ket{1}\bra{1}
+
\ket{1}\bra{1}
{}\bullet{}
\ket{0}\bra{0}
\Bigr),
\end{equation}
where
\begin{equation}
\mathcal{E}_\varphi^{-1}
=\mathcal{E}_\varphi
=\mathcal{E}\mathcal{P}_\varphi
=\mathcal{U}_\infty
\mathcal{P}_\varphi
=
X{}\bullet{}X
+\frac{1}{2}
\left(
P-\frac{(1-q)^2}{1+q^2}Z
\right)\bra{2}{}\bullet{}\ket{2}.
\end{equation}
See Fig.~\ref{fig:ZenoLimitBB}, where the convergence to the QZD is numerically demonstrated for several sets of parameters.
\begin{figure}
\centering
\begin{tabular}{r@{\ }r@{\ }r@{\ }r}
\includegraphics[scale=0.405]{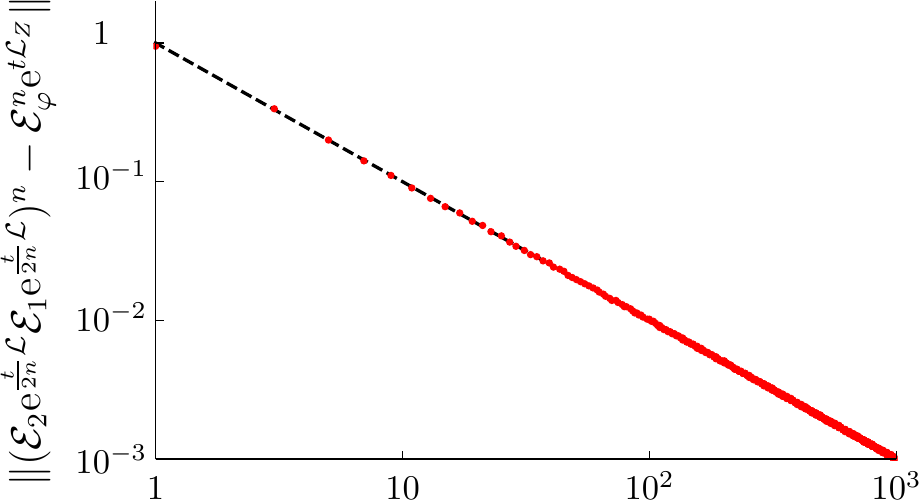}
&
\includegraphics[scale=0.405]{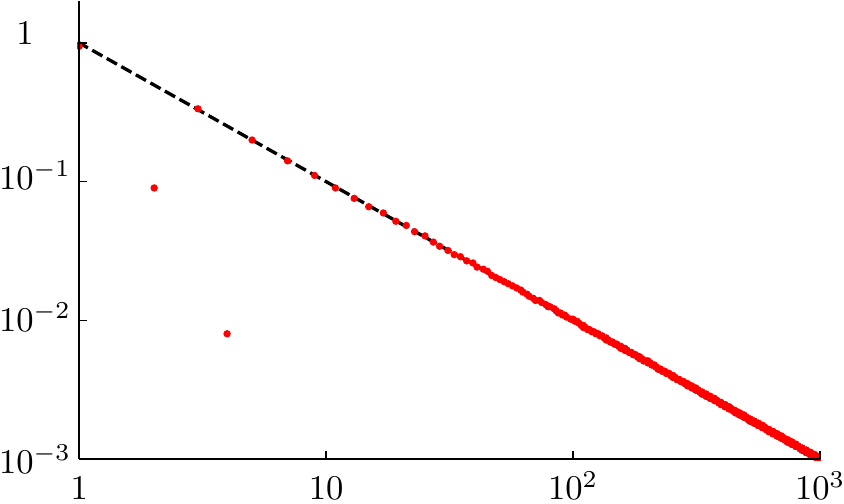}
&
\includegraphics[scale=0.405]{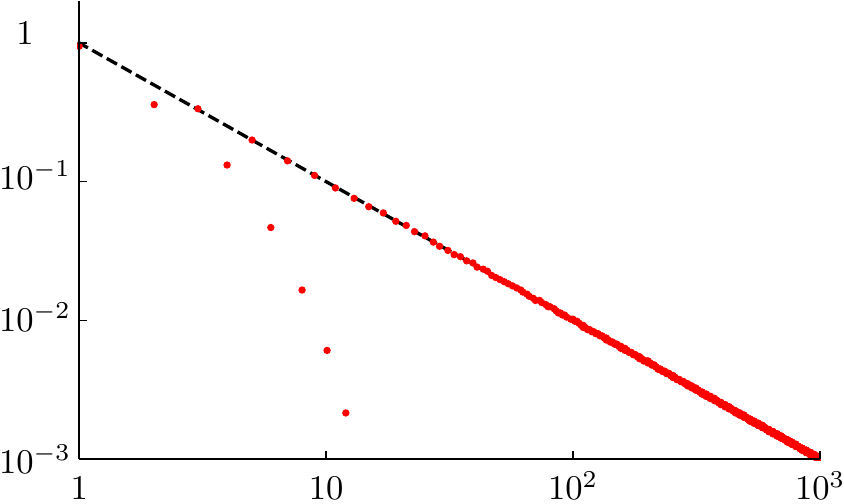}
&
\includegraphics[scale=0.405]{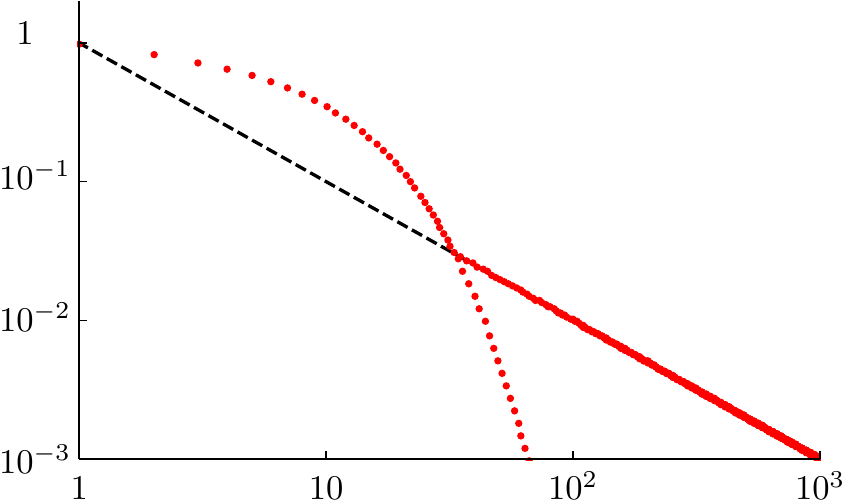}
\\[-24truemm]
\tiny$q=0.0$, $\Gamma t=0.0$
&
\tiny$q=0.3$, $\Gamma t=0.0$
&
\tiny$q=0.6$, $\Gamma t=0.0$
&
\tiny$q=0.9$, $\Gamma t=0.0$
\\[-2.9truemm]
&
\makebox(98.5,59){}
&
&
\\[2.2truemm]
\includegraphics[scale=0.405]{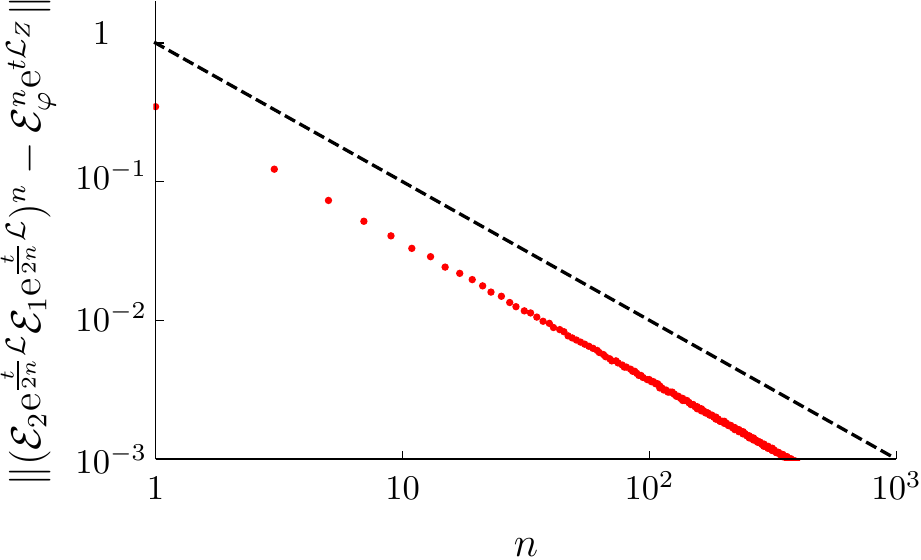}
&
\includegraphics[scale=0.405]{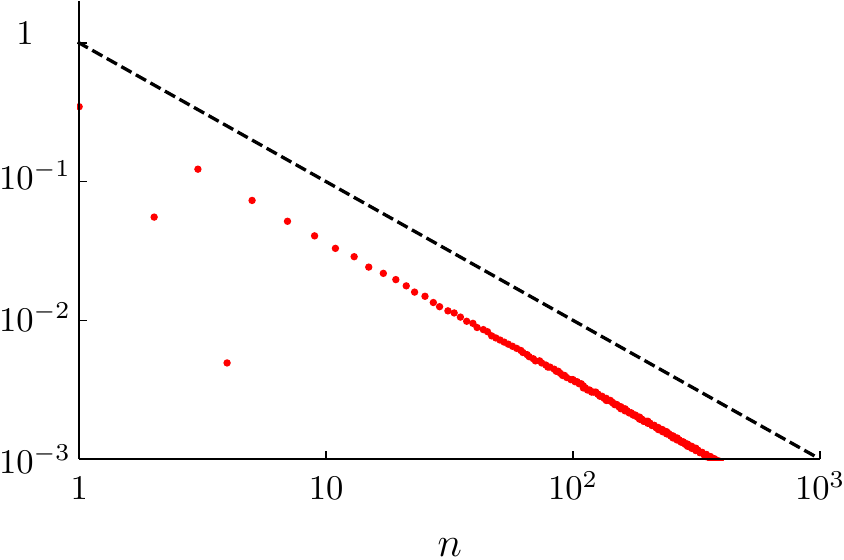}
&
\includegraphics[scale=0.405]{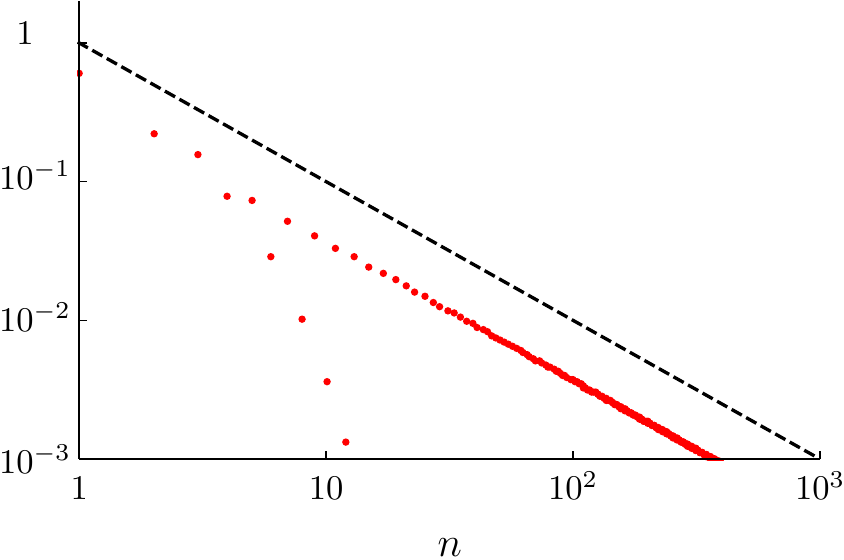}
&
\includegraphics[scale=0.405]{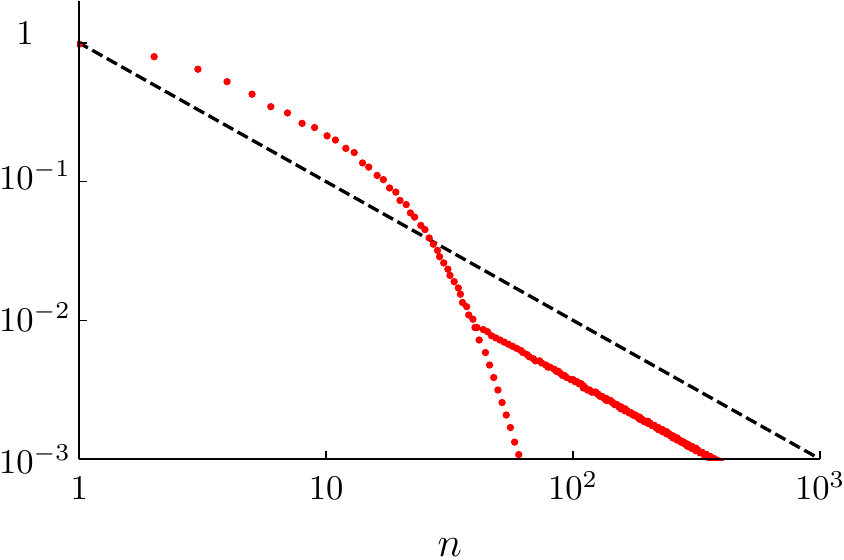}
\\[-26.2truemm]
\tiny$q=0.0$, $\Gamma t=2.0$
&
\tiny$q=0.3$, $\Gamma t=2.0$
&
\tiny$q=0.6$, $\Gamma t=2.0$
&
\tiny$q=0.9$, $\Gamma t=2.0$
\\[-2.9truemm]
&
\makebox(98.5,65.5){}
&
&
\end{tabular}
\caption{The convergence to the QZD by  alternating the kicks $\mathcal{E}_1$ and $\mathcal{E}_2$ in the model analyzed in Sec.~\ref{sec:ExampleBBkicks}. The evolution $\mathcal{L}$ to be projected is the same as the one considered in Fig.~\ref{fig:ZenoLimitCPTP}. The parameters other than $q$ and $\Gamma t$ are set at $\Omega_0t=0.0$, $\Omega_1t=1.0$, and $\Omega_2t=2.0$. The dashed lines are the plots of $1/n$. We have chosen the operator norm defined in Eq.\ (\ref{eqn:OpNorm}) to estimate the distance. There appear to be two sequences in each panel: one is for odd $n$ and the other for even $n$. The former asymptotically decays as $\mathcal{O}(1/n)$, while the latter decays faster.}
\label{fig:ZenoLimitBB}
\end{figure}

\subsection{QZD by Multiple Projective Measurements}
\label{sec:MultiP}
In Corollary~\ref{thm:MeasBB}, we presented the QZD via cycles of 
multiple selective projective measurements.
It is a generalization of the standard QZD via frequent repetitions of a projective measurement, and is a variant of the QZD via cycles of 
multiple CPTP kicks proved in Theorem~\ref{thm:CPTPBB}.
Here we provide simple examples for Corollary~\ref{thm:MeasBB} and Theorem~\ref{thm:CPTPBB}.

Let us consider a three-level system with a Hamiltonian $H$,
and two different projective measurements: one is characterized by a pair of Hermitian projection operators
\begin{equation}
P_1
=\ket{1}\bra{1}+\ket{2}\bra{2}
=\begin{pmatrix}	
0&&\\
&1&\\
&&1
\end{pmatrix},\qquad
Q_1
=
\begin{pmatrix}	
1&&\\
&0&\\
&&0
\end{pmatrix}
,
\end{equation}
and the other by another pair of Hermitian projection operators
\begin{equation}
	P_2=\frac{\ket{0}+\ket{1}}{\sqrt{2}}\frac{\bra{0}+\bra{1}}{\sqrt{2}}+\ket{2}\bra{2}
	=\begin{pmatrix}	
	1/2&1/2&\\
	1/2&1/2&\\
	&&1
	\end{pmatrix},
\quad
	Q_2
	=\begin{pmatrix}	
	1/2&-1/2&\\
	-1/2&1/2&\\
	&&0
	\end{pmatrix}.
\end{equation}

If we concatenate the two selective projective measurements $P_1$ and $P_2$, we get
\begin{equation}
	P_2P_1
	=\frac{1}{2}\,\Bigl(\ket{0}+\ket{1}\Bigr)\,\bra{1}+\ket{2}\bra{2}
	=
\begin{pmatrix}
0&1/2&\\
0&1/2&\\
&&1
\end{pmatrix}.
\end{equation}
This admits three eigenvalues $1$, $1/2$, and $0$, and its peripheral part is given by
\begin{equation}
	P_\varphi
	=\ket{2}\bra{2}
	=\begin{pmatrix}
	0&&\\
	&0&\\
	&&1
	\end{pmatrix},
\end{equation}
which is Hermitian (although $P_2P_1$ is not) and is the simultaneous eigenprojection of $P_1$ and $P_2$ belonging to the unit eigenvalue $1$.
This demonstrates Lemma~\ref{lem:P1P2} used in Corollary~\ref{thm:MeasBB}.
According to Corollary~\ref{thm:MeasBB}, the bang-bang sequence of the selective measurements $P_1$ and $P_2$ during the unitary evolution $\rme^{-\rmi tH}$ projects the dynamics to 
\begin{equation}
  (P_2\rme^{-\rmi\frac{t}{2n}H}P_1\rme^{-\rmi\frac{t}{2n}H})^n
  \to P_\varphi\rme^{-\rmi tP_\varphi HP_\varphi}
  \quad\text{as}\quad
  n\to+\infty.
  \label{eqn:ExampleMeasBB}
\end{equation}
The system is confined in the one-dimensional space $\ket{2}$, and the Zeno Hamiltonian $H_Z=P_\varphi HP_\varphi\propto\ket{2}\bra{2}$ yields just a phase as time goes on.
See Fig.~\ref{fig:ZenoLimitMultiP}(a), where this convergence is numerically demonstrated for the Hamiltonian
\begin{equation}
H=g\,\Bigl(\ket{0}\bra{1}+\ket{1}\bra{0}+\ket{1}\bra{2}+\ket{2}\bra{1}\Bigr)
	=
	\begin{pmatrix}	
	0&g&0\\
	g&0&g\\
	0&g&0
	\end{pmatrix}.
	\label{eqn:Ex4H}
\end{equation}

If we do not collect any outcomes of the measurements, the bang-bang sequence in Eq.\ (\ref{eqn:ExampleMeasBB}) is modified to $(\tilde{\mathcal{P}}_2\rme^{-\rmi\frac{t}{2n}\mathcal{H}}\tilde{\mathcal{P}}_1\rme^{-\rmi\frac{t}{2n}\mathcal{H}})^n$ with CPTP projections
\begin{equation}
\tilde{\mathcal{P}}_j=P_j{}\bullet{}P_j+Q_j{}\bullet{}Q_j\quad(j=1,2)
\end{equation}
representing the nonselective measurements, and $\mathcal{H}=[H,{}\bullet{}]$.
In this case, we get
\begin{equation}
\tilde{\mathcal{P}}_2\tilde{\mathcal{P}}_1
=\frac{1}{2}P\tr(P{}\bullet{})
	+\ket{2}\bra{2}{}\bullet{}\ket{2}\bra{2}
+\frac{1}{2}\,\Bigl(\ket{0}+\ket{1}\Bigr)\,\bra{1}{}\bullet{}\ket{2}\bra{2}
+\frac{1}{2}\ket{2}\bra{2}{}\bullet{}\ket{1}\,\Bigl(\bra{0}+\bra{1}\Bigr),
\end{equation}
where $
P=\ket{0}\bra{0}+\ket{1}\bra{1}
$.
The spectrum of $\tilde{\mathcal{P}}_2\tilde{\mathcal{P}}_1$ consists of the eigenvalues $1$, $1/2$, and $0$, with the projection onto the peripheral spectrum given by
\begin{equation}
	\tilde{\mathcal{P}}_\varphi
	=\frac{1}{2}P\tr(P{}\bullet{})
	+\ket{2}\bra{2}{}\bullet{}\ket{2}\bra{2}.
\end{equation}
According to Corollary~\ref{thm:MeasBB} or Theorem~\ref{thm:CPTPBB}, the bang-bang sequence of the nonselective measurements $\tilde{\mathcal{P}}_1$ and $\tilde
{\mathcal{P}}_2$ during the unitary evolution $\rme^{-\rmi t\mathcal{H}}$ projects the dynamics to 
\begin{equation}
  (\tilde{\mathcal{P}}_2\rme^{-\rmi\frac{t}{2n}\mathcal{H}}\tilde{\mathcal{P}}_1\rme^{-\rmi\frac{t}{2n}\mathcal{H}})^n
  \to\tilde{\mathcal{P}}_\varphi\rme^{-\rmi t\tilde{\mathcal{H}}_Z}
  \quad\text{as}\quad
  n\to+\infty,
  \label{eqn:ExampleCPTPBB}
\end{equation}
where
\begin{equation}
\tilde{\mathcal{H}}_Z=\tilde{\mathcal{P}}_\varphi\mathcal{H}\tilde{\mathcal{P}}_\varphi=0
\end{equation}
for any Hamiltonian $H$.
The Hilbert space is split into three subspaces $\{\ket{0}\}$, $\{\ket{1}\}$, and $\{\ket{2}\}$.
It is essentially the same as the QZD in Eq.\ (\ref{eqn:ExampleMeasBB}) by the selective measurements, concerning the subspace $\{\ket{2}\}$.
\begin{figure}
\centering
\begin{tabular}{r@{\quad}r@{\quad}r}
\makebox(130,10)[lt]{\footnotesize(a)}
&
\makebox(130,10)[lt]{\footnotesize(b)}
&
\makebox(130,10)[lt]{\footnotesize(c)}
\\
\makebox(130,71)[t]{\includegraphics[scale=0.49]{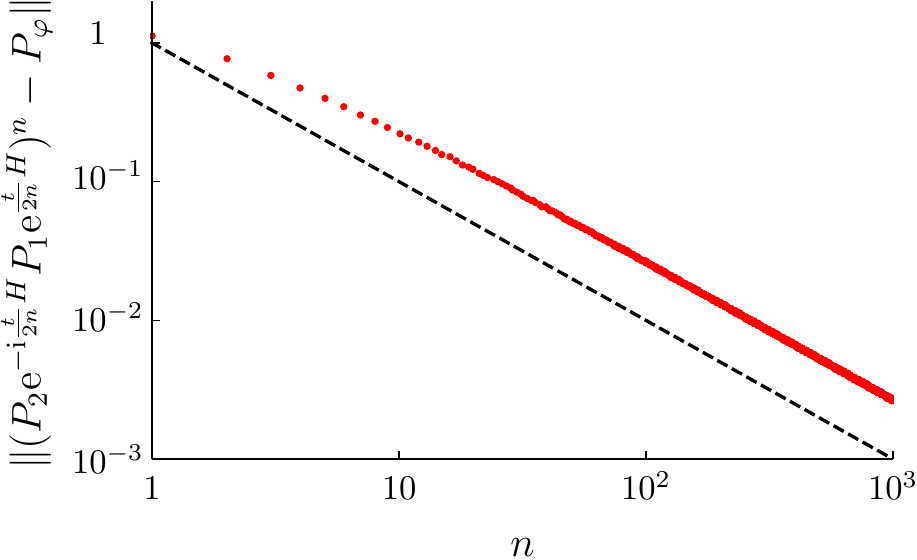}}
&
\makebox(130,71)[t]{\includegraphics[scale=0.49]{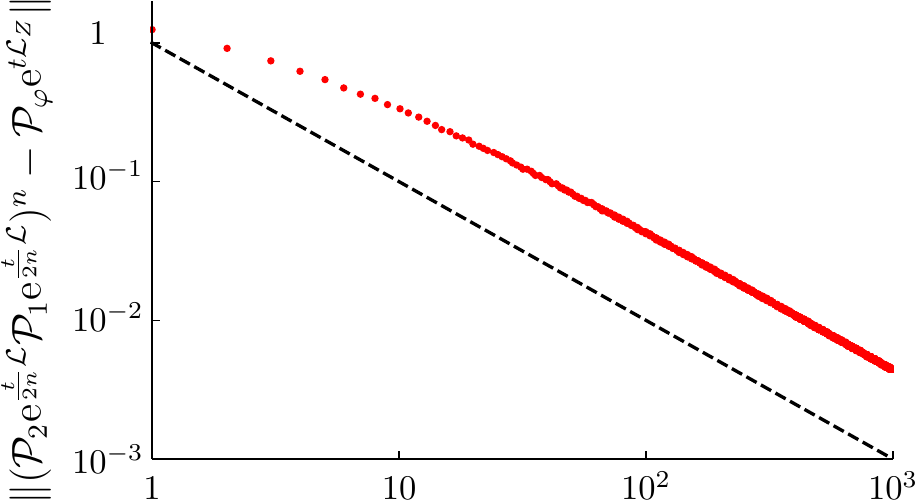}}
&
\makebox(130,71)[t]{\includegraphics[scale=0.49]{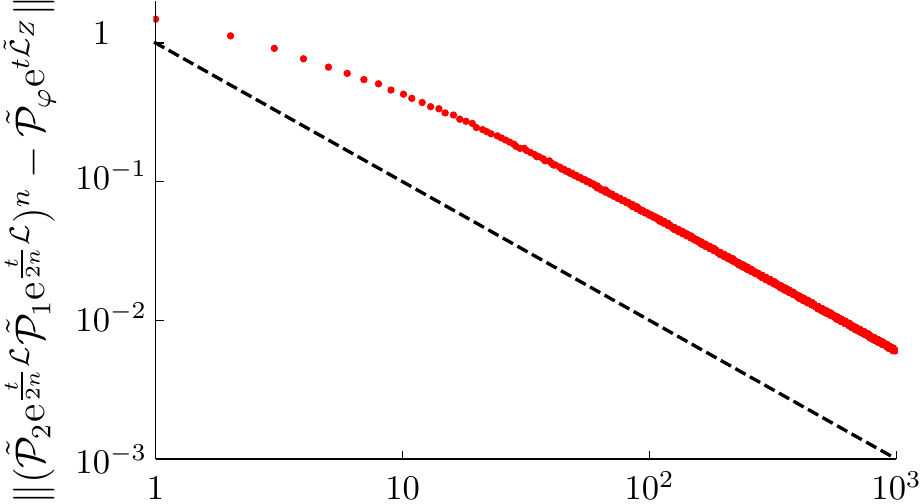}}
\\[-28.3truemm]
&
\tiny$\Gamma t=0.0$
&
\tiny$\Gamma t=0.0$
\\[-2.9truemm]
&
\makebox(130,71){}
&
\\[2.2truemm]
&
\makebox(130,79)[t]{\includegraphics[scale=0.49]{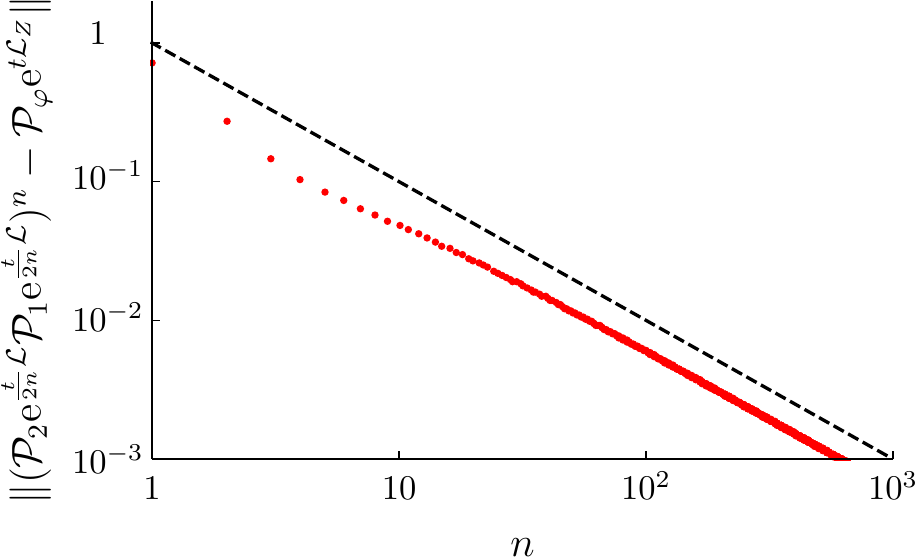}}
&
\makebox(130,79)[t]{\includegraphics[scale=0.49]{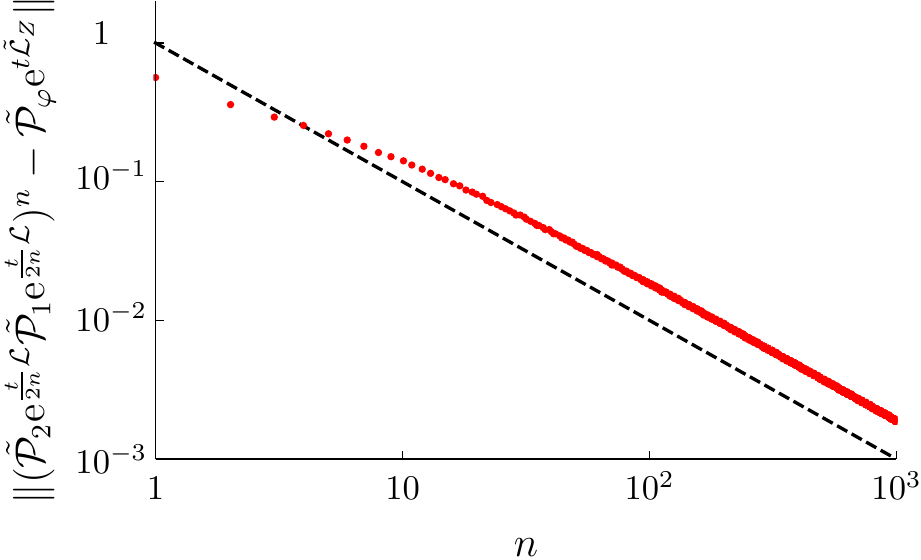}}
\\[-31.1truemm]
&
\tiny$\Gamma t=2.0$
&
\tiny$\Gamma t=2.0$
\\[-2.9truemm]
&
\makebox(130,79){}
&
\end{tabular}
\caption{The convergence to the QZD in the model analyzed in Sec.~\ref{sec:MultiP}: (a) via alternating the selective measurements $P_1$ and $P_2$ during the unitary evolution by the Hamiltonian $H$ in Eq.\ (\ref{eqn:Ex4H}); (b) via alternating the selective measurements $\mathcal{P}_1$ and $\mathcal{P}_2$ and (c) via alternating the nonselective measurements $\tilde{\mathcal{P}}_1$ and $\tilde{\mathcal{P}}_2$, during the evolution by the GKLS generator $\mathcal{L}$ in Eq.\ (\ref{eqn:Ex4L})--(\ref{eqn:Ex4G}) with the same Hamiltonian $H$ as in (a). The parameter other than $\Gamma t$ is set at $gt=1.0$. The dashed lines indicate $1/n$. We have chosen the operator norm to estimate the distance.}
\label{fig:ZenoLimitMultiP}
\end{figure}

For the nonunitary evolution
\begin{equation}
\mathcal{L}
=-\rmi[H,{}\bullet{}]-\frac{1}{2}(L^\dag L{}\bullet{}+{}\bullet{}L^\dag L-2L{}\bullet{}L^\dag)
\label{eqn:Ex4L}
\end{equation}
with
\begin{equation}
L
=
\sqrt{\Gamma}\,
\ket{1}\bra{2}
=
\sqrt{\Gamma}
\begin{pmatrix}
	0&0&0\\
	0&0&1\\
	0&0&0
\end{pmatrix},
\label{eqn:Ex4G}
\end{equation}
alternating selective measurements
\begin{equation}
\mathcal{P}_j=P_j{}\bullet{}P_j\quad(j=1,2)
\end{equation}
projects the dynamics as
\begin{equation}
  (\mathcal{P}_2\rme^{\frac{t}{2n}\mathcal{L}}\mathcal{P}_1\rme^{\frac{t}{2n}\mathcal{L}})^n
  \to\mathcal{P}_\varphi\rme^{t\mathcal{L}_Z}
  \quad\text{as}\quad
  n\to+\infty,
\label{eqn:ExampleSelectiveBB}
\end{equation}
where
\begin{equation}
\mathcal{L}_Z
=\mathcal{P}_\varphi\mathcal{L}\mathcal{P}_\varphi
=-\Gamma\ket{2}\bra{2}{}\bullet{}\ket{2}\bra{2},\qquad
\mathcal{P}_\varphi=P_\varphi{}\bullet{}P_\varphi,
\end{equation}
while alternating the nonselective measurements $\tilde{\mathcal{P}}_1$ and $\tilde{\mathcal{P}}_2$ projects the dynamics as
\begin{equation}
  (\tilde{\mathcal{P}}_2\rme^{\frac{t}{2n}\mathcal{L}}\tilde{\mathcal{P}}_1\rme^{\frac{t}{2n}\mathcal{L}})^n
  \to\tilde{\mathcal{P}}_\varphi\rme^{t\tilde{\mathcal{L}}_Z}
  \quad\text{as}\quad
  n\to+\infty,
\label{eqn:ExampleNonselectiveBB}
\end{equation}
where
\begin{equation}
\tilde{\mathcal{L}}_Z
=\tilde{\mathcal{P}}_\varphi\mathcal{L}\tilde{\mathcal{P}}_\varphi
=
-
\Gamma
\left(
\ket{2}\bra{2}
-\frac{1}{2}P
\right)
\bra{2}{}\bullet{}\ket{2}.
\end{equation}

See Figs.~\ref{fig:ZenoLimitMultiP}(b) and (c), where the convergences to the QZD in Eqs.\ (\ref{eqn:ExampleSelectiveBB}) and (\ref{eqn:ExampleNonselectiveBB}) via alternating the selective measurements $\{\mathcal{P}_1,\mathcal{P}_2\}$ and via alternating the nonselective measurements $\{\tilde{\mathcal{P}}_1,\tilde{\mathcal{P}}_2\}$, respectively, are numerically verified.

\subsection{Efficiency in Time by Pulsed Weak Measurements}
\label{sec:Efficiency}
In real experiments, it takes time to perform strong (projective) measurements (it takes time to project a system).
While the convergence to QZD would be faster with a stronger measurement (requiring less number of measurements; see Fig.~\ref{fig:ExampleWeakMeas}), it would consume more experimental time.
Since Theorem~\ref{thm:kicktozeno} tells us that QZD can be induced even via weak measurements, it could be better to proceed to the next measurement without waiting for a system being projected by a measurement, to save experimental time.
It is actually the case.
Let us see the efficiency in inducing QZD in terms of experimental time.

Consider, for instance, the model analyzed in Sec.~\ref{sec:ExampleWeakMeas}.
Suppose that we spend time $\tau$ for each measurement $\mathcal{E}$.
The strength of the measurement $p$ is a monotonically increasing function of the measurement time $\tau$ in general.
We perform $n$ measurements $\mathcal{E}$ at time intervals $t/n$, i.e., the system evolves as $(\mathcal{E}\rme^{-\rmi\frac{t}{n}\mathcal{H}})^n$.
Here, we assume that the unitary evolution of the system is turned off during the measurement process, or that the coupling to the measurement apparatus is strong enough so that the unitary evolution during the measurement process is negligible.
The total time spent for the $n$ measurements is given by $n\tau$ and the total experimental time is $n\tau+t$.
See Fig.\ \ref{fig:FiniteMeasTime}.
\begin{figure}
\centering
\includegraphics[scale=0.4]{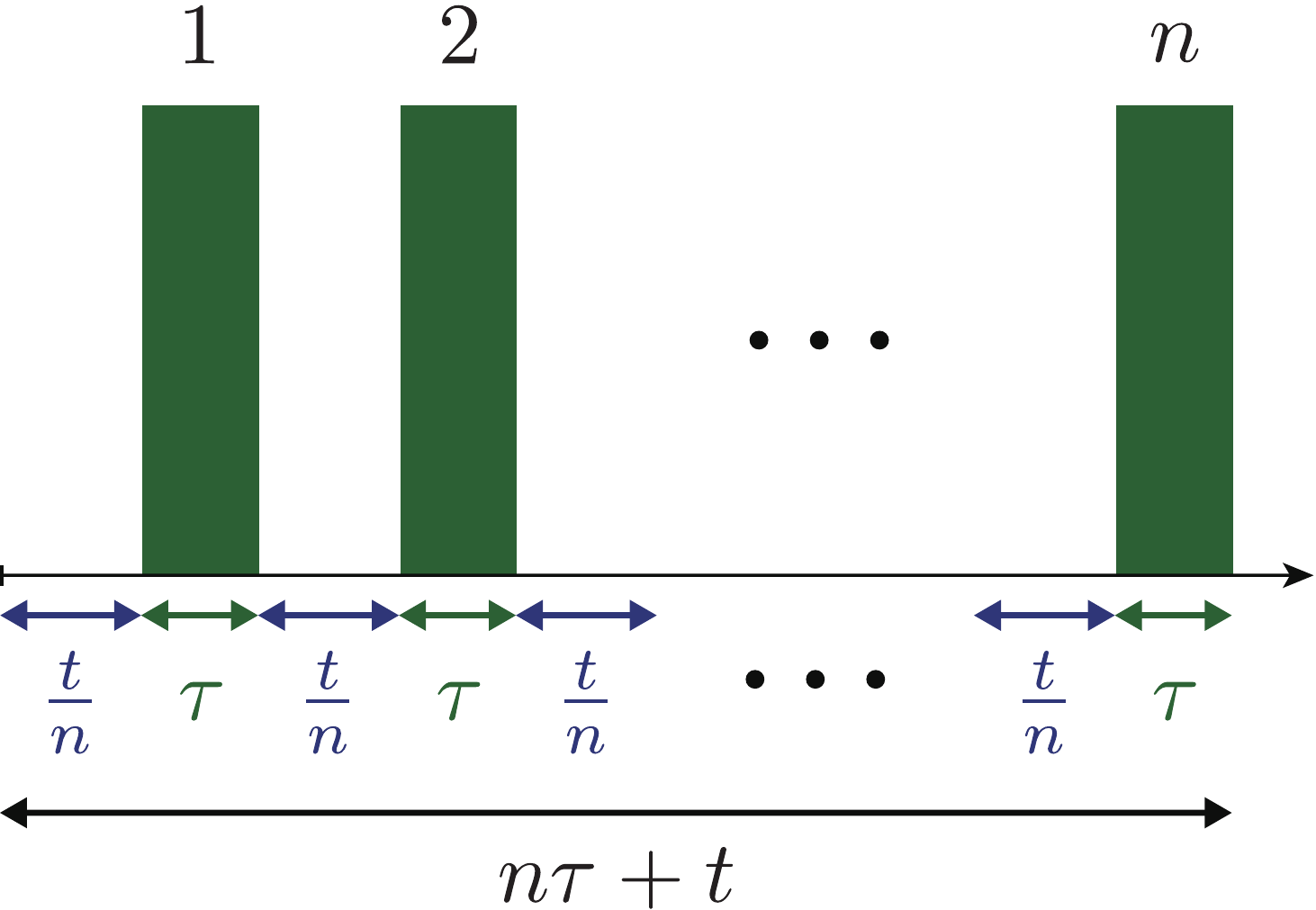}
\caption{A sequence of $n$ measurements with finite measurement times $\tau$ performed at regular time intervals $t/n +\tau$.}
\label{fig:FiniteMeasTime}
\end{figure}

\begin{figure}
\centering
\begin{tabular}{lll}
\footnotesize(a) $p(\tau)=1-\rme^{-\tau/T}$
&
\footnotesize(b) $p(\tau)=\sin(\pi\tau/2T)$
&
\footnotesize(c) $p(\tau)=\sin^2(\pi\tau/2T)$
\\[2truemm]
\includegraphics[scale=0.502]{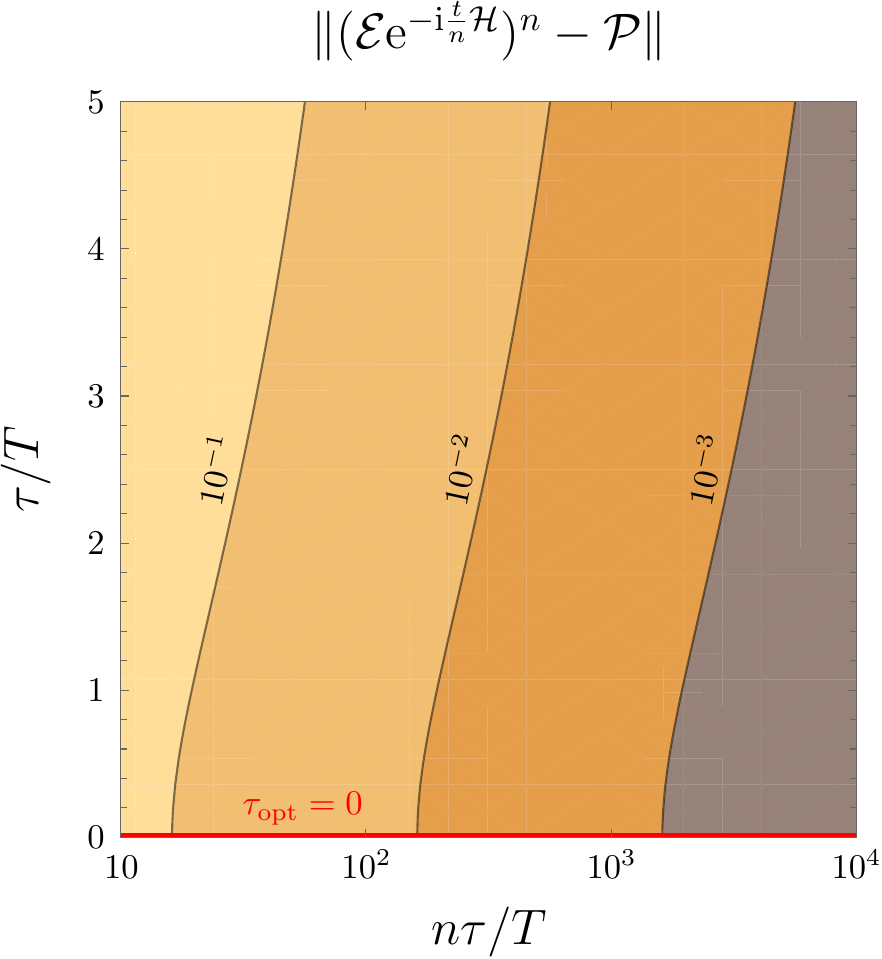}
&
\includegraphics[scale=0.502]{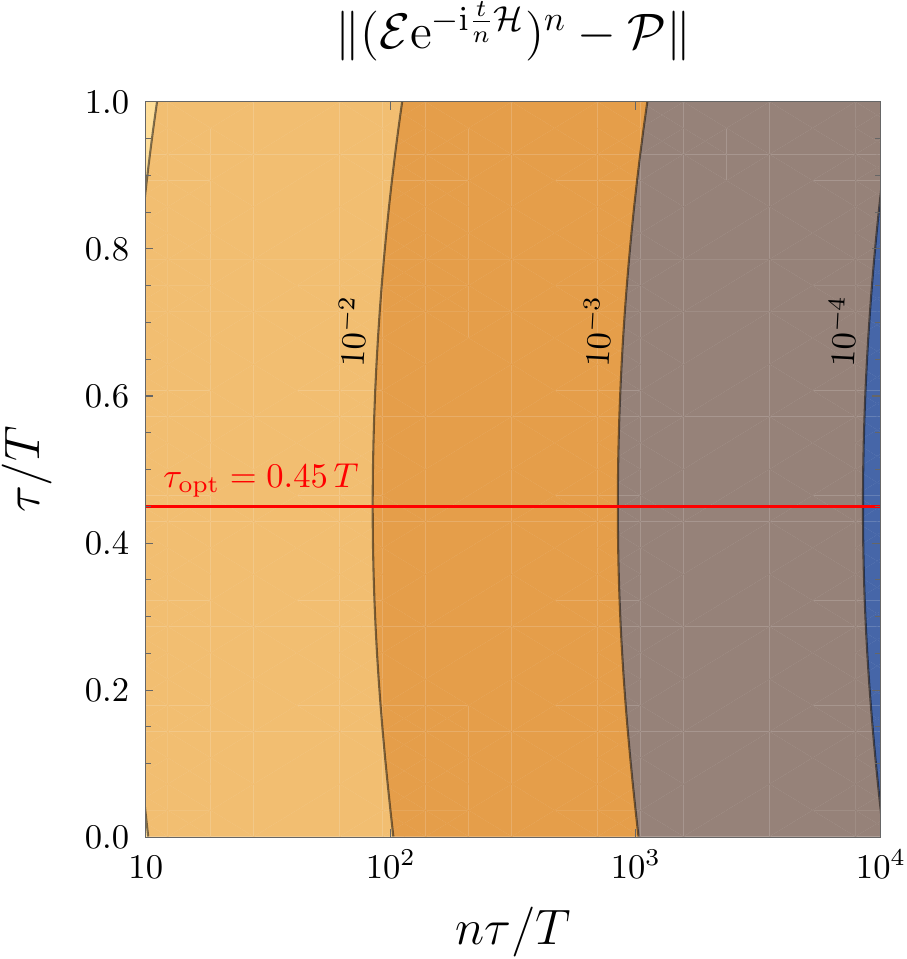}
&
\includegraphics[scale=0.502]{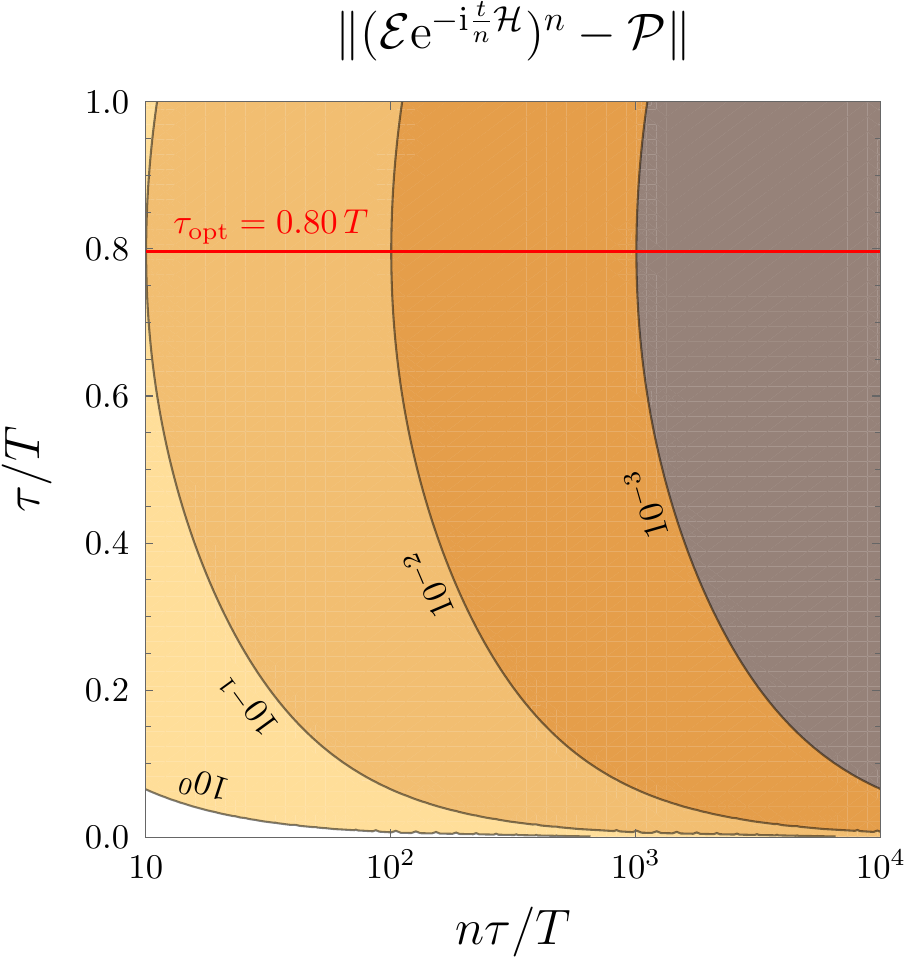}
\end{tabular}
\caption{Contour plots of $\|(\mathcal{E}\rme^{-\rmi\frac{t}{n}\mathcal{H}})^n-\mathcal{P}\|$ versus the total measurement time $n\tau$ and the time $\tau$ spent for each measurement, for the model given in Sec.~\ref{sec:ExampleWeakMeas} [the first term in Eq.\ (\ref{eqn:ExampleWeakMeasConv}) is actually plotted]. We choose three different functions $p(\tau)$ for the strength $p$ of the measurement $\mathcal{E}$: (a) $p(\tau)=1-\rme^{-\tau/T}$, (b) $p(\tau)=\sin(\pi\tau/2T)$, and (c) $p(\tau)=\sin^2(\pi\tau/2T)$, where $T$ is a characteristic time of each measurement process. The parameter is set at $\Omega t=1$. The optimal measurement time $\tau_\text{opt}$ which minimises the total measurement time given a certain degree of convergence is indicated by a red line for each measurement model.}
\label{fig:CPTPZenoEfficiency}
\end{figure}
We consider three models for the strength of the measurement $p(\tau)$ as a function of measurement time $\tau$: (a) $p(\tau)=1-\rme^{-\tau/T}$, (b) $p(\tau)=\sin(\pi\tau/2T)$, and (c) $p(\tau)=\sin^2(\pi\tau/2T)$, where $T$ is a characteristic time of each measurement process.
In the first model (a) projective measurement $p\to1$ is realized in the limit $\tau\to+\infty$, while in the other models (b) and (c) the measurement becomes perfectly projective $p=1$ at $\tau=T$.
In Fig.~\ref{fig:CPTPZenoEfficiency}, the distance $\|(\mathcal{E}e^{-\rmi\frac{t}{n}\mathcal{H}})^n-\mathcal{P}\|$ to the QZD is shown versus the total measurement time $n\tau$ and the time $\tau$ spent for each measurement, for the measurement models (a)--(c).
We see that it is better to proceed with nonprojective measurements to save time in these examples.

\section{Conclusions}
Our unification and generalization of QZDs has revealed an adiabatic evolution as the key ingredient. It is remarkable that such a variety of limits can be reduced to adiabaticity. We left for future studies a discussion of the tightness of our error bounds, and how they scale with the dimensionality of the Hilbert space. We also did not consider infinite-dimensional systems. Since the adiabatic theorem has itself many generalizations to infinite-dimensional systems and unbounded operators, this connection might pave the way to QZDs with unbounded operators, where there remain many open problems \cite{ref:PaoloSaverio-QZEreview-JPA,ref:unity1-view,ref:artzeno,ExnerIchinose,DDUnbounded,DDUnbounded2}. However, our proof via a generalized Baker-Campbell-Hausdorff formula does not easily generalize to infinite dimensions, and finding a more direct bridge between kicked dynamics and the adiabatic theorem would be desirable.

\section*{Acknowledgments}
DB acknowledges support by Waseda University and partial support by the EPSRC Grant No.~EP/M01634X/1\@.
This work was supported by the Top Global University Project from the Ministry of Education, Culture, Sports, Science and Technology (MEXT), Japan.
KY was supported by the Grants-in-Aid for Scientific Research (C) (No.~18K03470) and for Fostering Joint International Research (B) (No.~18KK0073) both from the Japan Society for the Promotion of Science (JSPS), and by the Waseda University Grant for Special Research Projects (No.~2018K-262).
PF and SP are  supported by INFN through the project `QUANTUM' and by MIUR via PRIN 2017 (Progetto di Ricerca di Interesse Nazionale), project QUSHIP (2017SRNBRK)\@. PF is supported by the Italian National Group of Mathematical Physics (GNFM-INdAM).
PF and SP are supported by Regione Puglia and by QuantERA ERA-NET Cofund in Quantum Technologies (GA No.\ 731473), project PACE-IN\@.

\appendix
\section*{Appendix}
\section{Some Basic Lemmas}
\label{app:BasicLemmas}
Here we prove five basic lemmas. The first one (Lemma~\ref{lem:NormCPT}) concerns a bound on the maps representing quantum operations. It gives a universal bound valid for any quantum operation.
The second one (Lemma~\ref{lem:BoundK}) is related to the relaxation to the peripheral eigenspace of quantum operation by its repeated applications.
Such bounds are known (e.g.\ Refs.\ \cite[Lemma~8.5 and Theorem~8.24]{ref:Mixing-Wolf}, \cite{ref:Szehr2015}, and \cite{ref:TerhalDiVincenzo}), but we need a generalized version of them for our purpose. 
The point is that, to apply Lemma~\ref{prop:Peripheral}, we need to bound families of maps like $(\mathcal{P}_\varphi\mathcal{E}_n\mathcal{P}_\varphi)^k$ and $(\mathcal{E}_n-\mathcal{P}_\varphi\mathcal{E}_n\mathcal{P}_\varphi)^k$ as in Eqs.\ (\ref{eqn:Bound1}) and (\ref{eqn:Bound2}) universally for any $n$ large enough.
To make the bounds in Lemma~\ref{lem:BoundK} more explicit for the particular sequence of quantum maps relevant to Theorem~\ref{thm:CPTPBB}, we prove Lemmas~\ref{lem:BBperturb} and~\ref{lem:MuN}\@.
Finally, in Lemma~\ref{lem:P1P2}, we prove some facts on the product of Hermitian projection operators, which are used to get Corollary~\ref{thm:MeasBB} from Theorem~\ref{thm:CPTPBB}.

\begin{lemma}[Norm of quantum operation \cite{ref:PrezGarcaWolfPetzRuskai-JMP2006}]\label{lem:NormCPT}
Consider the $d^2$-dimensional Hilbert space $\mathcal{T}_2$ of operators on a $d$-dimensional Hilbert space $\mathcal{H}$, with the Hilbert-Schmidt inner product $\langle X|Y \rangle_{2} =\tr(X^\dag Y)$ for $X,Y\in\mathcal{T}_2$. 
The operator norm of any quantum operation $\mathcal{E}:\mathcal{T}_2\to\mathcal{T}_2$ is bounded by
\begin{equation}
	\|\mathcal{E}\|\le\sqrt{d}.
\end{equation}	
\end{lemma}
\begin{proof}
Recall the definition of the operator norm of $\mathcal{A}:\mathcal{T}_2\to\mathcal{T}_2$ in Eq.\ (\ref{eqn:OpNorm}):
\begin{equation}
\|\mathcal{A} \| = \sup_{\|X\|_2 = 1} \| \mathcal{A} (X)\|_2,
\end{equation}
where $\|X\|_2 = [\tr (X^\dag X)]^{1/2}$.
We have $\|\mathcal{A}^\dag\| = \|\mathcal{A}\|$. Moreover, for any operator $X$ on the Hilbert space $\mathcal{H}$, we have that
\begin{equation}
\|X\|_\infty\leq \|X\|_2 \leq \sqrt{d}\,\|X\|_\infty,
\end{equation}
where $\|X\|_\infty=\sup_{\|v\|=1} \|X v\|$ with $v\in\mathcal{H}$.

Now, we have that $\|\mathcal{E}^\dag (X)\|_\infty \le\|\mathcal{E}^\dag (I)\|_\infty \|X\|_\infty$ by a theorem of Russo and Dye \cite[Corollary~2.9]{ref:PaulsenBook}. But the adjoint of a quantum operation $\mathcal{E}$ is subunital, $\mathcal{E}^\dag(I)\leq I$.
Therefore,
\begin{align}
\|\mathcal{E}^\dag(X)\|_2 \leq \sqrt{d}\,\|\mathcal{E}^\dag(X)\|_\infty
&\leq \sqrt{d}\,\|\mathcal{E}^\dag(I)\|_\infty \|X\|_\infty
\nonumber\\
&\leq \sqrt{d}\,\|I \|_\infty \|X\|_\infty = \sqrt{d}\,\|X\|_\infty \leq \sqrt{d}\,\|X\|_2,
\end{align}
whence
	\begin{equation}
		\|\mathcal{E}\|
		=\|\mathcal{E}^\dag\| \leq \sqrt{d}.
	\end{equation}
\end{proof}
We note that a more natural choice of norm for CP maps would be the trace norm, in which they are contractive. 
However, to prove the next Lemma~\ref{lem:BoundK}, the operator norm is more useful since we need to deal with spectral radius. 
A universal bound on quantum operation in the operator norm is required in Lemma~\ref{lem:BoundK}, and therefore, we have derived it in Lemma~\ref{lem:NormCPT}. 
Using the trace norm does not simplify matters and we stick to the operator norm.

The next lemma is a variant of Refs.\ \cite[Lemma~8.5 and Theorem~8.24]{ref:Mixing-Wolf} and \cite{ref:Szehr2015}.
\begin{lemma}\label{lem:BoundK}
Let $(\mathcal{E}_n)$ be a convergent sequence of quantum operations on
a $d$-dimensional quantum system  
with  
\begin{equation}
\mathcal{E}_n\to\mathcal{E}\quad  \text{as}\quad n\to+\infty.
\end{equation}
Let $\mathcal{P}_\varphi$ be the peripheral spectral projection of the quantum operation $\mathcal{E}$, and $\mu_0= r(\mathcal{E}- \mathcal{E}_\varphi) < 1$  the spectral radius of its nonperipheral part.
Then, for any $\mu\in(\mu_0,1)$ there exists an integer $n_0>0$ and a positive number $K>0$ such that 
\begin{eqnarray}
\|(\mathcal{E}_n-\mathcal{P}_\varphi\mathcal{E}_n\mathcal{P}_\varphi)^k\| \le K \mu^k,
\label{eqn:PowerBoundIII}
\end{eqnarray}
for all $k\in\mathbb{N}$ and for all $n>n_0$.
\end{lemma}
\begin{proof}
Given $\mu\in(\mu_0,1)$, fix a $\mu_1\in(\mu_0,\mu)$ and define $\mathcal{E}_n'=\mathcal{E}_n-\mathcal{P}_\varphi\mathcal{E}_n\mathcal{P}_\varphi$.
Since 
\begin{equation}
r
(\mathcal{E}_n') \to r(\mathcal{E}-\mathcal{P}_\varphi\mathcal{E}\mathcal{P}_\varphi)=\mu_0<1, 
\end{equation}
as $n\to+\infty$, there exists an integer $n_0>0$ such that
\begin{equation}
r(\mathcal{E}_n')
<\mu_1< \mu <1,
\quad
\forall n> n_0.
\end{equation}

Recall now that $\mathcal{E}_n'$ can be transformed into an upper-triangular matrix by a unitary transformation $\mathcal{U}_n$ (Schur triangulation),
\begin{equation}
\mathcal{E}_n'
=\mathcal{U}_n^\dag(
\Lambda_n+\mathcal{N}_n
)\mathcal{U}_n,
\end{equation}
where $\Lambda_n$  is diagonal while $\mathcal{N}_n$ is a strictly upper-triangular matrix with vanishing diagonal elements, which is nilpotent,
\begin{equation}
\mathcal{N}_n^{d^2}=0.
\end{equation}

Since $\mathcal{P}_\varphi$ is a quantum operation by Proposition~\ref{prop:CPTP} (iii),  
$\mathcal{P}_\varphi\mathcal{E}_n\mathcal{P}_\varphi$ is also a quantum operation.
Then, the operator norm 
of the nilpotent part $\mathcal{N}_n$ is bounded as 
\begin{align}
\|\mathcal{N}_n\|
&=\|\mathcal{U}_n
\mathcal{E}_n' \mathcal{U}_n^\dag
-\Lambda_n\|
\nonumber\displaybreak[0]\\
&\le\|\mathcal{E}_n'\|
+\|\Lambda_n\|
\nonumber\displaybreak[0]\\
&\le\|\mathcal{E}_n\|
+\|\mathcal{P}_\varphi\mathcal{E}_n\mathcal{P}_\varphi\|
+\|\Lambda_n\|
\nonumber\displaybreak[0]\\
&\le2\sqrt{d}+\mu_1,
\end{align}
where we have used the fact that $\|\mathcal{E}_n\|, \|\mathcal{P}_\varphi\mathcal{E}_n\mathcal{P}_\varphi\| \leq \sqrt{d}$ by Lemma~\ref{lem:NormCPT}, and the fact that $\|\Lambda_n\| = r(\mathcal{E}_n')\leq \mu_1$.

Now, we estimate the norm $\|\mathcal{E}_n'^k\|=\|(\Lambda_n+\mathcal{N}_n)^k\|$.
To this end, notice that, in the binomial expansion of $(\Lambda_n+\mathcal{N}_n)^k$, the terms in which $\mathcal{N}_n$ appears more than $d^2-1$ times vanish, irrespective of the fact that $\Lambda_n$ and $\mathcal{N}_n$ do not commute in general.
Then, we can bound the norm as 
\begin{align}
\|\mathcal{E}_n'^k\|
&=
\|(\Lambda_n+\mathcal{N}_n)^k\|
\nonumber\displaybreak[0]\\
&\le
\sum_{j=0}^{\min(k,d^2-1)}
\begin{pmatrix}
k\\j
\end{pmatrix}
\|
\Lambda_n
\| ^{k-j}
\|
\mathcal{N}_n
\| ^j
\nonumber\displaybreak[0]\\
&\le
\sum_{j=0}^{\min(k,d^2-1)}
\begin{pmatrix}
k\\j
\end{pmatrix}
\mu_1^{k-j}
(2\sqrt{d}+\mu_1)^j
\nonumber\displaybreak[0]\\
&\le
\sum_{j=0}^{d^2-1}
\frac{k^j}{j!}
\mu_1^{k-j}
(2\sqrt{d}+\mu_1)^j
\nonumber\displaybreak[0]\\
&\le k^{d^2-1} \mu_1^k
\sum_{j=0}^{+\infty}
\frac{1}{j!}
\left(\frac{2\sqrt{d}}{\mu_1}+1\right)^j
\nonumber\displaybreak[0]\\
&=
\rme^{2\sqrt{d}/\mu_1+1}
k^{d^2-1} \mu_1^k,
\qquad\forall k\in\mathbb{N}, \quad \forall n>n_0,
\end{align}
that is
\begin{equation}
\|(\mathcal{E}_n-\mathcal{P}_\varphi\mathcal{E}_n\mathcal{P}_\varphi)^k\|
\le K_1 k^{d^2-1} \mu_1^k, \qquad\forall k\in\mathbb{N}, \quad \forall n>n_0,
\label{eqn:PowerBoundII}
\end{equation}
with $K_1= \rme^{2\sqrt{d}/\mu_1+1}$.

Finally, since $\mu>\mu_1$, we can always find a $K\geq K_1$ such that
\begin{equation}
K_1 k^D\mu_1^k\le K \mu^k,\qquad\forall k\in\mathbb{N},
\end{equation}
where $D=d^2-1$.
Indeed, we have
\begin{align}
\log\frac{K \mu^k}{K_1 k^D\mu_1^k}
&=\log\frac{K}{K_1}
+k\log\frac{\mu}{\mu_1}-D\log k
\nonumber\\
&\ge\log\frac{K}{K_1}
+D
-D\log\frac{D}{\log(\mu/\mu_1)}.
\end{align}
We can make it nonnegative by choosing
\begin{equation}
K\ge  K_1 \left(\frac{D}{\rme \log(\mu/\mu_1)}\right)^D =
\rme^{2\sqrt{d}/\mu_1+1} \left(\frac{d^2-1}{\rme \log(\mu/\mu_1)}\right)^{d^2-1}.
\label{eqn:Kbound}
\end{equation}
Together with Eq.~\eqref{eqn:PowerBoundII} this gives the bound (\ref{eqn:PowerBoundIII}) of the lemma.
\end{proof}
\begin{remark}
The constant in Eq.~\eqref{eqn:PowerBoundIII} can be chosen as
\begin{equation}
K= \rme^{2\sqrt{d/\mu\mu_0}+1} \left(\frac{2(d^2-1)}{\rme\log(\mu/\mu_0)}\right)^{d^2-1},
\end{equation}
by putting $\mu_1=\sqrt{\mu \mu_0}$ in the lower bound in Eq.\ (\ref{eqn:Kbound}). Note that it was assumed that $\mu_1\in(\mu_0,\mu)$ in the proof of Lemma~\ref{lem:BoundK}.
Moreover, notice that in fact in the proof of Lemma~\ref{lem:BoundK} we have obtained the tighter bound~\eqref{eqn:PowerBoundII}:
\begin{equation}
\|(\mathcal{E}_n-\mathcal{P}_\varphi\mathcal{E}_n\mathcal{P}_\varphi)^k\|
\le K  k^{d^2-1} \mu^k, \qquad\forall k\in\mathbb{N}, \quad \forall n>n_0,
\end{equation}
with $K= \rme^{2\sqrt{d}/\mu+1}$.
\end{remark}

In the proof of Theorem~\ref{thm:CPTPBB}, we use Lemma~\ref{lem:BoundK} for the sequence of quantum operations $\tilde{\mathcal{E}}_n=\mathcal{E}_m\rme^{\frac{t}{mn}\mathcal{L}}\cdots\mathcal{E}_1\rme^{\frac{t}{mn}\mathcal{L}}$ in Eq.\ (\ref{eqn:BBseq}), which converges to a quantum operation $\mathcal{E}=\mathcal{E}_m\cdots\mathcal{E}_1$ in the limit $n\to+\infty$.
The following lemma explicitly clarifies the bound on the speed of the convergence $\tilde{\mathcal{E}}_n\to\mathcal{E}$.
\begin{lemma}
\label{lem:BBperturb}
Let $\{\mathcal{E}_1,\ldots,\mathcal{E}_m\}$ be a finite set of quantum operations and $\mathcal{L}$ be a GKLS generator of a $d$-dimensional quantum system. Then, we have 
\begin{equation}
\|
\mathcal{E}_m \rme^{\frac{t}{mn}\mathcal{L}}\cdots
\mathcal{E}_1 \rme^{\frac{t}{mn}\mathcal{L}}
-\mathcal{E}
\| 
\le
d^{m/2}
\frac{t}{n}\|\mathcal{L}\| \rme^{\frac{t}{n}\|\mathcal{L}\| },\quad\forall n\in\mathbb{N},
\end{equation}
uniformly in $t$ on compact intervals of $[0,+\infty)$, where $\mathcal{E} =\mathcal{E}_m\cdots\mathcal{E}_1$, and we have chosen the operator norm.
\end{lemma}
\begin{proof}
We split each piece of the evolution as $\rme^{\frac{t}{mn}\mathcal{L}}=1+(\rme^{\frac{t}{mn}\mathcal{L}}-1)$.
The deviation from the identity map is bounded by
\begin{equation}
	\|\rme^{\frac{t}{mn}\mathcal{L}}-1\| 
	\le\rme^{\frac{t}{mn}\|\mathcal{L}\| }-1.
\end{equation}
Then, the distance is bounded by
\begin{align}
\|
\mathcal{E}_m \rme^{\frac{t}{mn}\mathcal{L}}\cdots
\mathcal{E}_1 \rme^{\frac{t}{mn}\mathcal{L}}
-\mathcal{E}
\| 
&\le
d^{m/2}
\sum_{j=1}^m\begin{pmatrix}
m\\j
\end{pmatrix}
\|\rme^{\frac{t}{mn}\mathcal{L}}-1\| ^j
\nonumber\\
&=
d^{m/2}
\left[
\left(
1+\|\rme^{\frac{t}{mn}\mathcal{L}}-1\| 
\right)^m-1\right]
\nonumber\\
&\le
d^{m/2}
\left(
\rme^{\frac{t}{n}\|\mathcal{L}\| }
-1
\right)
\nonumber\\
&=
d^{m/2}
\frac{t}{n}\|\mathcal{L}\| \rme^{\frac{t_*}{n}\|\mathcal{L}\| }
\nonumber\displaybreak[0]\\
&\le
d^{m/2}
\frac{t}{n}\|\mathcal{L}\| \rme^{\frac{t}{n}\|\mathcal{L}\| },
\end{align}
where we have used $\|\mathcal{E}_j\| \le\sqrt{d}$ ($j=1,\ldots,m$) for the operator norm (Lemma~\ref{lem:NormCPT}) and the mean-value theorem, $[F(t)-F(0)]/t=F'(t_*)$ for some $t_*\in[0,t]$, for $F(t)=\rme^{\frac{t}{n}\|\mathcal{L}\| }$, with $F'(t)$ denoting the derivative of $F(t)$ with respect to $t$.
\end{proof}

Lemma~\ref{lem:BoundK} states that for any $\mu\in(\mu_0,1)$ there exists an integer $n_0>0$ such that Eq.~\eqref{eqn:PowerBoundIII} holds for all $n>n_0$.
The smaller $\mu$ is, the faster the bound shrinks as $k$ increases.
However, if we demand that $\mu$ is (larger than but) close to $\mu_0=r(\mathcal{E}-\mathcal{P}_\varphi\mathcal{E}\mathcal{P}_\varphi)$, the spectral radius  of the nonperipheral part of the limit map $\mathcal{E}$, namely the largest among the magnitudes of its nonperipheral eigenvalues, we would need a large $n_0$.
The following lemma clarifies such a trade-off between $n_0$ and $\mu$, for the bang-bang sequence $\tilde{\mathcal{E}}_n=\mathcal{E}_m\rme^{\frac{t}{mn}\mathcal{L}}\cdots\mathcal{E}_1\rme^{\frac{t}{mn}\mathcal{L}}$ in Eq.\ (\ref{eqn:BBseq}) relevant to Theorem~\ref{thm:CPTPBB}.
\begin{lemma}
\label{lem:MuN}
Let $\{\mathcal{E}_1,\ldots,\mathcal{E}_m\}$ be a finite set of quantum operations and $\mathcal{L}$ be a GKLS generator of a $d$-dimensional quantum system.
We consider the sequence of quantum operations $\tilde{\mathcal{E}}_n=\mathcal{E}_m \rme^{\frac{t}{mn}\mathcal{L}}\cdots
\mathcal{E}_1 \rme^{\frac{t}{mn}\mathcal{L}}$, which converges to  $\mathcal{E}=\mathcal{E}_m\cdots\mathcal{E}_1$ in the limit $n\to+\infty$, and let $\mathcal{P}_\varphi$ the peripheral spectral projection of $\mathcal{E}$.
Then, to bound $\|(\tilde{\mathcal{E}}_n-\mathcal{P}_\varphi\tilde{\mathcal{E}}_n\mathcal{P}_\varphi)^k\| $ by Lemma~\ref{lem:BoundK}, we can choose $n_0$ satisfying
\begin{equation}
\mu_0 + \left(
	(1+d)d^{m/2+2}
\frac{t}{n_0}\|\mathcal{L}\| \rme^{\frac{t}{n_0}\|\mathcal{L}\| }
	\right)^{1/d^2}\Bigl(1+\|\mathcal{N}_{n_0}\| \Bigr)
	\le\mu<1,
	\label{eqn:MuN}
\end{equation}
where $\mathcal{N}_{n}$ is the nilpotent part of $\Theta_n=
(\tilde{\mathcal{E}}_n-\mathcal{E})-\mathcal{P}_\varphi(\tilde{\mathcal{E}}_n-\mathcal{E})\mathcal{P}_\varphi$.
\end{lemma}
\begin{proof}
The constant $\mu>\mu_0=r(\mathcal{E}-\mathcal{P}_\varphi\mathcal{E}\mathcal{P}_\varphi)$ is to bound the spectral radius $r(\tilde{\mathcal{E}}_n')$ of $\tilde{\mathcal{E}}_n'=\tilde{\mathcal{E}}_n-\mathcal{P}_\varphi\tilde{\mathcal{E}}_n\mathcal{P}_\varphi$ for all $n>n_0$.
We wish to know how much $\mu$ should be larger than $\mu_0$ to bound $r(\tilde{\mathcal{E}}_n')$ from above.

To this end, we use Theorem~7.2.3 of Ref.\ \cite{ref:GolubVanLoanMatrixComp} concerning how much the eigenvalues can be altered by a perturbation for generic matrices (including nondiagonalizable ones).
Let $\lambda$ denote the largest (in magnitude) eigenvalue of $\mathcal{E}'=\mathcal{E}-\mathcal{P}_\varphi\mathcal{E}\mathcal{P}_\varphi$ and $\tilde{\lambda}_\ell$ the eigenvalues of $\tilde{\mathcal{E}}_n'$.
According to Theorem~7.2.3 of Ref.\ \cite{ref:GolubVanLoanMatrixComp}, we have
\begin{equation}
	\min_\ell|\tilde{\lambda}_\ell-\lambda|\le\max\{\theta_n,\theta_n^{1/p_n}\},
\end{equation}
where
\begin{equation}
\theta_n=\|\Theta_n\| \sum_{j=0}^{p_n-1}\|\mathcal{N}_n\| ^j,
\end{equation}
with $\Theta_n= \tilde{\mathcal{E}}_n'-\mathcal{E}'$, $\mathcal{N}_n$ being its nilpotent part, and 
$p_n$ being the smallest integer such that $\mathcal{N}_n^{p_n}=0$.
Caring about the worst case, we have a bound
\begin{equation}
	r(\tilde{\mathcal{E}}_n')\le r(\mathcal{E}')+\max\{\theta_n,\theta_n^{1/p_n}\}.
	\label{eqn:BoundSpecRadPrime}
\end{equation}
Using Lemma~\ref{lem:BBperturb}, $\|\Theta_n\| $ is bounded by
\begin{equation}
	\|\Theta_n\| \le(1+d)\|\tilde{\mathcal{E}}_n-\mathcal{E}\| 
\le
(1+d)d^{m/2}
\frac{t}{n}\|\mathcal{L}\| \rme^{\frac{t}{n}\|\mathcal{L}\| },
\end{equation}
and $\theta_n$ by
\begin{equation}
\theta_n\le (1+d)d^{m/2+2}
\frac{t}{n}\|\mathcal{L}\| \rme^{\frac{t}{n}\|\mathcal{L}\| }\,\Bigl(1+\|\mathcal{N}_n\| \Bigr)^{d^2}.
\end{equation}
Note that $\|\mathcal{P}_\varphi\| \le\sqrt{d}$ (Lemma~\ref{lem:NormCPT}), $\sum_{j=0}^{p_n-1}\|\mathcal{N}_n\| ^j\le p_n\max\{1,\|\mathcal{N}_n\| ^{p_n-1}\}\le d^2(1+\|\mathcal{N}_n\| )^{d^2}$, and $p_n\le d^2$.
Once this upper bound on $\theta_n$ becomes smaller than $1$ by increasing $n$, the bound (\ref{eqn:BoundSpecRadPrime}) is reduced to
\begin{align}
	r(\tilde{\mathcal{E}}_n')
	&\le r(\mathcal{E}')+\theta_n^{1/d^2}
	\nonumber\\
	&\le r(\mathcal{E}')
	+\left(
	(1+d)d^{m/2+2}
\frac{t}{n}\|\mathcal{L}\| \rme^{\frac{t}{n}\|\mathcal{L}\| }
	\right)^{1/d^2}\Bigl(1+\|\mathcal{N}_n\| \Bigr),
\end{align}
and we get the condition (\ref{eqn:MuN}) on $n_0$.
Note that $\mu$ should be strictly smaller than $1$; otherwise it does not make sense.
\end{proof}

\begin{lemma}[Peripheral part of the product of Hermitian projections]
\label{lem:P1P2}
Let $\{P_1, \ldots, P_m\}$ be a set of Hermitian projection operators, $P_j^2=P_j=P_j^\dag$ for $j=1,\ldots,m$, on a Hilbert space.
Then, the peripheral part of the product $P_m \cdots P_1$ is given by the Hermitian projection
\begin{equation}
P_\varphi=P_1\wedge \cdots \wedge P_m
\end{equation}
onto the intersection of the ranges $\ran P_1 \cap \cdots \cap \ran P_m$.
We get $P_\varphi P_j = P_j P_\varphi = P_\varphi$ for all $j=1,\ldots,m$.
\end{lemma}
\begin{proof}
Suppose that $P_2P_1$ admits a peripheral eigenvalue $\lambda$, and let $u\neq 0$ be the corresponding eigenvector,
\begin{equation}	
P_2P_1 u=\lambda u
\quad\text{with}\quad
|\lambda|=1.
\label{eqn:P2P1Eigen}
\end{equation}
By the Cauchy-Schwarz inequality, we have
\begin{equation}
  \|u\|^2
  =\|P_2P_1 u \|^2
  = \langle P_1 u |  P_2P_1 u \rangle 
  \le\|P_1 u\| \|P_2P_1 u\| \le \| u \|^2,
\end{equation}
where we have used $\| P_j\| \le 1$ and Eq.\ (\ref{eqn:P2P1Eigen}).
The two inequalities are actually saturated, implying
\begin{equation}
  P_1 u=u.
\end{equation}
This simplifies Eq.\ (\ref{eqn:P2P1Eigen}) to
\begin{equation}	
P_2 u =\lambda u,
\end{equation}
and hence one gets that $\lambda=1$.
In this way, the peripheral eigenvalue of $P_2P_1$, if any, is $1$, and the corresponding eigenvector $u$ is a simultaneous eigenvector of $P_1$ and $P_2$ belonging to their nonvanishing eigenvalue $1$.
Moreover, since we have
\begin{equation}
	(P_2 P_1)^\dagger u = P_1 P_2 u=u,
\end{equation}
the conjugate of $u$ is a left-eigenvector of $P_2P_1$ belonging to the same eigenvalue $1$.
Therefore, the peripheral part of $P_2P_1$ is diagonalizable and is the Hermitian projection
\begin{equation}
P_\varphi=P_1\wedge P_2 
\end{equation}
onto $\ran P_1 \cap \ran P_2$, the intersection of the ranges of $P_1$ and $P_2$.

By induction we get that the peripheral part of $P_m \cdots P_1$ is given by the Hermitian projection
\begin{equation}
P_\varphi=P_1\wedge \cdots \wedge P_m
\end{equation}
onto the intersection of the ranges $\ran P_1 \cap \cdots \cap \ran P_m$.
\end{proof}

\section{Bounding the BCH Formula for Lemma~\ref{thm:kicktofield}}
\label{app:BCH}
It is known that in general the BCH series converges only for small operators \cite{ref:Blanes-BCH}.
This is related to the issue of the convergence of the Magnus expansion \cite{ref:Blanes-Mugnus}.
For our purpose, however, this is not a problem.
We just need contributions up to the first order in $1/n$ in the BCH formula $Z(t)=\log(\rme^A\rme^{\frac{t}{n}L})$, and the higher-order corrections are under control, as proved in Lemma~\ref{thm:kicktofield}.
Here we provide a concise and explicit bound on the corrections.

We first need to know how to bound $\|h(X+Y)-h(X)\|$ for a primary operator function $h(X)$.
If a series expansion exists for the stem function $h(z)$, e.g.
\begin{equation}
h(z)=\sum_{n=0}^\infty c_nz^n,
\label{eqn:SeriesH}
\end{equation}
it can be easily bounded as \cite[Theorem~6.2.30]{ref:MatrixAnalysisTopics-HornJohnson}
\begin{equation}
\|
h(X+Y)-h(X)
\|
\le
h_\mathrm{abs}'(\|X\|+\|Y\|)
\|Y\|,
\label{eqn:SimpleBoundHz}
\end{equation}
where 
\begin{equation}
h_\mathrm{abs}(z)=\sum_{n=0}^\infty|c_n|z^n,
\end{equation}
and $h_\mathrm{abs}'(z)$ is its derivative.
It is, however, not always the case that a good series expansion exists for $h(z)$.
For instance, suppose that the spectrum of $X$ is distributed on the complex plane as in Fig.~\ref{fig:ContourGamma} and $h(X)=\log X$.
Due to the singularity at the origin, there is no series expansion of $\log z$ defined for all the eigenvalues.
In addition, even if every eigenvalue of $X$ lies within the convergence radius of the series expansion of $h(z)$, 
the norm $\|X\|$ can exceed the convergence radius, and in such a case the bound (\ref{eqn:SimpleBoundHz}) is not applicable. 
This would happen in particular when $X$ is not diagonalizable and possesses nilpotents in its spectral representation.
We want a bound valid for any $X$, not necessarily diagonalizable.

\begin{lemma}\label{lem:MatrixFuncBound}
Let $X$ be an operator on a $D$-dimensional Banach space, whose spectral representation is
\begin{equation}
X=\sum_k(x_kP_k+N_k),
\end{equation}
with $P_k$ and $N_k$ being the spectral projection and the nilpotent belonging to the $k$th eigenvalue $x_k$ of $X$.
Let $h(z)$ be analytic on and inside a closed contour $\Gamma$ in the complex plane that encloses the $R$-neighborhood of the spectrum of $X$, $\operatorname{spec}(X) = \{x_k\}$, that is, $\operatorname{dist}(\Gamma, \operatorname{spec}(X)) 
\ge R>0$, and let 
\begin{equation}
h(X)=\frac{1}{2\pi \rmi} \oint_\Gamma \d z \, h(z) \frac{1}{z-X}.
\end{equation}
Set 
\begin{equation}
\beta= \frac{DP}{R}\frac{1-(N/R)^D}{1-N/R},
\quad\text{where} \quad
P=\max_k\|P_k\|,\quad
N=\max_k\|N_k\|,
\label{eqn:PNR}
\end{equation}
and 
\begin{equation}
M=
\frac{1}{2\pi}\oint_\Gamma|\rmd z||h(z)|.
\label{eqn:Mlemma}
\end{equation}
Then, for any  operator $Y$ satisfying $\beta \|Y\|< 1$, we have
\begin{equation}
\|h(X+Y)-h(X)\|
\le
M\frac{
\beta^2
\|Y\|
}{
1
-
\beta
\|Y\|
}.
\label{eqn:BoundHPerturb}
\end{equation}
\end{lemma}
\begin{proof}
We wish to bound
\begin{align}
h(X+Y)-h(X)
&
=
\frac{1}{2\pi\rmi}
\oint_{\Gamma}\rmd z\,
h(z)
\left(
\frac{1}{zI-X-Y}
-
\frac{1}{zI-X}
\right)
\nonumber\\
&=\sum_{n=1}^\infty
\frac{1}{2\pi\rmi}
\oint_{\Gamma}\rmd z\,
h(z)
\frac{1}{zI-X}
\left(
Y
\frac{1}{zI-X}
\right)^n.
\label{eqn:PerturbedComplexIntegral}
\end{align}
Let us estimate the resolvent,
\begin{equation}
\frac{1}{zI-X}
=
\sum_k
\frac{1}{(z-x_k)I-N_k}P_k
=
\sum_k
\sum_{q=0}^{n_k-1}
\frac{1}{(z-x_k)^{q+1}}N_k^qP_k.
\end{equation}
On the contour $\Gamma$, it is bounded by
\begin{equation}
\left\|
\frac{1}{zI-X}
\right\|
\le
\sum_k
\sum_{q=0}^{n_k-1}
\frac{1}{R^{q+1}}\|N_k\|^q\|P_k\|
\le
D
\sum_{q=0}^{D-1}
\frac{1}{R^{q+1}}N^qP
=
\frac{DP}{R}
\frac{
1-(N/R)^D
}{
1-N/R
} = \beta.
\end{equation}
Using this bound, we get
\begin{align}
\|h(X+Y)-h(X)\|
&\le
\frac{1}{2\pi}\oint_\Gamma|\rmd z||h(z)|
\sum_{n=1}^\infty
\|Y\|^n
\beta^{n+1}
\le
M
\frac{
\beta^2 \|Y\|
}{
1
-
\beta
\|Y\|
},
\label{eqn:Fbound}
\end{align}
provided that $\beta \|Y\|<1$.
\end{proof}

We apply Lemma~\ref{lem:MatrixFuncBound} to $\log z$ and to $g(z)$ defined in Eq.\ (\ref{eq:g{z}def}) to bound the BCH formula (\ref{eq:inverted-1-1}).
Let us try to get more informative expressions for $M$ for these specific functions, where we see how the spectrum of the input operator and the singularities of the stem function matter.

\begin{figure}
\centering
\begin{tabular}{l@{\qquad\qquad\quad}l}
\footnotesize(a)
&
\footnotesize(b)\\
\includegraphics[scale=0.35]{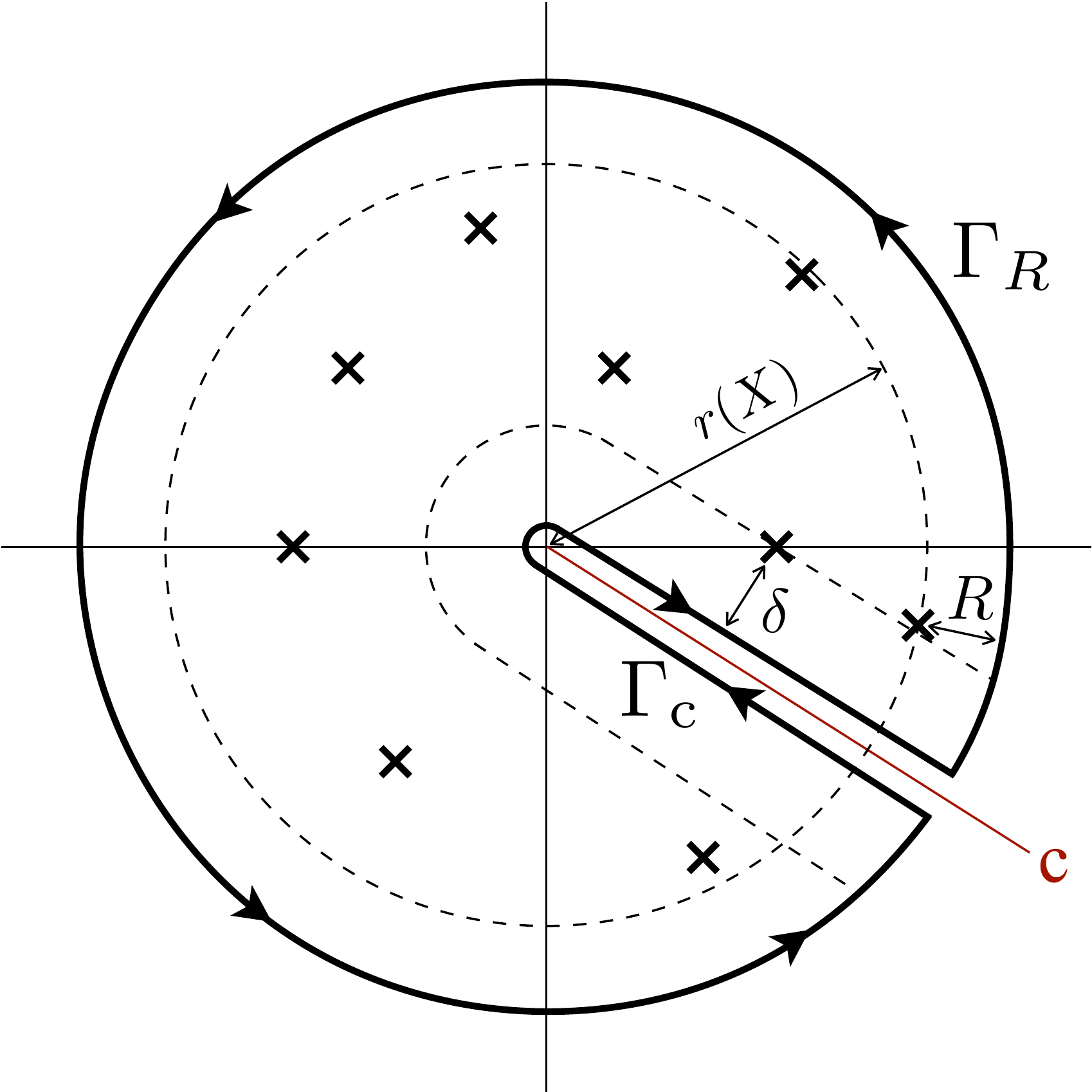}
&
\includegraphics[scale=0.35]{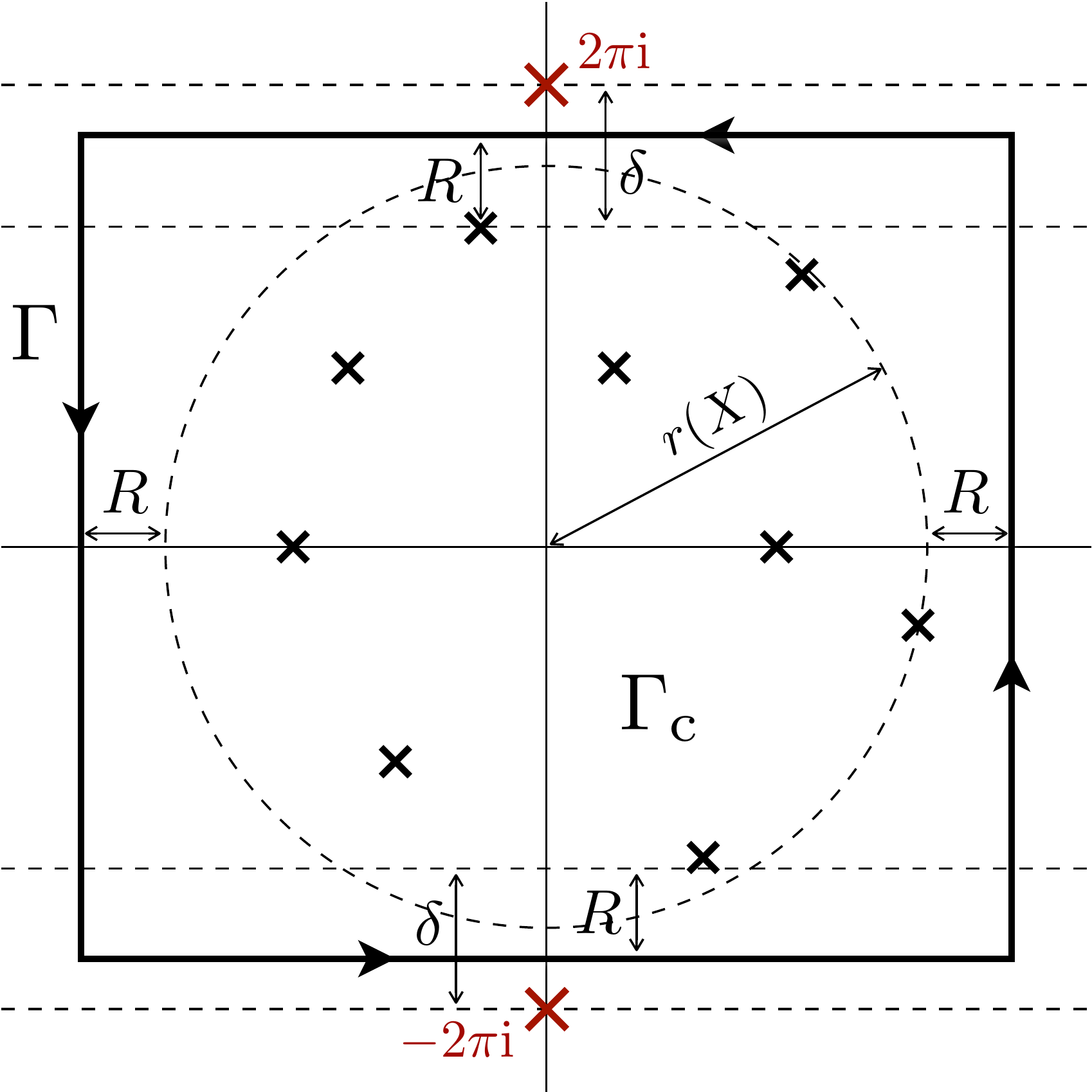}
\end{tabular}
\caption{(a) Contour $\Gamma$ to estimate a bound on $\|{\log(X+Y)-\log X}\|$. The crosses represent the eigenvalues $\{x_k\}$ of $X$, and $\mathrm{c}$ is the chosen  branch cut of $\log z$ (red thin straight line).
We take $R<\delta$ smaller than the gap $\delta$ between the spectrum of $X$ and the branch cut $\mathrm{c}$ (dashed curve running around the branch cut), and draw a contour $\Gamma$ (solid closed directed curve) which consists of a circle $\Gamma_R$ of radius $r(X)+R$, with $r(X)$ the spectral radius of $X$ (dashed circle), and a contour $\Gamma_\mathrm{c}$ going around the branch cut $\mathrm{c}$. (b) Contour $\Gamma$ to estimate a bound on $\|{g(X+Y)-g(X)}\|$ for $g(z)$ defined in Eq.\ (\ref{eq:g{z}def}). The black crosses represent the eigenvalues $\{x_k\}$ of $X$, while the red crosses are the two poles of $g(z)$ at $\pm2\pi\rmi$. For our purpose of bounding the BCH formula, the spectrum of $X$  is confined to a strip between the two poles, $\Im z\in(-2\pi +\delta ,2\pi - \delta)$, with a nonvanishing gap $\delta$ from the poles (dashed horizontal lines). We take $R<\delta$ and draw a rectangular contour $\Gamma$ with its horizontal lines running along $\Im z=\pm(2\pi-\delta+R)$ and its vertical lines along $\Re z=\pm[r(X)+R]$.}
\label{fig:ContLogCircleCut}
\end{figure}
In the case of $\log z$, we find it convenient to take care of its branch cut before applying Lemma~\ref{lem:MatrixFuncBound}.
We take the contour $\Gamma$ depicted in Fig.~\ref{fig:ContLogCircleCut}(a), which consists of a circle $\Gamma_R$ of radius $r(X)+R$, with $r(X)$ the spectral radius of $X$, and a contour $\Gamma_\mathrm{c}$ going around a branch cut $\mathrm{c}$.
Note that $R$ is bounded by the gap $\delta$ between the spectrum of $X$ 
and the branch cut $\mathrm{c}$, i.e.~$R<\delta= \operatorname{dist}(\mathrm{c},\operatorname{spec}(X))$. 
The integral along $\Gamma_\mathrm{c}$ 
simplifies, yielding
\begin{align}
\log(X+Y)-\log X
={}&
{-\sum_{n=1}^\infty}\rme^{\rmi\phi_\mathrm{c}}\int_0^{r(X)+R}\rmd y\,\frac{1}{y\rme^{\rmi\phi_\mathrm{c}}I-X}\left(Y\frac{1}{y\rme^{\rmi\phi_\mathrm{c}}I-X}\right)^n\nonumber\\
&{}+\sum_{n=1}^\infty\frac{1}{2\pi\rmi}\int_{\Gamma_R}\rmd z\log z\,\frac{1}{zI-X}\left(Y\frac{1}{zI-X}\right)^n
,
\label{eqn:LogDiff}
\end{align}
where $\phi_\mathrm{c}=\arg z$ along the branch cut $\mathrm{c}$. 
Each of the two contributions can be bounded as done in Lemma~\ref{lem:MatrixFuncBound}, and $\|{\log(X+Y)-\log X}\|$ is bounded by Eq.\ (\ref{eqn:BoundHPerturb}) with
\begin{equation}
M_{\log}(X,R)
=[r(X)+R]\left(
1+\sqrt{\log^2[r(X)+R]+\max[\phi_\mathrm{c}^2,(\phi_\mathrm{c}+2\pi)^2]}
\right)
\label{eqn:Mlog}
\end{equation}
in place of $M$.

On the other hand, for $g(z)$ defined in Eq.\ (\ref{eq:g{z}def}), let us take the rectangular contour $\Gamma$ 
shown in Fig.~\ref{fig:ContLogCircleCut}(b).
Note that the function $g(z)$ has poles at $z=2n\pi\rmi$ ($n\in\mathbb{Z}\setminus\{0\}$), while in bounding the BCH formula the spectrum of $X$ will be confined to the strip $\Im z \in(-2\pi,2\pi)$.
We take $R<\delta$, where $\delta=2\pi-\max_k|{\Im x_k}|$ is the gap between the spectral band and the poles at $\pm2\pi\rmi$, and let the horizontal lines of the rectangular contour $\Gamma$ run along $\Im z=\pm(2\pi-\delta+R)$ and the vertical lines along $\Re z=\pm[r(X)+R]$.
By using the following two bounds,
\begin{equation}
|g(x+\rmi y)|
=\sqrt{\frac{x^2+y^2}{1-2\rme^{-x}\cos y+\rme^{-2x}}}
\le
\begin{cases}
\medskip
\displaystyle
\frac{\sqrt{x^2+y^2}}{|1-\rme^{-x}|}
,\\
\displaystyle
\frac{\sqrt{x^2+y^2}}{\theta(\cos y)|{\sin y}|+\theta(-{\cos y})}
,
\end{cases}
\end{equation}
valid for any real $x$ and $y$, where $\theta(x)$ is the step function [$\theta(x)=1$ for $x>0$ and $0$ for $x<0$], 
we can easily  bound the integral in Eq.\ (\ref{eqn:Mlemma}) for $g(z)$ along the rectangular contour $\Gamma$ (the maximum of $|g(z)|$ multiplied by the path length along each of the four straight lines of the rectangular contour), 
and we get a bound (\ref{eqn:BoundHPerturb}) on $\|g(X+Y)-g(X)\|$ with
\begin{multline}
M_g(X,R)
=
\frac{2}{\pi}\sqrt{[r(X)+R]^2+(2\pi-\delta+R)^2}
\,\biggl(
\frac{2\pi-\delta+R}{2}\coth\frac{r(X)+R}{2}
\\
{}
+
\frac{r(X)+R}{\theta\bm{(}\cos(\delta-R)\bm{)}|{\sin(\delta-R)}|+\theta\bm{(}-{\cos(\delta-R)}\bm{)}}
\biggr)
\label{eqn:Mg}
\end{multline}
in place of $M$.

We can now proceed to bound the BCH formula.
To get a concise expression for the bound, we make use of a similarity transformation which brings $X$ into a Jordan normal form.
In the standard Jordan form, we put ``1''s next to the eigenvalues whose eigenspaces are not diagonalizable.
We point out that we can freely tune the similarity transformation so that the ``1''s are scaled to some positive constant $\nu$. 
Let us transform $X$ into such a Jordan form by a similarity transformation $T_\nu$,
\begin{equation}
\tilde{X}=T_\nu^{-1}XT_\nu=\sum_k(x_k\tilde{P}_k+\tilde{N}_k),
\label{eqn:TXT}
\end{equation} 
where $\{\tilde{P}_k\}$ are diagonal projections and $\{\tilde{N}_k\}$ are the nilpotents with entries $\nu$ or $0$ on the next diagonal.
Notice that the infinity norms (the largest singular values) of $\tilde{P}_k$ and $\tilde{N}_k$ are $\|\tilde{P}_k\|_\infty=1$ and $\|\tilde{N}_k\|_\infty=\nu$ or $0$, respectively.\footnote{The infinity norm of operator $X$ can be defined by $\|X\|_\infty=\sup_{\|v\|=1} \|X v\|$, through the Euclidean norm $\|v\|$ of vector $v$.}
This helps us simplify the bound.
The infinity norms before the similarity transformation are estimated as $\|P_k\|_\infty\le\|T_\nu^{-1}\|_\infty\|\tilde{P}_k\|_\infty\|T_\nu\|_\infty=\chi_\nu$ and $\|N_k\|_\infty\le\chi_\nu\nu$ or $0$, with $\chi_\nu=\|T_\nu^{-1}\|_\infty\|T_\nu\|_\infty\ge1$ called ``condition number'' of the similarity transformation $T_\nu$ \cite{ref:horn}.

We now present a bound on the BCH formula (\ref{eq:inverted-1-1}).
\begin{prop}[Bounding the BCH formula]\label{prop:BoundBCH}
Let $X$ and $Y$ be operators on a $D$-di\-men\-sion\-al Banach space, with the spectrum of $X$, $\operatorname{spec}(X)=\{x_k\}$,  confined within a strip $\Im z\in(\phi_\mathrm{c},\phi_\mathrm{c}+2\pi)$ for some $\phi_\mathrm{c}$.
We choose a primary logarithm such that $\log\rme^X=X$.
Then, for small enough $t\ge0$, the correction $W(t)$ in the BCH formula 
\begin{equation}
Z(t)
=\log(\rme^X\rme^{tY})
=
X
+
tg(\ad_X)(Y)
+W(t)
\end{equation}
is bounded by
\begin{equation}
\|W(t)\|_\infty
\le 
\frac{
(32M^2D^9\rme^{\alpha+\nu}/R^4)
t^2\chi_\nu^3\|Y\|_\infty^2\rme^{t\chi_\nu\|Y\|_\infty}
}{
1
-(1+8MD^4/R^2)
(2D^2\rme^{\alpha+\nu}/R)t\chi_\nu\|Y\|_\infty\rme^{t\chi_\nu\|Y\|_\infty}
},
\label{eqn:BoundW}
\end{equation}
where $g(z)$ is the function defined in Eq.\ (\ref{eq:g{z}def}), $R$ is a positive constant fulfilling 
\begin{equation}
R<\delta_1,\delta_2
\end{equation}
for the gaps
\begin{equation}
\delta_1=\min_k\min_{y\ge0}|\rme^{x_k}-y\rme^{\rmi\phi_\mathrm{c}}|,\qquad
\delta_2=2\pi-\max_{k,\ell}\Im(x_k-x_\ell),
\end{equation}
and
\begin{equation}
M=\max[M_{\log}(\rme^X,R),M_g(\ad_X,R)],
\end{equation}
with $M_{\log}$ and $M_g$ defined in Eqs.\ (\ref{eqn:Mlog}) and (\ref{eqn:Mg}), respectively.
$\chi_\nu$ is the condition number of the similarity transformation turning $X$ into a Jordan form with $\nu$ in the next diagonal in the nilpotents, and $\nu$ is tuned so that 
\begin{equation}
\nu\rme^{\alpha+\nu},2D\nu\le R/2,\qquad\alpha=\max_k\Re x_k
\label{eqn:CondNu}
\end{equation}
are satisfied.
$t$ should be small enough to satisfy 
\begin{equation}
(1+8MD^4/R^2)(2D^2\rme^{\alpha+\nu}/R)t\chi_\nu\|Y\|_\infty \rme^{t\chi_\nu\|Y\|_\infty}<1.
\end{equation}
\end{prop}
\begin{proof}
We start by preparing the spectral representations of $\rme^{\tilde{X}}$ and $\ad_{\tilde{X}}$ on the basis of the spectral representation of $\tilde{X}$ in Eq.\ (\ref{eqn:TXT}).

\paragraph{Spectral representation of $\bm{\rme^{\tilde{X}}}$:}
The spectral representation of $\rme^{\tilde{X}}$ can be constructed from the spectral representation of $\tilde{X}$ in Eq.\ (\ref{eqn:TXT}) as
\begin{equation}
\rme^{\tilde{X}}
=\sum_k\rme^{x_k}\sum_{q=0}^{n_k-1}\frac{1}{q!}\tilde{N}_k^q\tilde{P}_k
=\sum_k(\rme^{x_k}\tilde{P}_k+\tilde{N}_k^{(\rme)}),
\end{equation}
where the spectral projections of $\rme^{\tilde{X}}$ are identical with those of $\tilde{X}$ while the nilpotents $\tilde{N}_k^{(\rme)}$ of $\rme^{\tilde{X}}$ are given in terms of the nilpotents $\tilde{N}_k$ of $\tilde{X}$ by
\begin{equation}
\tilde{N}_k^{(\rme)}= \rme^{x_k}\sum_{q=1}^{n_k-1}\frac{1}{q!}\tilde{N}_k^q.
\end{equation}
Note that $\rme^{x_k}\neq\rme^{x_\ell}$ for $k\neq\ell$ thanks to the restriction $\Im x_k\in(\phi_\mathrm{c},\phi_\mathrm{c}+2\pi)$.
The nilpotents $\tilde{N}_k^{(\rme)}$ are bounded by
\begin{equation}
\|\tilde{N}_k^{(\rme)}\|_\infty
\le|\rme^{x_k}|\sum_{q=1}^{n_k-1}\frac{1}{q!}\nu^q
\le|\rme^{x_k}|(\rme^{\nu}-1)
\le|\rme^{x_k}|\nu\rme^\nu
\le\nu\rme^{\alpha+\nu},
\end{equation}
while the exponential $\rme^{\tilde{X}}$ itself is bounded by
\begin{equation}
\|\rme^{\tilde{X}}\|_\infty
\le
\sum_k|\rme^{x_k}|\sum_{q=0}^{n_k-1}\frac{1}{q!}\nu^q
\le D\rme^{\alpha+\nu}.
\end{equation}

\paragraph{Spectral representation of $\bm{\ad_{\tilde{X}}}$:}
We can also construct the spectral representation of $\ad_{\tilde{X}}=[\tilde{X},{}\bullet{}]$ from the spectral representation of $\tilde{X}$ in Eq.\ (\ref{eqn:TXT}) as
\begin{align}
\ad_{\tilde{X}}
&=
\sum_{k,\ell}(x_k\tilde{P}_k
+
\tilde{N}_k)
{}\bullet{}\tilde{P}_\ell
-
\sum_{k,\ell}\tilde{P}_k{}\bullet{}(x_\ell\tilde{P}_\ell
-
\tilde{N}_\ell
)
\nonumber\\
&=
\sum_{k,\ell}(x_k-x_\ell)\tilde{P}_k{}\bullet{}\tilde{P}_\ell
+
\sum_{k,\ell}
(
\tilde{N}_k{}\bullet{}\tilde{P}_\ell
-
\tilde{P}_k{}\bullet{}\tilde{N}_\ell
)
\nonumber\\
&=
\sum_m(\lambda_m\tilde{\mathcal{P}}_m+\tilde{\mathcal{N}}_m),
\end{align}
where the spectrum $\{\lambda_m\}$ of $\ad_{\tilde{X}}$ is given by the differences $\{x_k-x_\ell\}$ among the eigenvalues $\{x_k\}$ of $\tilde{X}$, and the spectral projections $\{\tilde{\mathcal{P}}_m\}$ and the nilpotents $\{\tilde{\mathcal{N}}_m\}$ are given by 
\begin{equation}
\tilde{\mathcal{P}}_m
=\sum_{k,\ell}\delta_{\lambda_m ,x_k-x_\ell}\tilde{P}_k{}\bullet{}\tilde{P}_\ell,\quad
\mathcal{N}_m
=\sum_{k,\ell}\delta_{\lambda_m ,x_k-x_\ell}
(
\tilde{N}_k{}\bullet{}\tilde{P}_\ell
-
\tilde{P}_k{}\bullet{}\tilde{N}_\ell
).
\end{equation}
Note that $\Im\lambda_m\in(-2\pi,2\pi)$ due to the restriction $\Im x_k\in(\phi_\mathrm{c},\phi_\mathrm{c}+2\pi)$.
In this proof, we use the norm induced by the infinity norm ($\infty$--$\infty$ norm) $
\|\mathcal{A}\|_{\infty-\infty}	
=\sup_{\|X\|_\infty=1}\|\mathcal{A}(X)\|_\infty
$
for superoperators.
In this norm, the spectral projections $\tilde{\mathcal{P}}_m$ and the nilpotents $\tilde{\mathcal{N}}_m$ are bounded by
\begin{equation}
\|\tilde{\mathcal{P}}_m\|_{\infty-\infty}
\le
\sum_{k,\ell}\delta_{\lambda_m ,x_k-x_\ell}
\le D,\quad
\|\tilde{\mathcal{N}}_m\|_{\infty-\infty}
\le
2\nu\sum_{k,\ell}\delta_{\lambda_m ,x_k-x_\ell}
\le 2D\nu.
\end{equation}
Since $\nu$ is tuned to satisfy the condition (\ref{eqn:CondNu}), we have
\begin{equation}
\|\tilde{N}_k^{(\rme)}\|_\infty/R\le1/2,
\qquad
\|\tilde{\mathcal{N}}_m\|_{\infty-\infty}/R\le1/2.
\label{eqn:NNBounds}
\end{equation}

\paragraph{Bounding $\bm{\|\tilde{Z}(t)-\tilde{X}\|}$:}
We are now ready to bound the correction $W(t)$.
For $\tilde{Y}=T_\nu^{-1}YT_\nu$, we have
\begin{equation}
\|\rme^{t\tilde{Y}}-I\|
\le
\rme^{t\|\tilde{Y}\|}-1
\le
t\|\tilde{Y}\|\rme^{t\|\tilde{Y}\|},
\qquad
\|\tilde{Y}\|_\infty\le\chi_\nu\|Y\|_\infty.	
\end{equation}
Then, for $\tilde{Z}(t)=\log(\rme^{\tilde{X}}\rme^{t\tilde{Y}})$, we use Lemma~\ref{lem:MatrixFuncBound} to bound
\begin{align}
\|\tilde{Z}(t)-\tilde{X}\|_\infty
&=\|{\log(\rme^{\tilde{X}}\rme^{t\tilde{Y}})-\tilde{X}}\|_\infty
\nonumber\\
&=\|{\log[\rme^{\tilde{X}}+\rme^{\tilde{X}}(\rme^{t\tilde{Y}}-I)]-\log\rme^{\tilde{X}}}\|_\infty
\nonumber\\
&\le
M\frac{
(2D/R)^2\|\rme^{\tilde{X}}(\rme^{t\tilde{Y}}-I)\|_\infty
}{
1-(2D/R)\|\rme^{\tilde{X}}(\rme^{t\tilde{Y}}-I)\|_\infty
}
\nonumber\\
&\le
(2MD/R)\frac{
(2D^2\rme^{\alpha+\nu}/R)t\|\tilde{Y}\|_\infty\rme^{t\|\tilde{Y}\|_\infty}
}{
1-(2D^2\rme^{\alpha+\nu}/R)t\|\tilde{Y}\|_\infty\rme^{t\|\tilde{Y}\|_\infty}
},
\end{align}
where we have bounded as $[1-(N/R)^D]/(1-N/R)\le2$ for $N=\max_k\|\tilde{N}_k^{(\rme)}\|_\infty$, under the condition in Eq.\ (\ref{eqn:NNBounds}).

\paragraph{Bounding $\bm{\|g(\ad_{\tilde{Z}(t)})-g(\ad_{\tilde{X}})\|}$:}
Next, by noting
\begin{equation}
\|{\ad_{\tilde{Z}(t)}-\ad_{\tilde{X}}}\|_{\infty-\infty}
\le
2\|\tilde{Z}(t)-\tilde{X}\|_\infty,
\end{equation}
we use Lemma~\ref{lem:MatrixFuncBound} to bound
\begin{align}
\|g(\ad_{\tilde{Z}(t)})-g(\ad_{\tilde{X}})\|_{\infty-\infty}
&\le 
M\frac{
(2D^3/R)^2\|{\ad_{\tilde{Z}(t)}-\ad_{\tilde{X}}}\|_{\infty-\infty}
}{
1-(2D^3/R)\|{\ad_{\tilde{Z}(t)}-\ad_{\tilde{X}}}\|_{\infty-\infty}
}
\nonumber\\
&\le 
(2MD^3/R)\frac{
(8MD^4/R^2)
\frac{
(2D^2\rme^{\alpha+\nu}/R)t\|\tilde{Y}\|_\infty\rme^{t\|\tilde{Y}\|_\infty}
}{
1-(2D^2\rme^{\alpha+\nu}/R)t\|\tilde{Y}\|_\infty\rme^{t\|\tilde{Y}\|_\infty}
}
}{
1
-(8MD^4/R^2)
\frac{
(2D^2\rme^{\alpha+\nu}/R)t\|\tilde{Y}\|_\infty\rme^{t\|\tilde{Y}\|_\infty}
}{
1-(2D^2\rme^{\alpha+\nu}/R)t\|\tilde{Y}\|_\infty\rme^{t\|\tilde{Y}\|_\infty}
}
}.
\end{align}
Therefore, 
\begin{align}
W(t)
&=\int_0^t\rmd s\,[g(\ad_{Z(s)})-g(\ad_X)](Y)
\nonumber\\
&=T_\nu^{-1}\left(
\int_0^t\rmd s\,[g(\ad_{\tilde{Z}(s)})-g(\ad_{\tilde{X}})](\tilde{Y})
\right)
T_\nu
\end{align}
is bounded by Eq.\ (\ref{eqn:BoundW}).
\end{proof}
\begin{remark}
The conditions (\ref{eqn:CondNu}) can be relaxed to $\nu\rme^{\alpha+\nu},2D\nu<R$, but to get a simpler expression for the bound we have strengthened the conditions to Eq.\ (\ref{eqn:CondNu}). 
\end{remark}
\begin{remark}
If $X$ is diagonalizable, we can set $\nu=0$ and replace $R/2$ with $R$ in the bound (\ref{eqn:BoundW}).
If $X$ can be diagonalized by a unitary transformation, we can further set $\chi_\nu=1$.
\end{remark}

\section{Proof of Lemma~\ref{prop:Peripheral}}
\label{sec:ProofPeripheral}
Here we prove one of our key lemmas, Lemma~\ref{prop:Peripheral}.
\begin{proof}[Proof of Lemma~\ref{prop:Peripheral}]
We split $E_n$ into two parts by the projection $P$ as
\begin{equation}
E_n=PE_nP+E_n',
\end{equation}
and consider the expansion
\begin{align}
E_n^n
={}&(PE_nP+E_n')^n
\nonumber
\\
={}&(PE_nP)^n
+\sum_{\ell=1}^{\lfloor\frac{n}{2}\rfloor}R_{n,2\ell}
+\sum_{\ell=1}^{\lfloor\frac{n}{2}\rfloor}R_{n,2\ell}'
+\sum_{\ell=1}^{\lfloor\frac{n-1}{2}\rfloor}R_{n,2\ell+1}
+\sum_{\ell=1}^{\lfloor\frac{n-1}{2}\rfloor}R_{n,2\ell+1}'
+E_n'^n
\label{eqn:Expansion}
\end{align}
with
\begin{align}
R_{n,2\ell}
&=
\mathop{\mathop{\sum}_{k, k'\in\mathbb{N}^\ell}}_{|k| + |k'|=n}
(PE_nP)^{k_1}E_n'^{k_1'}(PE_nP)^{k_2}E_n'^{k_2'}\cdots(PE_nP)^{k_\ell}E_n'^{k_\ell'},
\displaybreak[0]\\
R_{n,2\ell}'
&=
\mathop{\mathop{\sum}_{k, k'\in\mathbb{N}^\ell}}_{|k| + |k'|=n}
E_n'^{k_1'}(PE_nP)^{k_1}E_n'^{k_2'}(PE_nP)^{k_2}\cdots E_n'^{k_\ell'}(PE_nP)^{k_\ell},
\displaybreak[0]\\
R_{n,2\ell+1}
&=
\mathop{\mathop{\sum}_{k \in \mathbb{N}^{\ell +1},\, k'\in\mathbb{N}^\ell}}_{|k| + |k'|=n}
(PE_nP)^{k_1}E_n'^{k_1'}(PE_nP)^{k_2}E_n'^{k_2'}\cdots E_n'^{k_\ell'}(PE_nP)^{k_{\ell+1}},
\displaybreak[0]\\
R_{n,2\ell+1}'
&=
\mathop{\mathop{\sum}_{k \in \mathbb{N}^{\ell},\, k'\in\mathbb{N}^{\ell+1}}}_{|k| + |k'|=n}
E_n'^{k_1'}(PE_nP)^{k_1}E_n'^{k_2'}(PE_nP)^{k_2}\cdots(PE_nP)^{k_\ell}E_n'^{k_{\ell+1}'},
\end{align}
where $k = (k_1, \dots ,k_j)$ and $k' = (k_1', \dots ,k_j')$ are multi-indices of  integers with $j=\ell$ or $\ell+1$, $|k| = k_1 + \cdots + k_j$ denotes the order of $k$, and $\lfloor x\rfloor$ denotes the greatest integer less than or equal to $x$.
We are going to show that the corrections to the first contribution $(PE_nP)^n$ in Eq.\ (\ref{eqn:Expansion}) accumulate only up to $\mathcal{O}(1/n)$.

We first note that, since $PE_n'P=0$ by construction, the contributions with $k_i'=1$ for some $i$ between two $(PE_nP)$'s vanish and do not contribute to $R_m$ or $R_m'$.
Therefore, $k_i'\ge2$ for any $i$ sandwiched by $(PE_nP)$'s.
Second, $P E_n=E_n P +\mathcal{O}(1/n)$ [Eq.~\eqref{eqn:FamilyAsympMatrix} in condition 1] implies 
\begin{equation}
PE_nP = P E_n + \mathcal{O}(1/n) = E_n P + \mathcal{O}(1/n),
\end{equation} 
and hence, since $E'_n = E_n - P E_n P$,
\begin{equation}
P E_n' = \mathcal{O}(1/n), \qquad E_n' P = \mathcal{O}(1/n).
\label{eqn:Contact}
\end{equation}
Therefore, each contact between $PE_nP$ and $E_n'$ yields $\mathcal{O}(1/n)$, and we are advised to count the number of the contacts in the corrections to see the orders of the contributions: 
$R_{2\ell}$ and $R_{2\ell}'$ have $(2\ell-1)$ contacts, while $R_{2\ell+1}$ and $R_{2\ell+1}'$ have $2\ell$ contacts.
Third, since $\|PE_nP\|$ and $\|E_n'\|$ can be greater than $1$, it is not helpful to bound the corrections like $\|(PE_nP)^{k_1}E_n'^{k_1'}\cdots\|\le\|PE_nP\|^{k_1}\|E_n'\|^{k_1'}\cdots$.
On the other hand, we have the bounds (\ref{eqn:Bound1}) and (\ref{eqn:Bound2}) (conditions 2 and 3), which facilitate bounding the corrections.

Let us start estimating the corrections.
Thanks to Eq.\ (\ref{eqn:Bound2}) (condition 3), the last correction in Eq.\ (\ref{eqn:Expansion}) is bounded by
\begin{equation}
	\|E_n'^n\|
	\le K \mu^n ,
\end{equation}
which is $o(1/n)$ as $n\to+\infty$.

Let us set 
\begin{equation}
\label{eq:C_ndef}
C_n = \max\{ \|E_n' P \|, \| P E_n' \|\} = \mathcal{O}(1/n),
\end{equation}
where we have used Eq.~\eqref{eqn:Contact}. By the bounds (\ref{eqn:Bound1}) and (\ref{eqn:Bound2}) (conditions~3 and~4), the correction $R_{n,2\ell}$ is bounded as
\begin{align}
\|R_{n,2\ell}\|
&\le
\mathop{\sum_{k\in \mathbb{N}^\ell , \, k' \geq (2,\ldots,2,1) }}_{|k|+|k'|=n}
	\|(PE_nP)^{k_1}\|
	\| P E_n'\|
	\|E_n'^{k_1'-2}\|
	\|E_n' P\|
	\|(PE_nP)^{k_2}\|
	\cdots{}
\nonumber\displaybreak[0]\\[-8truemm]
&\qquad\qquad\qquad\qquad\qquad\qquad\qquad
	{}\cdots
	\|E_n'^{k_{\ell-1}'-2}\|
	\|E_n' P \|
	\|(PE_nP)^{k_\ell}\|
	\| P E_n'\|
	\|E_n'^{k_\ell'-1}\|
	\vphantom{\mathop{\sum\cdots\sum}_{k_1+\cdots+k_{m+1}=n}}
\nonumber\displaybreak[0]\\
&\le
	M^\ell K^\ell
	C_n^{2\ell - 1}
\mathop{\sum_{k\in \mathbb{N}^\ell , \, k' \geq (2,2,\dots,1)}}_{|k|+|k'|=n}
	\mu^{|k'| -2 \ell + 1}
	\nonumber\displaybreak[0]\\
&=
	M^\ell K^\ell
	C_n^{2\ell - 1}
 \mathop{\sum_{k'\in \mathbb{N}^\ell}} _{|k'|\le n-2\ell+1}
	\mu^{|k'| - \ell } \mathop{\sum_{k\in \mathbb{N}^\ell}}_{|k| = n - |k'| - \ell + 1} 1
	\nonumber\displaybreak[0]\\
&=
	M^\ell K^\ell
	C_n^{2\ell - 1}
\mathop{\sum_{k'\in \mathbb{N}^\ell}} _{|k'|\le n-2\ell+1}
\begin{pmatrix}
\medskip
	n-|k'|-\ell\\
	\ell-1
\end{pmatrix}
	\mu^{|k'|- \ell}
	\nonumber\displaybreak[0]\\
&\le
	M^\ell K^\ell
	C_n^{2\ell - 1}
\frac{n^{\ell-1}}{(\ell-1)!}\mathop{\sum_{k'\in \mathbb{N}^\ell}} _{|k'|\le n-2\ell+1}
	\mu^{|k'|- \ell}
	\nonumber\displaybreak[0]\\
	&\le
	M^\ell K^\ell
	C_n^{2\ell - 1}
\frac{n^{\ell-1}}{(\ell-1)!}
\frac{1}{(1-\mu)^\ell}
	\nonumber\displaybreak[0]\\
&
\equiv a_{n,\ell},
\end{align}
where we have used the identities
\begin{equation}
\mathop{\sum_{k\in \mathbb{N}^\ell}}_{|k| = m} 1 = 
\begin{pmatrix}
	m -1 \\
	\ell-1
\end{pmatrix} ,
\qquad
\sum_{k\in \mathbb{N}^\ell} \mu^{|k|-\ell} = \biggl(\sum_{j\geq 0} \mu^j\biggr)^\ell = \frac{1}{(1-\mu)^\ell},
\end{equation}
and  the inequality
\begin{equation}
\begin{pmatrix}
	n-m \\
	j
\end{pmatrix}
\le\frac{n^{j}}{j!}.
\end{equation}

The correction $R_{n,2\ell}'$ is bounded as
\begin{align}
\|R_{n,2\ell}'\|
&\le
\mathop{\sum_{k\in \mathbb{N}^\ell , \, k' \geq (1,2,\dots,2) }}_{|k|+|k'|=n}
	\|E_n'^{k_1'-1}\|
	\|E_n' P \|
	\|(PE_nP)^{k_1}\|
	\|P E_n'\|
	\|E_n'^{k_2'-2}\|
	\cdots{}
\nonumber
\\[-8truemm]
&\qquad\qquad\qquad\qquad\qquad\qquad
	{}\cdots
	\|(PE_nP)^{k_{\ell-1}}\|
	\| P E_n'\|
	\|E_n'^{k_\ell'-2}\|
	\|E_n' P \|
	\|(PE_nP)^{k_\ell}\|
	\vphantom{\mathop{\sum\cdots\sum}_{k_1+\cdots+k_{m+1}=n}}
\nonumber\displaybreak[0]\\
&\le
	M^\ell K^\ell
	C_n^{2\ell - 1}
\mathop{\sum_{k\in \mathbb{N}^\ell , \, k' \geq (1,2,\dots,2)}}_{|k|+|k'|=n}
	\mu^{|k'| -2 \ell + 1}
	\nonumber\displaybreak[0]\\
&=
	M^\ell K^\ell
	C_n^{2\ell - 1}
 \mathop{\sum_{k'\in \mathbb{N}^\ell}} _{|k'|\le n-2\ell+1}
	\mu^{|k'| - \ell } \mathop{\sum_{k\in \mathbb{N}^\ell}}_{|k| = n - |k'| - \ell + 1} 1
	\nonumber\displaybreak[0]\\
	&
\le a_{n,\ell},
\end{align}
 as for $\|R_{n,2\ell}\|$.

The correction $R_{n,2\ell+1}$ is bounded as
\begin{align}
\|R_{n,2\ell+1}\|
&\le
\mathop{\sum_{k\in\mathbb{N}^{\ell +1},\, k' \ge (2,\ldots,2)}}_{|k|+|k'|=n}
	\|(PE_nP)^{k_1}\|
	\|P E_n'\|
	\|E_n'^{k_1'-2}\|
	\|E_n' P\|
	\|(PE_nP)^{k_2}\|
	\cdots{}
\nonumber\displaybreak[0]\\[-8truemm]
&\qquad\qquad\qquad\qquad\qquad\quad
	{}\cdots
	\|(PE_nP)^{k_\ell}\|
	\| P E_n'\|
	\|E_n'^{k_\ell'-2}\|
	\|E_n' P \|
	\|(PE_nP)^{k_{\ell+1}}\|
	\vphantom{\mathop{\sum\cdots\sum}_{k_1+\cdots+k_{m+1}=n}}
\nonumber\displaybreak[0]\\
&\le
	M^{\ell+1}K^\ell
	C_n^{2\ell}
\mathop{\sum_{k\in\mathbb{N}^{\ell +1},\, k' \ge (2,\dots,2)}}_{|k|+|k'|=n}
	\mu^{|k'|- 2\ell}
\nonumber\displaybreak[0]\\
&=
	M^{\ell+1} K^\ell
	C_n^{2\ell}
 	\mathop{\sum_{k'\in \mathbb{N}^\ell}}_{|k'|\le n-2\ell-1}
	\mu^{|k'| - \ell } \mathop{\sum_{k\in \mathbb{N}^{\ell+1}}}_{|k| = n - |k'| - \ell} 1
	\nonumber\displaybreak[0]\\
&=
	M^{\ell+1}K^\ell
	C_n^{2\ell}
	\mathop{\sum_{k'\in \mathbb{N}^\ell}}_{|k'|\le n-2\ell-1} 
	\begin{pmatrix}
	\medskip
	n-|k'|-\ell-1\\
	\ell
	\end{pmatrix}
	\mu^{|k'| - \ell }
	\nonumber\displaybreak[0]\\
&\le
	M^{\ell+1}K^\ell
	C_n^{2 \ell}
	\frac{n^\ell}{\ell!}
	\frac{1}{(1-\mu)^\ell}
	\nonumber\displaybreak[0]\\
	&
\equiv b_{n,\ell}.
\end{align}

Finally, the correction $R_{n,2\ell+1}'$ is bounded as
\begin{align}
\|R_{n,2\ell+1}'\|
&\le
\mathop{\sum_{k\in\mathbb{N}^\ell, \, k' \geq (1,2,\dots,2,1)}}_{|k| + |k'| = n}
	\|E_n'^{k_1'-1}\|
	\|E_n' P\|
	\|(PE_nP)^{k_1}\|
	\| P E_n'\|
	\|E_n'^{k_2'-2}\|
	\cdots
\nonumber
\\[-8truemm]
&\qquad\qquad\qquad\qquad\qquad\qquad\quad
	{}\cdots
	\|E_n'^{k_\ell'-2}\|
	\|E_n' P \|
	\|(PE_nP)^{k_\ell}\|
	\| P E_n'\|
	\|E_n'^{k_{\ell+1}'-1}\|
	\vphantom{\mathop{\sum\cdots\sum}_{k_1+\cdots+k_{m+1}=n}}
\nonumber\displaybreak[0]\\
&\le
	M^{\ell} K^{\ell+1}
	C_n^{2\ell}
	\mathop{\sum_{k\in\mathbb{N}^\ell, \, k' \geq (1,2,\dots,2,1)}}_{|k| + |k'| = n}	
	\mu^{|k'|- 2 \ell }
\nonumber\displaybreak[0]\\
&=
	M^{\ell} K^{\ell+1}
	C_n^{2\ell}
 	\sum_{k'\in \mathbb{N}^{\ell+1}} 
	\mu^{|k'| - \ell  -1 } \mathop{\sum_{k\in \mathbb{N}^\ell}}_{|k| = n - |k'| - \ell +1 } 1
	\nonumber\displaybreak[0]\\
&=
	M^\ell K^{\ell+1}
	C_n^{2\ell}
	\sum_{k' \in \mathbb{N}^{\ell +1}}
	\begin{pmatrix}
	\medskip
	n-|k'|-\ell\\
	\ell-1
	\end{pmatrix}
	\mu^{|k'| - \ell -1 }
\nonumber\displaybreak[0]\\
&\le
	M^\ell K^{\ell+1}
	C_n^{2\ell}
	\frac{n^{\ell-1}}{(\ell-1)!}
	\frac{1}{(1-\mu)^{\ell+1}}
\nonumber\displaybreak[0]\\
&
\equiv c_{n,\ell}.
\end{align}

Collecting all these bounds, the correction to the leading contribution $(PE_nP)^n$ in Eq.\ (\ref{eqn:Expansion}) is bounded by
\begin{align}
&\Biggl\|
\sum_{\ell=1}^{\lfloor\frac{n}{2}\rfloor}R_{n,2\ell}
+\sum_{\ell=1}^{\lfloor\frac{n}{2}\rfloor}R_{n,2\ell}'
+\sum_{\ell=1}^{\lfloor\frac{n-1}{2}\rfloor}R_{n,2\ell+1}
+\sum_{\ell=1}^{\lfloor\frac{n-1}{2}\rfloor}R_{n,2\ell+1}'
+E_n'^n
\Biggr\|
\nonumber\displaybreak[0]\\
&\qquad
\le
	2\sum_{\ell=1}^\infty a_{n,\ell}
	+\sum_{\ell=1}^\infty b_{n,\ell}
	+\sum_{\ell=1}^\infty c_{n,\ell}
	+K \mu^n
\nonumber\displaybreak[0]\\
&\qquad
=
M\left[
\left(
1+
\frac{K}{1-\mu} C_n
	\right)^2
\rme^{
	\frac{MK}{1-\mu} n C_n^2
	}-1
	\right]
	+K \mu^{n}
\nonumber\\
&\qquad
=\mathcal{O}(1/n),\vphantom{\Bigl(}
\label{eqn:TotalCorrection}
\end{align}
since $C_n=\mathcal{O}(1/n)$ by Eq.~\eqref{eq:C_ndef}.
This proves the lemma.
\end{proof}


\end{document}